\documentclass[a4paper,reqno]{amsart}
\usepackage[utf8]{inputenc}
\usepackage[T1]{fontenc}
\usepackage{amsfonts}
\usepackage[foot]{amsaddr}
\usepackage{amsmath}
\usepackage{todonotes}
\numberwithin{equation}{section}
\usepackage{xcolor}
\DeclareUnicodeCharacter{00A0}{~}
\usepackage[hidelinks]{hyperref}
\usepackage[margin=1.2in]{geometry}
\usepackage{enumerate}
\usepackage{graphicx}
\usepackage{caption}
\usepackage{subcaption}

\usepackage{array}

\usepackage{MnSymbol}

\usepackage{tikz}
\usepackage{tikz-cd}
\tikzcdset{column sep/normal=1.1cm}
\usetikzlibrary{calc}
\usepackage{color, mathtools,epsfig, graphicx}
\usetikzlibrary{positioning}
\usetikzlibrary{positioning, shapes, shadows, arrows}

\newcommand{\Z}{\mathbb{Z}}
\newcommand{\C}{\mathbb{C}}
\newcommand{\p}{\mathbb{P}}
\newcommand{\E}{\mathcal{E}}
\newcommand{\F}{\mathcal{F}}
\newcommand{\h}{\mathcal{H}}
\newcommand{\pain}[1]{\operatorname{P}_{\mathrm{#1}}}
\newcommand{\defeq}{\vcentcolon=}

\newcommand{\Pic}{\operatorname{Pic}}
\newcommand{\nod}{\operatorname{nod}}

\usepackage{amsthm}

\newtheorem{theorem}{Theorem}[section]

\newtheorem{lemma}[theorem]{Lemma}

\newtheorem{definition}[theorem]{Definition}
\newtheorem{proposition}[theorem]{Proposition}

\newtheorem{remark}[theorem]{Remark}
\newtheorem{corollary}[theorem]{Corollary}

\numberwithin{equation}{section}

\title[On real and imaginary roots of generalised Okamoto polynomials]{On real and imaginary roots of generalised Okamoto polynomials}

\author{Pieter Roffelsen$^{1}$}

\address{$^{1}$School of Mathematics and Statistics F07, The University of Sydney, NSW 2006, Australia.}

\author{Alexander Stokes$^{2,3}$}

\address{$^{2}$Graduate School of Mathematical Sciences, The University of Tokyo, 3-8-1 Komaba, Meguro-ku, Tokyo 153--8914, Japan.}
\address{$^{3}$Faculty of Mathematics, Informatics and Mechanics, University of Warsaw, ul. Banacha 2, 02-097, Warsaw, Poland.}

\date{\today}

\begin{document}

\maketitle
\begin{abstract}
Recently, B. Yang and J. Yang derived a family of rational solutions to the Sasa-Satsuma equation, and showed that any of its members constitutes a partial-rogue wave provided that an associated generalised Okamoto polynomial has no real roots or no imaginary roots.
In this paper, we derive exact formulas for the number of real and the number of imaginary roots of the generalised Okamoto polynomials. On the one hand, this yields a list of partial-rogue waves that satisfy the Sasa-Satsuma equation. On the other hand, it gives families of rational solutions of the fourth Painlev\'e equation that are pole-free on either the real line or the imaginary line. 
To obtain these formulas, we develop an algorithmic procedure to derive the qualitative distribution of singularities on the real line for real solutions of Painlev\'e equations, starting from the known distribution for a seed solution, through the action of B\"acklund transformations on the rational surfaces forming their spaces of initial conditions.

\end{abstract}


\section{Introduction}


The fourth Painlev\'e equation ($\pain{IV}$) constitutes one of the six classical families of second-order ordinary differential equations whose solutions are free of movable branch points, derived by Fuchs, Gambier, Painlev\'e and Picard around the turn of the twentieth century.
It governs monodromy preserving deformations of rank two linear systems with one Fuchsian singularity and one irregular singularity of Poincar\'e rank two, which has allowed for the derivation of various asymptotic properties of both the general solution \cite{kapaev1996, kapaev1998,MCB97,vereshch}
and special solutions \cite{buckinghamP4, buckmiller, itskapaev, davidepieter, davidepieterhermite}.


On the other hand, very few exact results on singularities of solutions of $\pain{IV}$
are available in the literature. 
In this regard, we mention the fact that any transcendental solution has infinitely many zeros and poles \cite{gromakbook,joshiradnovicP4}, and that, for generic real parameter values, the order in which they appear for real solutions is restricted \cite{twiton1,twiton2}.

In this paper, we determine the exact numbers of real zeroes and poles, as well as the order in which they appear on the real line, for rational solutions of $\pain{IV}$ in the $-\tfrac{2}{3} t$ hierarchy. These are the rational solutions that can be expressed in terms of generalised Okamoto polynomials \cite{noumiyamada}.


Our main motivation is the recent discovery, by B. Yang and J. Yang \cite{yangyang}, that these numbers predicate the existence of partial-rogue wave solutions to the Sasa-Satsuma equation, which we describe next.

The Sasa-Satsuma equation is an integrable extension of the nonlinear Schr\"odinger equation \cite{kodamahasegawa,sasasatsuma}, with applications in nonlinear optics \cite{optics1,optics2, optics3, optics4}, which we will consider in the form
\begin{equation*}
u_t = u_{xxx} + 6 | u|^2 u_x + 3 u \left( |u|^2\right)_x,
\end{equation*}
where $u(x,t)\in\mathbb{C}$, $(x,t)\in\mathbb{R}^2$. 
B. Yang and J. Yang \cite{yangyang} derived two families of rational solutions to this equation 
on a constant-amplitude background,
\begin{equation}\label{yangyangsols}
u_{M,N}^{\pm}(x,t) = h_{M,N}^{\pm}(x,t) u_{\operatorname{bg}}(x,t),\qquad u_{\operatorname{bg}}(x,t):=e^{i[\alpha(x+6t)-\alpha^3 t]},
\end{equation}
where $\alpha=\tfrac{1}{2}$ and $h_{M,N}^{\pm}$ is a rational function of $x$ and $t$, for $M,N \in \Z_{\geq 0}$. Explicit formulas for the rational functions $h_{M,N}^{\pm}(x,t)$ are given in Appendix \ref{app:rationalsolutions}.

Some of these solutions constitute \emph{partial-rogue waves},
which Yang and Yang describe as 
``localised waves that `come from nowhere but leave with a trace' ''. With regards to the families \eqref{yangyangsols}, they are the members $u_{M,N}^{\pm}(x,t)$ that asymptotically reduce to their constant-amplitude background $u_{\operatorname{bg}}(x,t)$ as $t\rightarrow -\infty$ and as $|x|\rightarrow \infty$ for $t$ bounded from above, but not as $t\rightarrow +\infty$. 

Yang and Yang showed that whether $u_{M,N}^{\pm}$ is a partial-rogue wave or not is predicated on the non-existence of real or imaginary zeros of the Wronskian
\begin{equation}\label{eq:wronskianrep}
    Q^{[YY]}_{M,N} = \mathcal{W}\left[p_2, p_5, \dots, p_{3M-1}, p_1, p_4, \dots, p_{3N-2} \right], \quad M,N \in \mathbb{Z}_{\geq 0},
\end{equation} 
where the entries $p_j(z)$, $j\geq 0$,  are Schur polynomials, generated by
\begin{equation*}
    \operatorname{exp}(z \epsilon+ \epsilon^2) = \sum_{j=0}^{\infty} p_j(z) \epsilon^j.
\end{equation*}


The relevant results of \cite{yangyang} are as follows.

\begin{theorem}[Yang and Yang \cite{yangyang}]
Let $M$ and $N$ be non-negative integers.
\begin{enumerate}[(a)]
\item If $Q^{[YY]}_{M,N}$ has real but not imaginary roots,
 the solution $u_{M,N}^{+}(x,t)$ is a partial-rogue wave. 
In addition, when $t\gg1$, it splits into $\rho_{\operatorname{re}}(M,N)$ fundamental rational solitons, where $\rho_{\operatorname{re}}(M,N)$ is the number of real zeroes of $Q^{[YY]}_{M,N}$.
\item If $Q^{[YY]}_{M,N}$ has imaginary but not real roots, then the solution $u_{M,N}^{-}(x,t)$ is a partial-rogue wave. 
In addition, when $t\gg1$, it splits into $\rho_{\operatorname{im}}(M,N)$ fundamental rational solitons, where $\rho_{\operatorname{im}}(M,N)$ is the number of imaginary roots of $Q^{[YY]}_{M,N}$.
\end{enumerate}
\end{theorem}
The above theorem immediately leads to the question of how many real and imaginary roots the polynomials $Q^{[YY]}_{M,N}$ have.  These polynomials form a subset of the generalised Okamoto polynomials and the main results of this paper yield  the following corollary.



\begin{corollary} \label{cor:numberofrootsYY}
    The numbers of real and imaginary roots of the polynomials $Q^{[YY]}_{M,N}$, $M,N \in \Z_{\geq0}$, relevant to the rational solutions \eqref{yangyangsols} of the Sasa-Satsuma equation are 
    \begin{equation*}
        \rho_{\operatorname{re}}(M,N) = 
        \begin{cases}
        N & \text{ if } M \text{ even and } M \leq N-1, \\
        M+1 & \text{ if } M \text{ odd and } M \leq N-1, \\
        N & \text{ if } N \text{ even and } N \leq M, \\
        M+1 & \text{ if } N \text{ odd and } N \leq M, 
        \end{cases} 
    \end{equation*}
    and
    \begin{equation*}
       \rho_{\operatorname{im}}(M,N) = 
        \begin{cases}
        N-M & \text{ if } N-M \text{ even and } M \leq N-1, \\
        N & \text{ if } N-M \text{ odd and } M \leq N-1, \\
        M-N & \text{ if } M-N \text{ even and } N \leq M, \\
        M+1 & \text{ if } M-N \text{ odd and } N \leq M. 
	\end{cases}
	\end{equation*}
    In particular the solution $u^{-}_{M,N}(x,t)$ 
    constitutes a partial-rogue wave if and only if $N=0$,
    and when $t\gg 1$ it splits into $M$ (resp. $M+1$) fundamental rational solitons if $M$ is even (resp. odd).
Similarly the solution $u^{+}_{M,N}(x,t)$ 
   constitutes a partial-rogue wave if and only if $M=N$, 
    and when $t\gg 1$ it splits into $M$ (resp. $M+1$) fundamental rational solitons if $M$ is even (resp. odd).
\end{corollary}

A second application of our results can be found in the construction of Hamiltonians constrained to full-fill a third-order shape-invariance condition \cite{HMZ22}, which we expand on further in Remark \ref{rem:marquette} in the main text.


The distribution of singularities of solutions of Painlev\'e equations is a subject that goes back to the seminal work of Boutroux \cite{boutrouxP1, boutrouxP12} on solutions of the first Painlev\'e equation that are asymptotically free of poles in certain sectors of the complex plane.
For particular solutions, for example the tritronqu\'ee solutions discovered by Boutroux, the locations of singularities have been studied through various approaches \cite{adalitanveer, bertolaAS,CHT, deanoP1,davidetritronquee1, davidetritronquee2, joshikitaev}.
The possible sequences in which singularities of real solutions of Painlev\'e equations can occur on the real line has been studied in the case of $\pain{IV}$ by Schiff and Twiton \cite{twiton1,twiton2} and in the case of the sixth Painlev\'e equation $\pain{VI}$ by Eremenko and Gabrielov \cite{eremenkogabrielov}, who with Hinkkanen also described solutions of $\pain{VI}$ which do not take singular values anywhere in the complex plane \cite{EGH17}.

The distributions of singularities of rational and algebraic solutions of Painlev\'e equations in particular has attracted a lot of interest in recent times. 
Thse singularities are related to the roots of associated polynomials, the locations of which appear to form highly regular patterns in the complex plane (see e.g. the survey \cite{clarksonsurvey} for plots of these).
The most detailed descriptions of these distributions so far have been asymptotic ones derived in various large-parameter limits, often via  isomonodromy and Riemann-Hilbert techniques \cite{bothnermiller, bothnermillersheng,buckinghamP4, buckinghammillerP21,buckinghammillerP22, buckmiller,buckinghammillerP3,davidepieter, davidepieterhermite}.

In the case of $\pain{IV}$ there are two hierarchies of rational special solutions, expressed in terms of generalised Okamoto and generalised Hermite polynomials respectively.
The generalised Okamoto polynomials were introduced by Noumi and Yamada \cite{noumiyamada} and include those found by Okamoto \cite{okamotostudies} (see also \cite{umemura,umemurawatanabe,fukutani,clarksonpolynomials}).
They provide polynomial $\tau$-functions for $\pain{IV}$ at the corresponding special parameter values and have representations as Hankel determinants \cite{kajiwaraohta,JKM}.


In this paper we will study the distribution of singularities of solutions of $\pain{IV}$ on the real line using the \emph{initial value space} or \emph{space of initial conditions}, introduced by Okamoto \cite{OKAMOTO1979} and extended by Sakai \cite{SAKAI2001}. 
This space has been used to study singularities and asymptotics of solutions of Painlev\'e equations, in particular of generic solutions of $\pain{IV}$ in \cite{joshiradnovicP4}.
We will develop an inductive procedure to derive the sequence in which singularities occur on the real line for real solutions related by B\"acklund transformations, starting from the known sequence for a seed solution.

%

\subsection{Outline of the paper}
The manuscript is structured as follows. In Section \ref{sec:prelim} we recall basic facts about the fourth Painlev\'e equation and singularities of its solutions, as well as the hierarchy of rational solutions expressed in terms of the generalised Okamoto polynomials. 
In Section \ref{sec:mainresults} we state the main results of the paper on the qualitative distributions of singularities of the generalised Okamoto rational solutions on the real line as well as the real roots of the generalised Okamoto polynomials.
In Section \ref{sec:tools} we introduce the tools based on the Okamoto-Sakai theory of rational surfaces associated with Painlev\'e equations which we will be used in the proofs of the main results to which Sections \ref{sec:region1}--\ref{sec:regions456} are devoted.
In Section \ref{sec:conclusions} we give some concluding remarks.
The Appendix \ref{app:rationalsolutions} contains explicit expressions for the rational solutions of the Sasa-Satsuma equation in question.

\subsection*{Acknowledgements}
AS was supported by a Japan Society for the Promotion of Science (JSPS) Postdoctoral Fellowship for Research in Japan and also acknowledges the support of JSPS KAKENHI Grant Numbers 21F21775 and 22KF0073. PR was supported by Australian Research Council Discovery Project \#DP210100129.

We would like to thank Jianke Yang for bringing our attention to the connection between the numbers of real and imaginary roots of generalised Okamoto polynomials and the partial-rogue waves in the Sasa-Satsuma equation.
We would also like to thank Ian Marquette and Yang Shi for interesting discussions on the topic. 
Some of the ideas of this paper have their origins in work done by AS as part of an Australian Mathematical Sciences Institute (AMSI) Vacation Research Scholarship 2015-2016 supervised by Nalini Joshi, Nobutaka Nakazono and Yang Shi, while an undergraduate student at The University of Sydney, and AS would like to thank his supervisors as well as AMSI.

\section{Preliminaries} \label{sec:prelim}
We consider the fourth Painlev\'e equation $\pain{IV}$ in the form
\begin{equation}\label{eq:piv}
	q''=\frac{(q')^2}{2q}+ \frac{3}{2} q^3 + 4\, t\,q^2+2\left(t^2+a_2-a_0\right)q-\frac{2a_1^2}{q}, \qquad '=\frac{d}{dt},
\end{equation}
with parameters $a=(a_0,a_1,a_2)$ satisfying the single constraint
\begin{equation*}
    a_0+a_1+a_2=1.
\end{equation*}
We will only be concerned with real parameter values, $a\in\mathbb{R}^3$, with $a_1\neq 0$.

It will be helpful to rewrite the ordinary differential equation (ODE) \eqref{eq:piv} as a first order system. 
To this end, we set
\begin{equation} \label{fgtoq}
f=q,\quad g=t+\frac{1}{2} q+\frac{a_1}{q}+\frac{q'}{2 q},
\end{equation}
which yields the system of ODEs
    \begin{equation} \label{eq:systfg}
        \begin{aligned}
    f'&=-2a_1-f(2t+f-2g),\\
    g'&=+2a_2+g(2t+2f-g),
        \end{aligned}
    \end{equation}
which is the Hamiltonian form of $\pain{IV}$ due to Okamoto \cite{okamotohamiltonians}.

\subsection{Apparent singularities} \label{subsec:apparentsingularities}
Regarding the variable $q(t)$ as taking values in the Riemann sphere, 
the fourth Painlev\'e equation \eqref{eq:piv} is singular whenever $q(t)\in\{0,\infty\}$. 
These singularities are apparent, since they are resolved by a finite number of blowups in extending the system to the initial value space, the details of which will be given in Section \ref{sec:tools}.

\begin{definition}[plus or minus zeros and poles]
    Let $q(t)$ be a solution of $\pain{IV}$ with parameters $a$ and suppose $q(t_*)=\infty$ for some $t_*\in\mathbb{C}$. 
    Then $q(t)$ has a simple pole at $t=t_*$,
\begin{equation*}
    q(t)=\frac{\epsilon}{t-t_*}+ \mathcal{O}(1)\qquad (t\to t_*),
\end{equation*}
with residue $\epsilon=+1$ or $\epsilon=-1$. 
    If $\epsilon = +1$ we call $t=t_*$ a \emph{plus pole} of $q(t)$ and if $\epsilon=-1$ we call $t=t_*$ a \emph{minus pole} of $q(t)$.
    
    Similarly if $q(t_*)=0$ for some $t_*\in\mathbb{C}$, then $q(t)$ has a simple zero at $t=t_*$,
\begin{equation*}
    q(t)=2 \epsilon a_1 (t-t_*)+\mathcal{O}\left((t-t_*)^2\right)\qquad (t\to t_*),
\end{equation*}
with $\epsilon = +1$ or $\epsilon=-1$. 
If $\epsilon = +1$ we call $t=t_*$ a \emph{plus zero} of $q(t)$ and if $\epsilon=-1$ we call  $t=t_*$ a \emph{minus zero} of $q(t)$.
\end{definition}

\begin{remark}
As noted before, we will not consider the case $a_1=0$ in this paper, in which case zeros of $q(t)$ have multiplicity greater than one.
\end{remark}
\begin{lemma} \label{lem:singularities_expansions}
In terms of the variables $f$, $g$ introduced above, apparent singularities correspond to the following Laurent series expansions.
\begin{description}
    \item[$\bullet$ $p_+$ (\emph{plus pole})] If a solution $q(t)$ of \eqref{eq:piv} has a plus pole at $t=t_*$, then the corresponding solution to the system \eqref{eq:systfg} is given by 
    \begin{equation*}
        \begin{aligned}
            f(t) &= \frac{+1}{t-t_*} - t_* + \frac{1}{3}\left(t_*^2 - 2a_1 - 4a_2 - 2\right) (t-t_*) + \eta (t-t_*)^2 + \mathcal{O}\left((t-t_*)^3\right),\\
            g(t) &= - 2 a_2 (t-t_*) + \left(2 \eta - t_* - 2 a_2 t_*\right)(t-t_*)^2 
            + \mathcal{O}\left((t-t_*)^3\right),
            \end{aligned}
    \end{equation*}
    where $\eta\in \mathbb{C}$, and all coefficients in the Laurent series are determined by $\eta$ and $t_*$.
    
    \item[$\bullet$ $p_-$ (\emph{minus pole})] If a solution $q(t)$ of \eqref{eq:piv} has a minus pole at $t=t_*$, then the corresponding solution to the system \eqref{eq:systfg} is given by 
    \begin{equation*}
        \begin{aligned}
            f(t) &= \frac{-1}{t-t_*} - t_* + \frac{1}{3}\left(2a_1 + 4a_2  - 6- t_*^2\right) (t-t_*) + \eta (t-t_*)^2 + \mathcal{O}\left((t-t_*)^3\right),\\
            g(t) &= \frac{-1}{t-t_*} + t_* + \frac{1}{3}\left(6 -4 a_1 -2 a_2-t_*^2\right) (t-t_*) + \mathcal{O}\left((t-t_*)^2\right),
        \end{aligned}
    \end{equation*}
    where $\eta\in \mathbb{C}$, and all coefficients in the Laurent series are determined by $\eta$ and $t_*$.
    
    \item[$\bullet$ $z_+$ (\emph{plus zero})] If a solution $q(t)$ of \eqref{eq:piv} has a plus zero at $t=t_*$, then the corresponding solution to the system \eqref{eq:systfg} is given by 
    \begin{equation*}
        \begin{aligned}
            f(t) &=+2 a_1 (t-t_*) + \eta (t-t_*)^2 + \frac{4}{3} a_1 \left(a_1+2a_2 -1 + t_*^2\right)(t-t_*)^3 + \mathcal{O}\left((t-t_*)^4\right),\\
            g(t) &= \frac{+1}{t-t_*} + t_* + \frac{1}{3}\left(2  +4 a_1 +2 a_2+t_*^2\right) (t-t_*) + \mathcal{O}\left((t-t_*)^2\right),
        \end{aligned}
    \end{equation*}
    where $\eta\in \mathbb{C}$, and all coefficients in the Laurent series are determined by $\eta$ and $t_*$.
    
    \item[$\bullet$ $z_-$ (\emph{minus zero})] If a solution $q(t)$ of \eqref{eq:piv} has a minus zero at $t=t_*$, then the corresponding solution to the system \eqref{eq:systfg} is given by 
    \begin{equation*}
        \begin{aligned}
            f(t) &= - 2 a_1 (t-t_*) + \eta (t-t_*)^2 +  \frac{4}{3} a_1 \left(a_1+2a_2 -1 + t_*^2\right)(t-t_*)^3 + \mathcal{O}\left((t-t_*)^4\right),\\
            g(t) &= t_* - \frac{\eta}{2 a_1} + \frac{1}{4 a_1^2}\left(4 a_1^2 t_*^2 +8 a_1^2a_2 + \eta^2\right) (t-t_*) + \mathcal{O}\left((t-t_*)^2\right),
        \end{aligned}
    \end{equation*}
    where $\eta\in \mathbb{C}$, and all coefficients in the Laurent series are determined by $\eta$ and $t_*$.    
\end{description}
\end{lemma}
\begin{remark}
The arbitrary constants $\eta$ are sometimes called resonant parameters as they correspond to coefficients that can be chosen freely when computing the expansion of solutions about a movable singularity, due to resonances occurring in the recursive relations for these coefficients. 
They correspond to accessory parameters in the context of the associated linear problem.
We will see in Section \ref{sec:tools} that they determine where the solution crosses a certain exceptional curve on the space of initial conditions at $t=t_*$.
\end{remark}

\subsection{Translation symmetries and the Okamoto rational solution hierarchy}

It is well-known that the fourth Painlev\'e equation admits an extended affine Weyl group of B\"acklund transformation symmetries, which relate solutions for different values of the parameters.
Among these are transformations corresponding to translation elements of the group, which are used to generate hierarchies of special solutions.

These translations form a subgroup generated by $T_1$ and $T_2$, whose explicit expressions are given in Section \ref{sec:tools}, see equations \eqref{BTT1} and \eqref{BTT2}, which relate solutions for different parameter values as follows:
If $q$ is a solution of $\pain{IV}$ (or equivalently $(f,g)$ is a solution of system \eqref{eq:systfg}) with parameters $a=(a_0,a_1,a_2)$, then applying the transformation $T_1$, respectively $T_2$, yields a solution for parameters
\begin{equation*}
T_1 a = (a_0 - 1, a_1+1,a_2),
\end{equation*}
respectively
\begin{equation*}
T_2 a = (a_0 - 1, a_1,a_2+1).
\end{equation*}
By recursively applying $T_1$ and $T_2$ to some seed solution, one may generate a hierarchy of solutions.

The hierarchy of rational solutions we will consider is generated from the seed solution
\begin{equation*}
    q_{0,0}=-\tfrac{2}{3}t,\quad a_0=a_1=a_2=\tfrac{1}{3}.
\end{equation*}
Applying $T_1^n T_2^m$ yields the solution 
\begin{equation} \label{qrat}
    q_{m,n}=-\tfrac{2}{3}t+\frac{d}{dt}\log{\frac{Q_{m-1,n}}{Q_{m,n}}},
\end{equation}
    for parameters 
\begin{equation*}
        a_0=\tfrac{1}{3}-m-n,\quad a_1=\tfrac{1}{3}+n,\quad a_2=\tfrac{1}{3}+m, \quad m,n \in \mathbb{Z},
\end{equation*}
where $Q_{m,n}$ are the \emph{generalised Okamoto polynomials} introduced in \cite{noumiyamada}, which in special cases recover the Okamoto polynomials $Q_{n}$ and $R_{n}$ from \cite{okamotostudies}.

The generalised Okamoto polynomials are defined by the recursive formulas
\begin{equation*} 
    \begin{aligned}
    Q_{m,n-1} Q_{m,n+1} &= \frac{9}{2} \left( Q_{m,n}'' Q_{m,n} - (Q_{m,n}')^2 \right) + \left( 2 t^2 + 3 (m+2n+1)\right)Q_{m,n}^2, \\
    Q_{m+1,n}Q_{m-1,n} &= \frac{9}{2} \left( Q_{m,n}'' Q_{m,n} - (Q_{m,n}')^2 \right) + \left( 2 t^2 + 3 (-2m+n-1)\right)Q_{m,n}^2,
    \end{aligned}
\end{equation*}
with initial conditions
\begin{equation*}
    Q_{0,0}=Q_{-1,0}=Q_{0,-1}=1,\quad Q_{-1,-1}=\sqrt{2}t.
\end{equation*}

We refer to the rational solutions $q_{m,n}$ as the \emph{generalised Okamoto rationals}. 
We have collected explicit expressions for the polynomials $Q_{m,n}$ for small values of the indices in Table \ref{table:genOkamotopolynomials}.

\renewcommand{\arraystretch}{1.5}
\begin{table}[htb]
\begin{tabular}{|>{$}c<{$}|| >{$}c<{$} | >{$}c<{$} | >{$}c<{$} | >{$}c<{$} |} 
 \hline
~   & m=-2 & m=-1 & m=0 & m=1 \\
 \hline
 \hline
n=2 & 4 t^4-12 t^2-9 & 4 t^4+12 t^2-9 & 8 t^6+60 t^4+90 t^2+135 & 
\begin{array}{c}
    32 t^{10}+240 t^8-720 t^6 \\
    - 5400 t^4-12150 t^2+6075 
\end{array}
\\
\hline
n=1 & 2 t^2-3 & \sqrt{2} t & 2 t^2+3 & \sqrt{2} t \left(4 t^4-45\right) \\
\hline
n=0 & 2 t^2+3 & 1 & 1 & 2 t^2-3 \\
\hline
n=-1 & 4 t^4+12 t^2-9 & \sqrt{2} t & 1 & \sqrt{2} t \\
\hline
n=-2 & 16 t^8-504 t^4-567 & 4 t^4-12 t^2-9 & 2 t^2-3 & 2 t^2 +3  \\
\hline
\end{tabular}
\caption{Generalised Okamoto polynomials $Q_{m,n}(t)$ for $-2\leq m\leq 1$ and $-2\leq n\leq 2$.}
\label{table:genOkamotopolynomials}
\end{table}

\begin{remark} \label{remark:facts}
The following facts can be found, for example, in  \cite{noumiyamada}.
\begin{itemize}
    \item All roots of $Q_{m,n}$ are simple.
    \item The plus and minus poles of $q_{m,n}$ are respectively given by roots of $Q_{m-1,n}$ and $Q_{m,n}$.
    \item The plus and minus zeroes of $q_{m,n}$ are respectively given by roots of $Q_{m,n-1}$ and $Q_{m-1,n+1}$.
    \item The polynomials satisfy the symmetry 
    \begin{equation} \label{realimagsymmetry}
    Q_{m,n}(i t) = i^{\operatorname{deg}Q_{m,n}} Q_{n,m}(t), 
    \end{equation}
    where the degree of the polynomial $Q_{m,n}$ is 
    \begin{equation*}
       \operatorname{deg} Q_{m,n} = m^2+m\hspace{0.3mm} n+n^2+m+n.
    \end{equation*}
    \item The generalised Okamoto rationals are odd functions, 
    \begin{equation}\label{eq:okamotorationaloddness}
        q_{m,n}(-t)=-q_{m,n}(t).
    \end{equation}
\end{itemize}
\end{remark}

\begin{remark} 
There are various different conventions for the indexing and normalisation of the generalised Okamoto polynomials in the literature.
We collect the relations between $Q_{m,n}(t)$ and some of these other forms below.
\begin{align*}
    &~ &~ Q_{m,n}(t)&=Q_{-m,-n}^{[C]}(t), &&(\text{Clarkson \cite{clarksonpolynomials}}), &&~ \\
     &~ &~ &=Q_{-m,-n}^{[HMZ]}(t),          &&(\text{Hussin, Marquette, Zelaya \cite{HMZ22}} ), &&~ \\
     &~ &~ &=Q_{-m,-n}^{[MR]}(t),          &&(\text{Masoero, Roffelsen \cite{davidepieter}} ), &&~ \\
     &~ &~ &=c_{m,n} Q_{-m-n,-n}^{[NY]}(\sqrt{\tfrac{2}{3}}t),       &&(\text{Noumi, Yamada \cite{noumiyamada}}), &&~
\end{align*}
where $c_{m,n}$ are some immaterial real multipliers.
The original Okamoto polynomials $Q_n$ and $R_n$ from \cite{okamotostudies} are related to these by
    \begin{equation*}
        Q_n(x)= Q^{[NY]}_{n,0}(x) = c_{-n,0}^{-1}Q_{-n,0}(\sqrt{\tfrac{3}{2}}x), \qquad R_n(x)=Q^{[NY]}_{n+1,1}(x) = c_{-n,-1}^{-1}Q_{-n,-1}(\sqrt{\tfrac{3}{2}}x). 
    \end{equation*}    
The polynomials in \cite{yangyang} form a family $Q_{M,N}^{[YY]}(z)$, with $M,N \geq 0$, related to ours by 
\begin{equation*}
\begin{aligned}
    Q^{[YY]}_{M,N}(z) &= \gamma_{M,N} Q_{M-N,-M-1}(\tfrac{\sqrt{3}}{2}z), &&\quad\text{where } M, N\geq 0,\\
    Q_{m,n}(t) &= \tilde{\gamma}_{m,n} Q^{[YY]}_{-n-1,-m-n-1}( \tfrac{2}{\sqrt{3}} t), &&\quad\text{where } n \leq -1, m+ n \leq -1, 
\end{aligned}
\end{equation*}
where $\gamma_{M,N}$ and $\tilde{\gamma}_{m,n}$ are similarly immaterial real multipliers.

\end{remark}

\subsection{Singularity signatures of solutions on the real line}

We are interested in the qualitative distributions of apparent singularities of real solutions on the real line. We capture such distributions as strings of symbols from the alphabet $\left\{ z_+, z_-, p_+, p_- \right\}$, as made precise in the following definition.
\begin{definition}[singularity signatures of real solutions]
    Let $q=q(t)$ be a real solution of the fourth Painlev\'e equation with only finitely many apparent singularities on the real line, say at $t=t_1,\dots,t_N$ with $t_1<t_2<\dots<t_N$. 
    Then the \emph{singularity signature} of $q$ on the real line is a finite string of symbols
    \begin{equation*}
        \mathfrak{S}(q) \defeq s_{1}\,s_{2}\,\dots \,s_{N-1}\,s_{N}, \qquad s_i \in \left\{ z_+, z_-, p_+, p_- \right\},
    \end{equation*}
    with $N\geq 0$, such that $q$ has an apparent singularity of type $s_i$ at $t=t_i$, for $1\leq i\leq N$.
    For solutions $q$ with infinitely many apparent singularities on the real line, $\mathfrak{S}(q)$ is defined similarly as a left-infinite, right-infinite or doubly infinite string.
    We will use standard notation to denote concatenations of strings, for example $(s_1 \,s_2)^2 = s_1\, s_2\,s_1\,s_2$.
    Furthermore, we sometimes mark a symbol by a hat, as in $\hat{s}$, to indicate that the corresponding apparent singularity occurs at the origin.
\end{definition}

In Figure \ref{fig:okamotoratplots} we give some plots of the generalised Okamoto rationals $q_{m,n}$ for particular values of $m,n$, indicating the apparent singularities as follows: zeroes are indicated by dots, with $z_+$ and $z_-$ in red and black respectively, while poles are indicated by crosses and the corresponding asymptotes by dashed vertical lines, with $p_+$ and $p_-$ in red and black respectively.
From these plots it can be seen that the singularity signature of these solutions on the real line are as follows:
\begin{equation*}
    \begin{aligned}
        \mathfrak{S}(q_{3,3}) &= 
        (p_-\,z_+\,z_-\,p_+)^1 (p_-\,z_+)^2\,\hat{p}_-\,(z_+\,p_-)^2(p_+\,z_-\,z_+\,p_-)^1, \\
        \mathfrak{S}(q_{3,4}) &=  
        (p_-\,z_+\,z_-\,p_+)^1 (p_-\,z_+)^2\,p_-\,\hat{z}_+\,p_-\,(z_+\,p_-)^2(p_+\,z_-\,z_+\,p_-)^1, \\
        \mathfrak{S}(q_{4,3}) &=  
        (p_-\,z_+\,z_-\,p_+)^2 (z_-\,p_+)^1\,z_-\,\hat{p}_+\,z_-\,(p_+\,z_-)^1(p_+\,z_-\,z_+\,p_-)^2, \\
        \mathfrak{S}(q_{4,4}) &=  
        (p_-\,z_+\,z_-\,p_+)^2 (z_-\,p_+)^2\,\hat{z}_-\,(p_+\,z_-)^2(p_+\,z_-\,z_+\,p_-)^2.
    \end{aligned}
\end{equation*}

\begin{figure}[ht]
    \centering
     \begin{subfigure}[b]{0.48\textwidth}
         \centering
         \includegraphics[width=\textwidth]{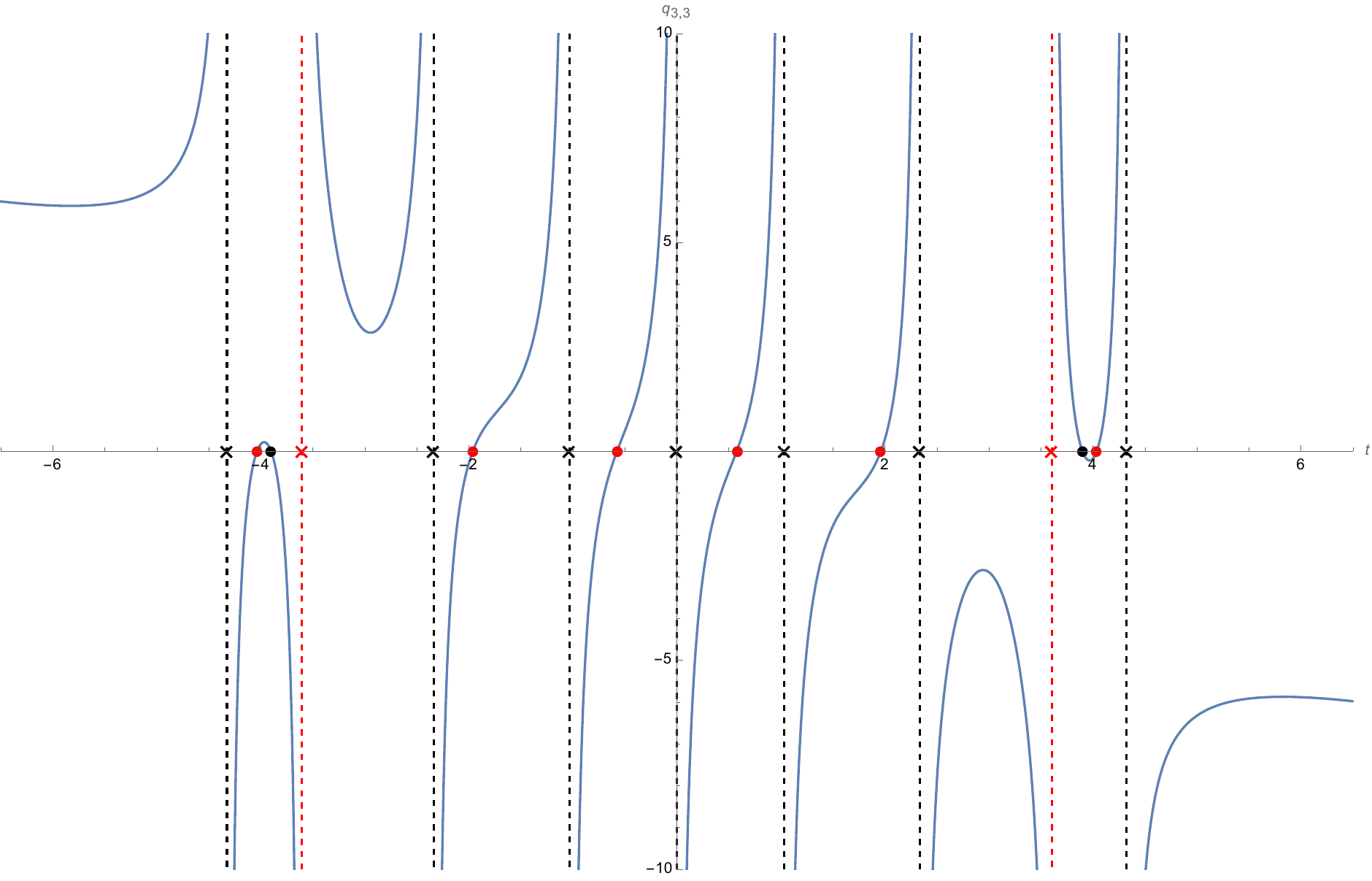}
         \caption{$m=3$, $n=3$}
         \label{fig:okamotorat33}
     \end{subfigure}
    \begin{subfigure}[b]{0.48\textwidth}
         \centering
         \includegraphics[width=\textwidth]{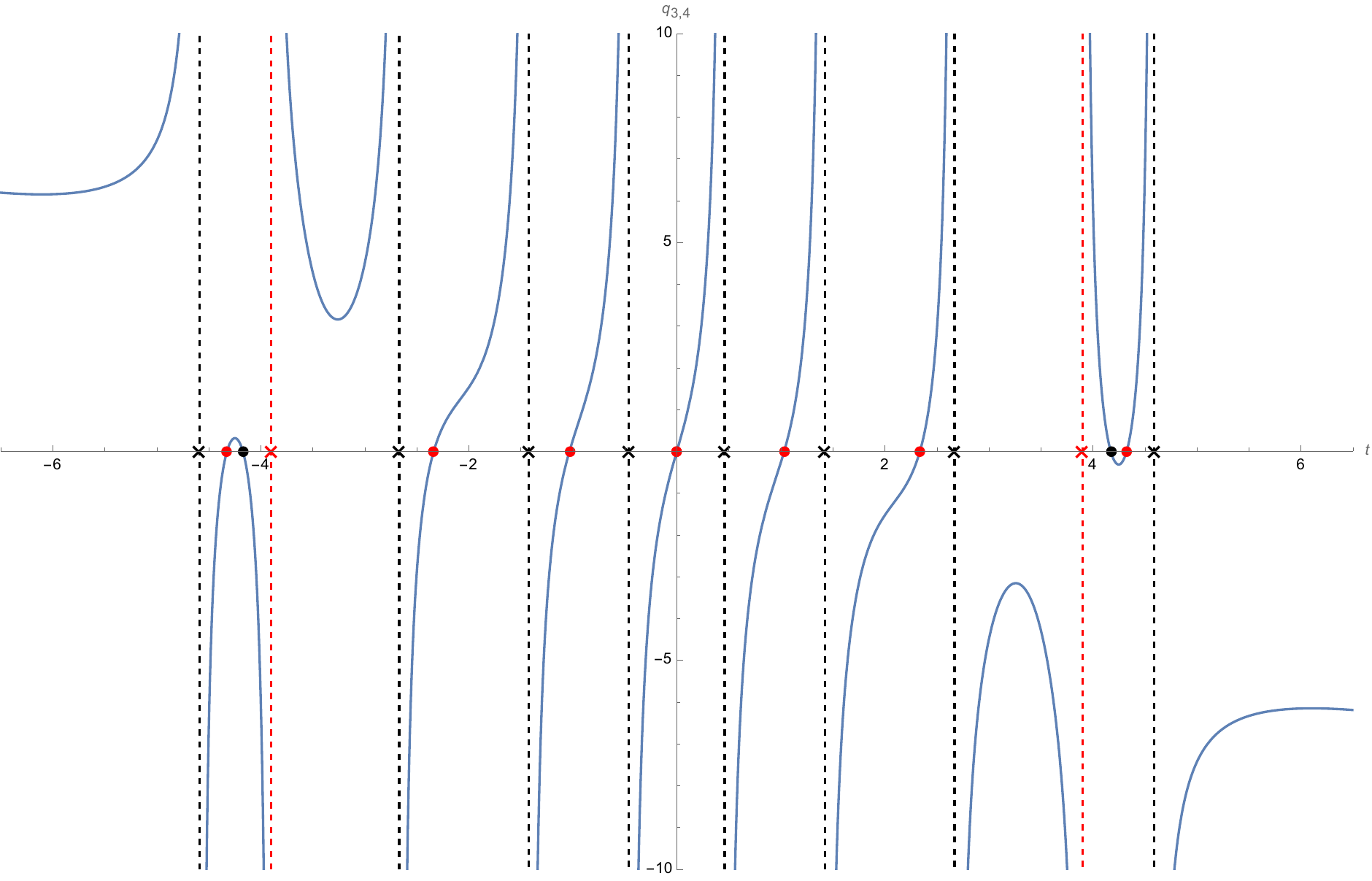}
         \caption{$m=3$, $n=4$}
         \label{fig:okamotorat34}
     \end{subfigure}
    \begin{subfigure}[b]{0.48\textwidth}
         \centering
         \includegraphics[width=\textwidth]{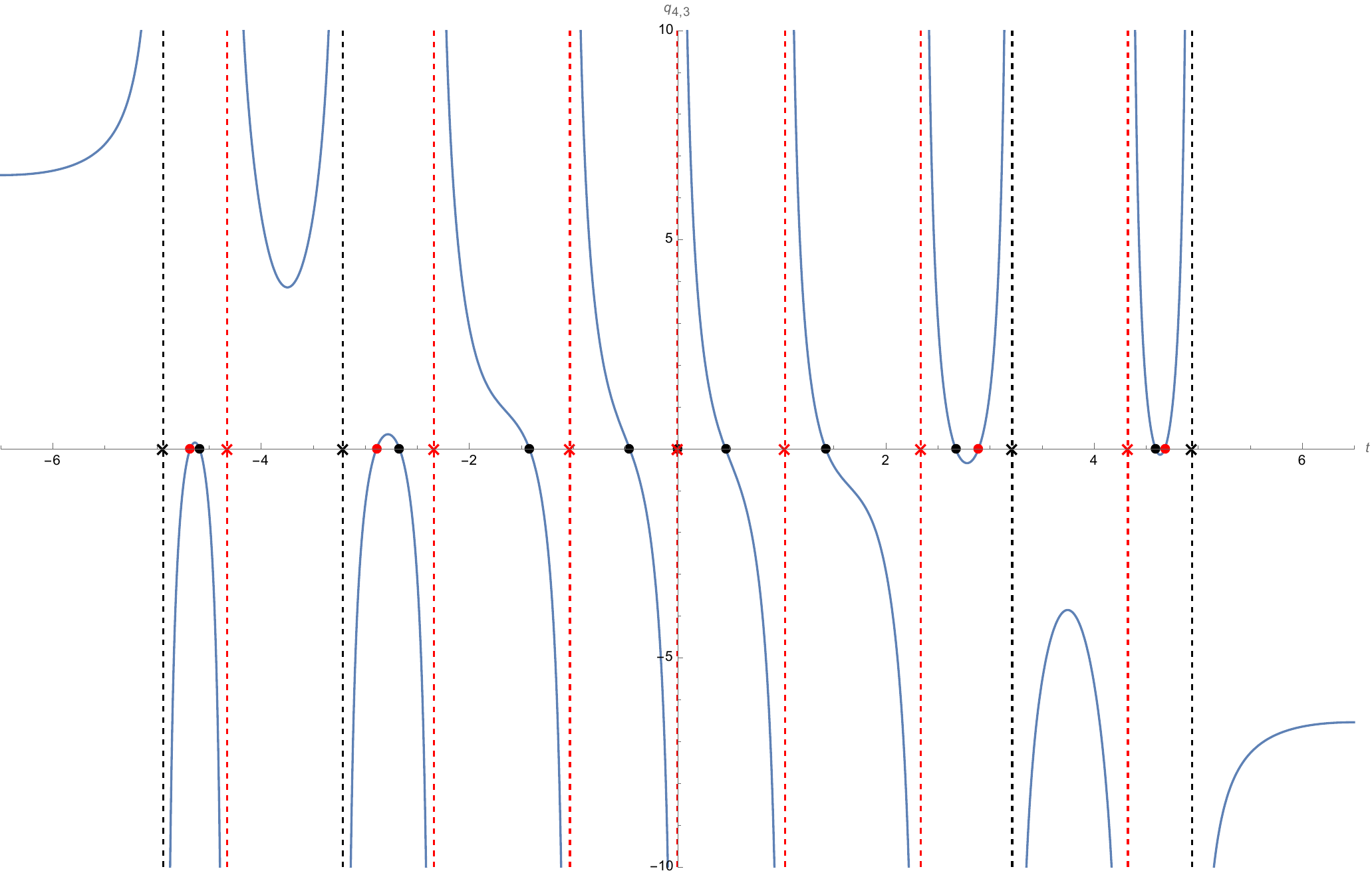}
         \caption{$m=4$, $n=3$}
         \label{fig:okamotorat43}
      \end{subfigure}
   \begin{subfigure}[b]{0.48\textwidth}
         \centering
         \includegraphics[width=\textwidth]{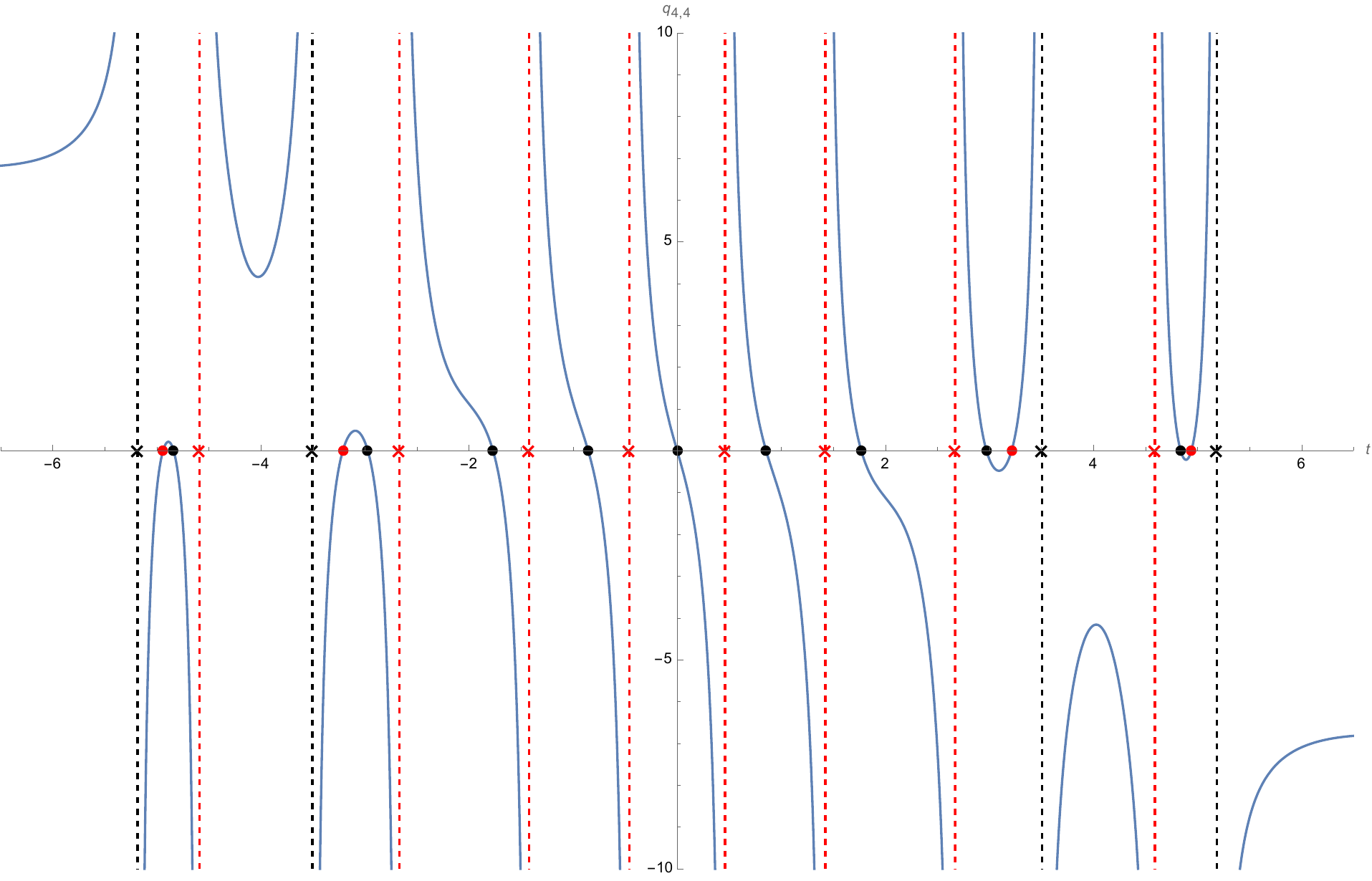}
         \caption{$m=4$, $n=4$}
         \label{fig:okamotorat44}
      \end{subfigure}
   \caption{Plots of some generalised Okamoto rationals $q_{m,n}(t)$}
    \label{fig:okamotoratplots}
\end{figure}

\section{Main results}  \label{sec:mainresults}

We present an algorithmic procedure to inductively describe singularity signatures of families of real solutions generated by translations.
These translations act as isomorphisms between complex algebraic surfaces forming Okamoto's space of initial values for the fourth Painlev\'e equation, introduced in Section \ref{sec:tools}.
Rational solutions parametrise paths in Okamoto's space with apparent singularities corresponding to crossings of certain exceptional curves on the surfaces.
Restricting to the real part of Okamoto's space allows us to construct topological arguments for the order in which exceptional curves must be crossed by a solution path, based on an inductive hypothesis of the singularity signature of a solution related by a translation. 



The inductive argument is broken up into six subcases defined by different regions in the $m,n$ plane, see Figure \ref{fig:regions}.
The reason for this is that, between the regions, exceptional curves relevant to the arguments undergo topological changes.
These six regions are bounded by the lines $a_0=0$, $a_1=0$ and $a_2=0$ in parameter space, where nodal curves appear on the surface.

\begin{figure}[htb]
    \centering
 \begin{tikzpicture}[scale=.5]
    \foreach \x in {-7,-6,...,7} {
        \foreach \y in {-7,-6,...,7} {
            \fill[color=black] (\x,\y) circle (0.075);
        }
    }
    \draw [black] 	(0,-8) 	-- (0,8) node [above] {$n$};
    \draw [black] 	(-8,0) 	-- (8,0) node [right] {$m$};
    \draw [red, dashed]   (-8,-1/3) -- (8,-1/3) node [pos=0,left] {$a_1=0$};
    \draw [red, dashed]   (-1/3,-8) -- (-1/3,8) node [pos=0,below] {$a_2=0$};
    \draw [red, dashed]   (-8+1/3,25/3-1/3) -- (8-1/3,-23/3+1/3) node [pos=1,right] {$a_0=0$};
    \draw [draw=red!50, fill=red!50, opacity=0.2, rounded corners] 
    (-1/4,-1/4) -- (7+1/4,-1/4) -- (7+1/4,7+1/4) --(-1/4,7+1/4) -- (-1/4,-1/4) ;
    \node at (3.5,3.5) [font=\fontsize{15}{0}\selectfont] {$\mathrm{I}$};
    \draw [draw=blue!50, fill=blue!50, opacity=0.2, rounded corners] 
    (1+1/3,-1/3-1/4) -- (7+1/4,-1/3-1/4) -- (7+1/4,-7+1/4) --(1+1/3,-1+1/4) ;
    \node at (5.5,-2.5) [font=\fontsize{15}{0}\selectfont] {$\mathrm{VI}$};
    \draw [draw=red!50, fill=red!50, opacity=0.2, rounded corners] 
    (-1/4,+2/4) -- (7+2/4,-7-1/4) -- (-1/4,-7-1/4) --(-1/4,+1/4) ;
    \node at (2.5,-5.5) [font=\fontsize{15}{0}\selectfont] {$\mathrm{V}$};
    \draw [draw=blue!50, fill=blue!50, opacity=0.2, rounded corners] 
    (-1/3-1/4,-1/3-1/4) -- (-7-1/4,-1/3-1/4) -- (-7-1/4,-7-1/4) --(-1/3-1/4,-7-1/4) -- (-1/3-1/4,-1/3-1/4) ;
    \node at (-3.5,-3.5) [font=\fontsize{15}{0}\selectfont] {$\mathrm{IV}$};
    \draw [draw=red!50, fill=red!50, opacity=0.2, rounded corners] 
    (+2/4,-1/4) -- (-7-1/4,-1/4) -- (-7-1/4,7+2/4) --(+2/4,-1/4) ;
    \node at (-5.5,2.5) [font=\fontsize{15}{0}\selectfont] {$\mathrm{III}$};
    \draw [draw=blue!50, fill=blue!50, opacity=0.2, rounded corners] 
    (-1/3-1/4,1+1/3) -- (-1/3-1/4,7+1/4) -- (-7+1/4,7+1/4) --(-1+1/4,1+1/3) ;
    \node at (-2.5,5.5) [font=\fontsize{15}{0}\selectfont] {$\mathrm{II}$};
    \end{tikzpicture}    
    \caption{Regions in the $(m,n)$-plane}
    \label{fig:regions}
\end{figure}
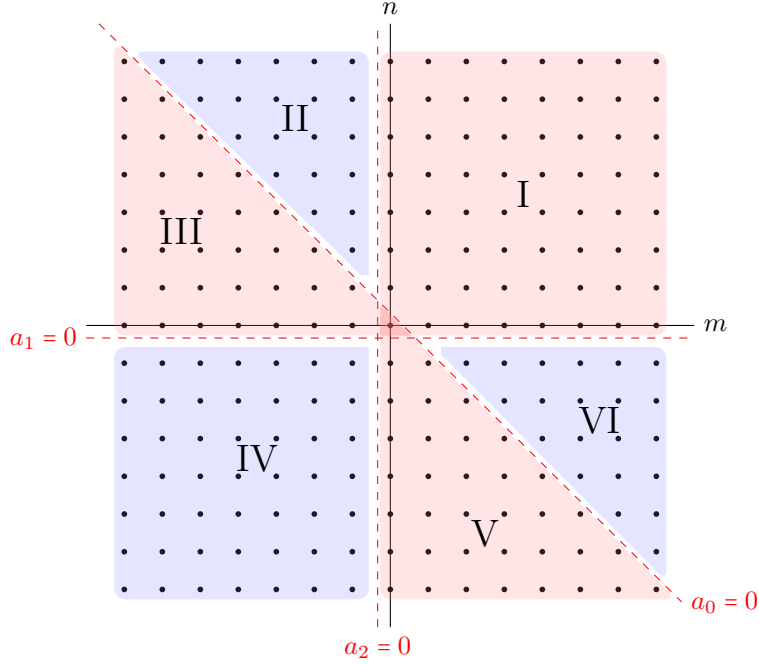

We will now state the main result within each of the six regions.
We start with the following theorem concerning region $\mathrm{I}$, which is proved in Section \ref{sec:region1}.

\begin{theorem}[Region $\mathrm{I}$]  
\label{th:region1}
For $(m,n)\in\mathbb{Z}^2$ with $m \geq 0$, $n\geq 0$, the singularity signature of the generalised Okamoto rational $q_{m,n}(t)$ is as follows:

\begin{description}
    \item[$m = 2 \mu$ even, $n = 2\nu$ even] 
        \begin{flalign*}
           ~&\qquad\qquad\qquad   \mathfrak{S}(q_{m,n}) = (p_-\,z_+\,z_-\,p_+)^{\mu}\,(z_-\,p_+)^{\nu}\,\hat{z}_-\,(p_+\,z_-)^{\nu}\,(p_+\,z_-\,z_+,\,p_-)^{\mu}. &
        \end{flalign*}        
    \item[$m=2\mu$ even, $n=2\nu+1$ odd] 
        \begin{flalign*}
           ~&\qquad\qquad\qquad \mathfrak{S}(q_{m,n}) =  (p_-\,z_+\,z_-\,p_+)^{\mu}\,(z_-\,p_+)^{\nu}\,z_-\,\hat{p}_+\, z_-\,(p_+\,z_-)^{\nu}\,(p_+\,z_-\,z_+\,p_-)^{\mu}. &
        \end{flalign*}
    \item[$m = 2\mu+1$ odd, $n=2\nu$ even] 
        \begin{flalign*}
           ~&\qquad\qquad\qquad   \mathfrak{S}(q_{m,n}) = (p_-\,z_+\,z_-\,p_+)^{\mu} (p_-\,z_+)^{\nu}\,p_- \,\hat{z}_+\,p_-(z_+\,p_-)^{\nu}(p_+\,z_-\,z_+\,p_-)^{\mu}. &
        \end{flalign*}
    \item[$m=2\mu+1$ odd, $n=2\nu+1$ odd] 
        \begin{flalign*}
           ~&\qquad\qquad\qquad \mathfrak{S}(q_{m,n}) = (p_-\,z_+\,z_-\,p_+)^{\mu}(p_-\,z_+)^{\nu+1}\,\hat{p}_-\,(z_+\,p_-)^{\nu+1}(p_+\,z_-\,z_+\,p_-)^{\mu}. &
    \end{flalign*}    
\end{description}
\end{theorem}


The following theorem, concerning region  $\mathrm{II}$, is proved in Section \ref{sec:region2}.
\begin{theorem}[Region $\mathrm{II}$]  
\label{th:region2}
For $(m,n)\in\mathbb{Z}^2$ with $m < 0$, $m+n\geq1$, the singularity signature of the generalised Okamoto rational $q_{m,n}(t)$ is as follows:
\begin{description}
    \item[$m=-2\mu$ even, $n=2\mu+2\nu$ even] 
        \begin{flalign*}
           ~&\qquad\qquad\qquad   \mathfrak{S}(q_{m,n}) = (z_-\,p_+\,p_-\,z_+)^{\mu}\,(z_-\,p_+)^{\nu}\,\hat{z}_-\,(p_+\,z_-)^{\nu}\,(z_+\,p_-\,p_+\,z_-)^{\mu}. &
        \end{flalign*}
    \item[$m=-2\mu$ even, $n=2\mu+2\nu+1$ odd] 
        \begin{flalign*}
           ~&\qquad\qquad\qquad   \mathfrak{S}(q_{m,n}) = (z_-\,p_+\,p_-\,z_+)^{\mu}\,(z_-\,p_+)^{\nu}\,z_-\,\hat{p}_+\,z_-\,(p_+\,z_-)^{\nu}\,(z_+\,p_-\,p_+\,z_-)^{\mu}. &
        \end{flalign*}
    \item[$m=-2\mu-1$ odd, $n=2\mu+2\nu$ even] 
        \begin{flalign*}
           ~&\qquad\qquad\qquad   \mathfrak{S}(q_{m,n}) = (z_-\,p_+\,p_-\,z_+)^{\mu}\,z_-\,p_+\,p_-\,(z_+\,p_-)^{\nu-1}\hat{z}_+\,(p_-\,z_+)^{\nu-1}\,p_-\,p_+\,z_-\,(z_+\,p_-\,p_+\,z_-)^{\mu}. &
        \end{flalign*}    
    \item[$m=-2\mu-1$ odd, $n=2\mu+2\nu+1$ odd] 
        \begin{flalign*}
           ~&\qquad\qquad\qquad   \mathfrak{S}(q_{m,n}) = (z_-\,p_+\,p_-\,z_+)^{\mu}\,z_-\,p_+\,(p_-\,z_+)^{\nu}\hat{p}_-\,(z_+\,p_-)^{\nu}\,p_+\,z_-\,(z_+\,p_-\,p_+\,z_-)^{\mu}. &
        \end{flalign*}    
\end{description}
\end{theorem}

The following theorem, concerning region  $\mathrm{III}$, is proved in Section \ref{sec:region3}.
\begin{theorem}[Region $\mathrm{III}$] 
\label{th:region3}
 For $(m,n)\in\mathbb{Z}^2$ with $m+n\leq0 $, $n\geq 0$, the singularity signature of the generalised Okamoto rational $q_{m,n}(t)$ is as follows:
\begin{description}
    \item[$m=-2\mu-2\nu$ even, $n=2\nu$ even] 
        \begin{flalign*}
           ~&\qquad\qquad\qquad    \mathfrak{S}(q_{m,n}) = (z_-\,p_+\,p_-\,z_+)^{\nu}\, (z_-\,z_+)^{\mu}\,\hat{z}_-\,(z_+\,z_-)^{\mu}(z_+\,p_-\,p_+\,z_-)^{\nu}. &
        \end{flalign*}
    \item[$m=-2\mu-2\nu$ even, $n=2\nu+1$ odd] 
        \begin{flalign*}
           ~&\qquad\qquad\qquad    \mathfrak{S}(q_{m,n}) = (z_-\,p_+\,p_-\,z_+)^{\nu}\, z_-\,(p_+\,p_-)^{\mu}\,\hat{p}_+\,(p_-\,p_+)^{\mu}\,z_-\,(z_+\,p_-\,p_+\,z_-)^{\nu}. &
        \end{flalign*} 
    \item[$m=-2\mu-2\nu-1$ odd, $n=2\nu$ even] 
        \begin{flalign*}
           ~&\qquad\qquad\qquad    \mathfrak{S}(q_{m,n}) = (z_-\,p_+\,p_-\,z_+)^{\nu}\, z_-\,(z_+\,z_-)^{\mu}\,\hat{z}_+\,(z_-\,z_+)^{\mu}\,z_-(z_+\,p_-\,p_+\,z_-)^{\nu}. &
        \end{flalign*} 
    \item[$m=-2\mu-2\nu-1$ odd, $n=2\nu+1$ odd] 
        \begin{flalign*}
           ~&\qquad\qquad\qquad    \mathfrak{S}(q_{m,n}) = (z_-\,p_+\,p_-\,z_+)^{\nu}\, z_-\,(p_+\,p_-)^{\mu}\,p_+\,\hat{p}_-\,p_+\,(p_-\,p_+)^{\mu}\,z_-\,(z_+\,p_-\,p_+\,z_-)^{\nu}. &
        \end{flalign*} 
\end{description}
\end{theorem}

The remaining three regions are dealt with largely by reusing lemmas required for the proofs of the theorems above, exploiting the symmetry 
\begin{equation*}
    (a_0,a_1,a_2)\mapsto(a_0+a_1,-a_1,a_2+a_1),
\end{equation*}
which leaves equation \eqref{eq:piv} completely invariant.


The following theorem, concerning region  $\mathrm{IV}$, is proved in Section \ref{subsec:region4}.
\begin{theorem}[Region $\mathrm{IV}$] 
\label{th:region4}
For $(m,n)\in\mathbb{Z}^2$ with $m \leq -1$, $n\leq-1$, the singularity signature of the generalised Okamoto rational $q_{m,n}(t)$ is as follows:

\begin{description}
    \item[$m=-2\mu$ even, $n=-2\nu$ even]
        \begin{flalign*}
           ~&\qquad\qquad\qquad \mathfrak{S}(q_{m,n}) = (z_+\,p_+\,p_-\,z_-)^{\nu}\,(z_+\,z_-)^{\mu-1}\,z_+\,\hat{z}_-\,z_+\,(z_-\,z_+)^{\mu-1}\,(z_-\,p_-\,p_+,\,z_+)^{\nu}.  &      
        \end{flalign*}
    \item[$m=-2\mu$ even, $n=-2\nu-1$ odd] 
        \begin{flalign*}
           ~&\qquad\qquad\qquad  \mathfrak{S}(q_{m,n}) = (z_+\,p_+\,p_-\,z_-)^{\nu}\,z_+ \,(p_+\,p_-)^{\mu}\,\hat{p}_+\,(p_-\,p_+)^{\mu}\,z_+\,(z_-\,p_-\,p_+,\,z_+)^{\nu}.  &
        \end{flalign*}
     \item[$m=-2\mu-1$ odd, $n=-2\nu$ even]
        \begin{flalign*}
           ~&\qquad\qquad\qquad   \mathfrak{S}(q_{m,n}) = (z_+\,p_+\,p_-\,z_-)^{\nu}\,(z_+\,z_-)^{\mu}\,\hat{z}_+\,(z_-\,z_+)^{\mu}\,(z_-\,p_-\,p_+,\,z_+)^{\nu}.  &
        \end{flalign*}
    \item[$m=-2\mu-1$ odd, $n=-2\nu-1$ odd]
        \begin{flalign*}
           ~&\qquad\qquad\qquad   \mathfrak{S}(q_{m,n}) = (z_+\,p_+\,p_-\,z_-)^{\nu}\,z_+ \,(p_+\,p_-)^{\mu}\,p_+\,\hat{p}_-\,p_+\,(p_-\,p_+)^{\mu}\,z_+\,(z_-\,p_-\,p_+,\,z_+)^{\nu}. &
        \end{flalign*}
\end{description}
\end{theorem}

\begin{remark}[Relation between signatures in regions III and IV]
 For $(m,n)\in\mathbb{Z}^2$ with $m+n\leq0 $, $n\geq 0$, the singularity signature of the solution $q_{m,n}(t)$ is given by
 \begin{equation*}
    \mathfrak{S}(q_{m,n})= \mathfrak{S}(q_{m+n-1,-n})|_{z_+\leftrightarrow z_-},
 \end{equation*}
where the left-hand side is given by Theorem \ref{th:region3} for Region $\mathrm{III}$ and the right-hand side is given by Theorem \ref{th:region4} for Region $\mathrm{IV}$, with the symbols $z_+$ and $z_-$ swapped.
\end{remark}

The following theorem, concerning region  $\mathrm{V}$, is proved in Section \ref{subsec:region5}.
\begin{theorem}[Region $\mathrm{V}$] 
\label{th:region5}
 For $(m,n)\in\mathbb{Z}^2$ with $m\geq0 $, $m+n\leq 0$, the singularity signature of the Okamoto rational $q_{m,n}(t)$ is as follows:
\begin{description}
   \item[$m=0$, $n=0$] 
        \begin{flalign*}
           ~&\qquad\qquad\qquad    \mathfrak{S}(q_{m,n}) = \hat{z}_-. &
        \end{flalign*}
   \item[$m=2\mu$ even, $n=-2\mu-2\nu$ even, $\mu\geq 1$] 
        \begin{flalign*}
           ~&\qquad\qquad\qquad    \mathfrak{S}(q_{m,n}) = (z_+\,p_+\,p_-\,z_-)^{\nu}\, (p_-\,z_-)^{\mu-1}\,p_-\,\hat{z}_-\,p_-\,(z_-\,p_-)^{\mu-1}(z_-\,p_-\,p_+\,z_+)^{\nu}. &
        \end{flalign*}
        
   \item[$m=2\mu$ even, $n=-2\mu-2\nu$ even, $\nu\geq 1$] 
        \begin{flalign*}
           ~&\qquad\qquad\qquad    \mathfrak{S}(q_{m,n}) = (z_+\,p_+\,p_-\,z_-)^{\nu-1}\,(z_+\,p_+)\, (p_-\,z_-)^{\mu}\,p_-\,\hat{z}_-\,p_-\,(z_-\,p_-)^{\mu}\,(p_+\,z_+)\,(z_-\,p_-\,p_+\,z_+)^{\nu-1}. &
        \end{flalign*}
        
    \item[$m=2\mu$ even, $n=-2\mu-2\nu-1$ odd] 
        \begin{flalign*}
           ~&\qquad\qquad\qquad    \mathfrak{S}(q_{m,n}) = (z_+\,p_+\,p_-\,z_-)^{\nu}\, (z_+\,p_+)^{\mu}\,z_+\,\hat{p}_+\,z_+\,(p_+\,z_+)^{\mu}(z_-\,p_-\,p_+\,z_+)^{\nu}. &
        \end{flalign*}
    \item[$m=2\mu+1$ odd, $n=-2\mu-2\nu$ even] 
        \begin{flalign*}
           ~&\qquad\qquad\qquad    \mathfrak{S}(q_{m,n}) = (z_+\,p_+\,p_-\,z_-)^{\nu-1}\, (z_+\,p_+)^{\mu+1}\,\hat{z}_+\,(p_+\,z_+)^{\mu+1}(z_-\,p_-\,p_+\,z_+)^{\nu-1}. &
        \end{flalign*}
    \item[$m=2\mu+1$ odd, $n=-2\mu-2\nu-1$ odd] 
        \begin{flalign*}
           ~&\qquad\qquad\qquad    \mathfrak{S}(q_{m,n}) = (z_+\,p_+\,p_-\,z_-)^{\nu}\, (p_-\,z_-)^{\mu}\,\hat{p}_-\,(z_-\,p_-)^{\mu}(z_-\,p_-\,p_+\,z_+)^{\nu}. &
        \end{flalign*}
\end{description}
\end{theorem}


\begin{remark}[Relation between signatures in regions $\mathrm{II}$ and $\mathrm{V}$]
 For $(m,n)\in\mathbb{Z}^2$ with $m\geq0 $, $m+n\leq 0$, the singularity signature of the solution $q_{m,n}$ is given by
 \begin{equation*}
    \mathfrak{S}(q_{m,n})=\mathfrak{S}(q_{m+n+1,-n})|_{z_+\leftrightarrow z_-}
 \end{equation*}
where the left-hand side is given by Theorem \ref{th:region5} for Region $\mathrm{V}$ and the right-hand side is given by Theorem \ref{th:region2} for Region $\mathrm{II}$, with the symbols $z_+$ and $z_-$ swapped.
\end{remark}

The following theorem, concerning region  $\mathrm{VI}$, is proved in Section \ref{subsec:region6}.
\begin{theorem}[Region $\mathrm{VI}$] 
\label{th:region6}
 For $(m,n)\in\mathbb{Z}^2$ with $m+n\geq 1$, $n\leq -1$, the singularity signature of the generalised Okamoto rational $q_{m,n}(t)$ is as follows:
\begin{description}
               
    \item[$m=2\mu+2\nu$ even, $n=-2\nu$ even] 
        \begin{flalign*}
           ~&\qquad\qquad\qquad    \mathfrak{S}(q_{m,n}) = (p_-\,z_-\,z_+\,p_+)^{\mu} (p_-\,z_-)^{\nu-1}\,p_- \,\hat{z}_-\,p_-(z_-\,p_-)^{\nu-1}(p_+\,z_+\,z_-\,p_-)^{\mu}. &
        \end{flalign*}
    \item[$m=2+2\mu+2\nu$ even, $n=-1-2\nu$ odd] 
        \begin{flalign*}
           ~&\qquad\qquad\qquad  \mathfrak{S}(q_{m,n}) =  (p_-\,z_-\,z_+\,p_+)^{\mu}\,p_-\,z_-\,(z_+\,p_+)^{\nu}\,z_+\,\hat{p}_+\, z_+\,(p_+\,z_+)^{\nu}\,z_-\,p_-\,(p_+\,z_+\,z_-\,p_-)^{\mu}. &
        \end{flalign*} 
    \item[$m=1+2\mu+2\nu$ odd, $n=-2\nu$ even]
        \begin{flalign*}
           ~&\qquad\qquad\qquad    \mathfrak{S}(q_{m,n}) = (p_-\,z_-\,z_+\,p_+)^{\mu+1}\,(z_+\,p_+)^{\nu-1}\,\hat{z}_+\,(p_+\,z_+)^{\nu-1}\,(p_+\,z_+\,z_-,\,p_-)^{\mu+1}. &
        \end{flalign*}

    \item[$m=1+2\mu+2\nu$ odd, $n=-1-2\nu$ odd] 
        \begin{flalign*}
           ~&\qquad\qquad\qquad   \mathfrak{S}(q_{m,n}) = (p_-\,z_-\,z_+\,p_+)^{\mu}(p_-\,z_-)^{\nu}\,\hat{p}_-\,(z_-\,p_-)^{\nu}(p_+\,z_+\,z_-\,p_-)^{\mu}. &
        \end{flalign*}  
        
\end{description}
\end{theorem}

\begin{remark}[Relation between signatures in regions $\mathrm{I}$ and $\mathrm{VI}$]
 For $(m,n)\in\mathbb{Z}^2$ with $m+n\geq 1$, $n\leq -1$, the singularity signature of the solution $q_{m,n}(t)$ is given by
 \begin{equation*}
    \mathfrak{S}(q_{m,n})= \mathfrak{S}(q_{m+n+1,-n-2})|_{z_+\leftrightarrow z_-},
 \end{equation*}
where the left-hand side is given by Theorem \ref{th:region6} for Region $\mathrm{VI}$ and the right-hand side is given by Theorem \ref{th:region1} for Region $\mathrm{I}$, with the symbols $z_+$ and $z_-$ swapped.
\end{remark}


In light of the fact that the generalised Okamoto rational $q_{m,n}(t)$ has a minus pole at $t=t_*$ if and only if $Q_{m,n}(t)$ has a root at $t=t_*$, by counting the number of occurrences of $p_-$ in $\mathfrak{S}(q_{m,n})$ and noting that real roots of $Q_{m,n}(t)$ are related to imaginary roots of $Q_{n,m}(t)$ via the relation \eqref{realimagsymmetry} in Remark \ref{remark:facts}, we obtain the following.

\begin{corollary}\label{cor:realrootfree}
For $m,n\in\mathbb{Z}$, the number of real roots of $Q_{m,n}$ is given in Table \ref{table:number_roots}.
In particular, the polynomial $Q_{m,n}$ has no real roots if and only if
\begin{enumerate}
    \item $m=0$ and $n\geq 0$, in which case the number of imaginary roots is 
    \begin{equation*}
    \rho_{im}(Q_{0,n})=\rho_{re}(Q_{n,0})=\begin{cases}
        n & \text{if $n$ is even,}\\
        n+1 & \text{if $n$ is odd,}
    \end{cases}
\end{equation*}
    \item  $n=0$ and $m\leq 0$, in which case the number of imaginary roots is 
    \begin{equation*}
    \rho_{im}(Q_{m,0})=\rho_{re}(Q_{0,m})=\begin{cases}
        -m & \text{if $m$ is even,}\\
        -m+1 & \text{if $m$ is odd,}
    \end{cases}
\end{equation*}
or
    \item $m\geq 0$, $n=-1-m$, in which case the number of imaginary roots is 
    \begin{equation*}
    \rho_{im}(Q_{m,-1-m})=\rho_{re}(Q_{-1-m,m})=\begin{cases}
        m & \text{if $m$ is even,}\\
        m+1 & \text{if $m$ is odd.}
    \end{cases}
\end{equation*}
\end{enumerate}
\end{corollary} 
\begin{remark}
    In \cite{HMZ22}, the number of real roots of $Q_{m,n}$, for $(m,n)$ in region $\mathrm{IV}$, is determined by applying \cite[Theorem 1.4]{paper65} to a Wronskian representation of $Q_{m,n}$, which is equivalent to \eqref{eq:wronskianrep} up to some scaling.
\end{remark}

\renewcommand{\arraystretch}{1.1}
\begin{table}[ht]
\centering
\begin{tabular}{|c || c | c | c | c |} 
 \hline
region & $m,n$ even & $m$ even, $n$ odd & $m$ odd, $n$ even & $m$ odd, $n$ odd\\
 \hline
 $\mathrm{I}$ & $|m|$ & $|m|$ & $|m+n+1|$ & $|m+n+1|$\\
  $\mathrm{II}$ & $|m|$ & $|m|$ & $|n|$ & $|n|$\\
   $\mathrm{III}$ & $|n|$ & $|m|$ & $|n|$ & $|m|$\\
    $\mathrm{IV}$ & $|n|$ & $|m+n+1|$ & $|n|$ & $|m+n+1|$\\
     $\mathrm{V}$ & $|n|$ & $|m+n+1|$ & $|m+n+1|$ & $|n|$\\
      $\mathrm{VI}$ & $|m|$ & $|m+n+1|$ & $|m+n+1|$ & $|m|$\\\hline
\end{tabular}
\caption{Number of real roots of $Q_{m,n}(t)$ dependent on the parity of $m$ and $n$ as well as the region where the indices $(m,n)\in\mathbb{Z}^2$ lie.
}
\label{table:number_roots}
\end{table}

In Figure \ref{fig:zeroestable}, the numbers 
 of real roots of generalised Okamoto polynomials are given for the indices $m,n$ ranging between $-7$ and $7$.


\begin{figure}[htb]
    \centering
 \begin{tikzpicture}[scale=.75]
\draw [green!70!black!30,thick] 	(1,-1) 	-- (1,8);
\draw [green!70!black!30,thick] 	(1,-1) 	-- (8,-8);
\draw [green!70!black!30,thick] 	(2,-1) 	-- (2,8);
\draw [green!70!black!30,thick] 	(2,-1) 	-- (8,-7);
\draw [green!70!black!30,thick] 	(3,-1) 	-- (3,8);
\draw [green!70!black!30,thick] 	(3,-1) 	-- (8,-6);
\draw [green!70!black!30,thick] 	(4,-1) 	-- (4,8);
\draw [green!70!black!30,thick] 	(4,-1) 	-- (8,-5);
\draw [green!70!black!30,thick] 	(5,-1) 	-- (5,8);
\draw [green!70!black!30,thick] 	(5,-1) 	-- (8,-4);
\draw [green!70!black!30,thick] 	(6,-1) 	-- (6,8);
\draw [green!70!black!30,thick] 	(6,-1) 	-- (8,-3);
\draw [green!70!black!30,thick] 	(7,-1) 	-- (7,8);
\draw [green!70!black!30,thick] 	(7,-1) 	-- (8,-2);
\draw [green!70!black!30,thick] 	(-8,1) 	-- (-1,1);
\draw [green!70!black!30,thick] 	(-1,1) 	-- (-1,8);
\draw [green!70!black!30,thick] 	(-8,2) 	-- (-2,2);
\draw [green!70!black!30,thick] 	(-2,2) 	-- (-2,8);
\draw [green!70!black!30,thick] 	(-8,3) 	-- (-3,3);
\draw [green!70!black!30,thick] 	(-3,3) 	-- (-3,8);
\draw [green!70!black!30,thick] 	(-8,4) 	-- (-4,4);
\draw [green!70!black!30,thick] 	(-4,4) 	-- (-4,8);
\draw [green!70!black!30,thick] 	(-8,5) 	-- (-5,5);
\draw [green!70!black!30,thick] 	(-5,5) 	-- (-5,8);
\draw [green!70!black!30,thick] 	(-8,6) 	-- (-6,6);
\draw [green!70!black!30,thick] 	(-6,6) 	-- (-6,8);
\draw [green!70!black!30,thick] 	(-8,7) 	-- (-7,7);
\draw [green!70!black!30,thick] 	(-7,7) 	-- (-7,8);
\draw [green!70!black!30,thick] 	(-8,-1) -- (-1,-1);
\draw [green!70!black!30,thick] 	(-1,-1) -- (6,-8);
\draw [green!70!black!30,thick] 	(-8,-2) -- (-1,-2);
\draw [green!70!black!30,thick] 	(-1,-2) -- (5,-8);
\draw [green!70!black!30,thick] 	(-8,-3) -- (-1,-3);
\draw [green!70!black!30,thick] 	(-1,-3) -- (4,-8);
\draw [green!70!black!30,thick] 	(-8,-4) -- (-1,-4);
\draw [green!70!black!30,thick] 	(-1,-4) -- (3,-8);
\draw [green!70!black!30,thick] 	(-8,-5) -- (-1,-5);
\draw [green!70!black!30,thick] 	(-1,-5) -- (2,-8);
\draw [green!70!black!30,thick] 	(-8,-6) -- (-1,-6);
\draw [green!70!black!30,thick] 	(-1,-6) -- (1,-8);
\draw [green!70!black!30,thick] 	(-8,-7) -- (-1,-7);
\draw [green!70!black!30,thick] 	(-1,-7) -- (0,-8);
    \foreach \x in {-7,-6,...,7} {
        \foreach \y in {-7,-6,...,7} {
            \fill[color=black,opacity=.5] (\x,\y) circle (0.075);
        }
    }
    \draw [black] 	(0,-8) 	-- (0,8) node [above] {$n$};
    \draw [black] 	(-8,0) 	-- (8,0) node [right] {$m$};
    \draw [red, dashed]   (-8,-1/3) -- (8,-1/3) node [pos=0,left] {$a_1=0$};
    \draw [red, dashed]   (-1/3,-8) -- (-1/3,8) node [pos=0,below] {$a_2=0$};
    \draw [red, dashed]   (-8+1/3,25/3-1/3) -- (8-1/3,-23/3+1/3) node [pos=1,right] {$a_0=0$};
    \draw [draw=red!50, fill=red!50, opacity=0.2, rounded corners] 
    (-1/4,-1/4) -- (7+1/4,-1/4) -- (7+1/4,7+1/4) --(-1/4,7+1/4) -- (-1/4,-1/4) ;
    \draw [draw=blue!50, fill=blue!50, opacity=0.2, rounded corners] 
    (1+1/3,-1/3-1/4) -- (7+1/4,-1/3-1/4) -- (7+1/4,-7+1/4) --(1+1/3,-1+1/4) ;
    \draw [draw=red!50, fill=red!50, opacity=0.2, rounded corners] 
    (-1/4,+2/4) -- (7+2/4,-7-1/4) -- (-1/4,-7-1/4) --(-1/4,+1/4) ;
    \draw [draw=blue!50, fill=blue!50, opacity=0.2, rounded corners] 
    (-1/3-1/4,-1/3-1/4) -- (-7-1/4,-1/3-1/4) -- (-7-1/4,-7-1/4) --(-1/3-1/4,-7-1/4) -- (-1/3-1/4,-1/3-1/4) ;
    \draw [draw=red!50, fill=red!50, opacity=0.2, rounded corners] 
    (+2/4,-1/4) -- (-7-1/4,-1/4) -- (-7-1/4,7+2/4) --(+2/4,-1/4) ;
    \draw [draw=blue!50, fill=blue!50, opacity=0.2, rounded corners] 
    (-1/3-1/4,1+1/3) -- (-1/3-1/4,7+1/4) -- (-7+1/4,7+1/4) --(-1+1/4,1+1/3) ;
\begin{scope}[xshift=.25cm,yshift=.275cm]
    \node at (-7,7) {$7$};
    \node at (-6,7) {$6$};
    \node at (-5,7) {$7$};
    \node at (-4,7) {$4$};
    \node at (-3,7) {$7$};
    \node at (-2,7) {$2$};
    \node at (-1,7) {$7$};
    \node at (0,7) {$0$};
    \node at (1,7) {$9$};
    \node at (2,7) {$2$};
    \node at (3,7) {$11$};
    \node at (4,7) {$4$};
    \node at (5,7) {$13$};
    \node at (6,7) {$6$};
    \node at (7,7) {$15$};
    \node at (-7,6) {$6$};
    \node at (-6,6) {$6$};
    \node at (-5,6) {$6$};
    \node at (-4,6) {$4$};
    \node at (-3,6) {$6$};
    \node at (-2,6) {$2$};
    \node at (-1,6) {$6$};
    \node at (0,6) {$0$};
    \node at (1,6) {$8$};
    \node at (2,6) {$2$};
    \node at (3,6) {$10$};
    \node at (4,6) {$4$};
    \node at (5,6) {$12$};
    \node at (6,6) {$6$};
    \node at (7,6) {$14$};
    \node at (-7,5) {$7$};
    \node at (-6,5) {$6$};
    \node at (-5,5) {$5$};
    \node at (-4,5) {$4$};
    \node at (-3,5) {$5$};
    \node at (-2,5) {$2$};
    \node at (-1,5) {$5$};
    \node at (0,5) {$0$};
    \node at (1,5) {$7$};
    \node at (2,5) {$2$};
    \node at (3,5) {$9$};
    \node at (4,5) {$4$};
    \node at (5,5) {$11$};
    \node at (6,5) {$6$};
    \node at (7,5) {$13$};
    \node at (-7,4) {$4$};
    \node at (-6,4) {$4$};
    \node at (-5,4) {$4$};
    \node at (-4,4) {$4$};
    \node at (-3,4) {$4$};
    \node at (-2,4) {$2$};
    \node at (-1,4) {$4$};
    \node at (0,4) {$0$};
    \node at (1,4) {$6$};
    \node at (2,4) {$2$};
    \node at (3,4) {$8$};
    \node at (4,4) {$4$};
    \node at (5,4) {$10$};
    \node at (6,4) {$6$};
    \node at (7,4) {$12$};
    \node at (-7,3) {$7$};
    \node at (-6,3) {$6$};
    \node at (-5,3) {$5$};
    \node at (-4,3) {$4$};
    \node at (-3,3) {$3$};
    \node at (-2,3) {$2$};
    \node at (-1,3) {$3$};
    \node at (0,3) {$0$};
    \node at (1,3) {$5$};
    \node at (2,3) {$2$};
    \node at (3,3) {$7$};
    \node at (4,3) {$4$};
    \node at (5,3) {$9$};
    \node at (6,3) {$6$};
    \node at (7,3) {$11$};
    \node at (-7,2) {$2$};
    \node at (-6,2) {$2$};
    \node at (-5,2) {$2$};
    \node at (-4,2) {$2$};
    \node at (-3,2) {$2$};
    \node at (-2,2) {$2$};
    \node at (-1,2) {$2$};
    \node at (0,2) {$0$};
    \node at (1,2) {$4$};
    \node at (2,2) {$2$};
    \node at (3,2) {$6$};
    \node at (4,2) {$4$};
    \node at (5,2) {$8$};
    \node at (6,2) {$6$};
    \node at (7,2) {$10$};
    \node at (-7,1) {$7$};
    \node at (-6,1) {$6$};
    \node at (-5,1) {$5$};
    \node at (-4,1) {$4$};
    \node at (-3,1) {$3$};
    \node at (-2,1) {$2$};
    \node at (-1,1) {$1$};
    \node at (0,1) {$0$};
    \node at (1,1) {$3$};
    \node at (2,1) {$2$};
    \node at (3,1) {$5$};
    \node at (4,1) {$4$};
    \node at (5,1) {$7$};
    \node at (6,1) {$6$};
    \node at (7,1) {$9$};
    \node at (-7,0) {$0$};
    \node at (-6,0) {$0$};
    \node at (-5,0) {$0$};
    \node at (-4,0) {$0$};
    \node at (-3,0) {$0$};
    \node at (-2,0) {$0$};
    \node at (-1,0) {$0$};
    \node at (0,0) {$0$};
    \node at (1,0) {$2$};
    \node at (2,0) {$2$};
    \node at (3,0) {$4$};
    \node at (4,0) {$4$};
    \node at (5,0) {$6$};
    \node at (6,0) {$6$};
    \node at (7,0) {$8$};
    \node at (-7,-1) {$7$};
    \node at (-6,-1) {$6$};
    \node at (-5,-1) {$5$};
    \node at (-4,-1) {$4$};
    \node at (-3,-1) {$3$};
    \node at (-2,-1) {$2$};
    \node at (-1,-1) {$1$};
    \node at (0,-1) {$0$};
    \node at (1,-1) {$1$};
    \node at (2,-1) {$2$};
    \node at (3,-1) {$3$};
    \node at (4,-1) {$4$};
    \node at (5,-1) {$5$};
    \node at (6,-1) {$6$};
    \node at (7,-1) {$7$};

    \node at (-7,-2) {$2$};
    \node at (-6,-2) {$2$};
    \node at (-5,-2) {$2$};
    \node at (-4,-2) {$2$};
    \node at (-3,-2) {$2$};
    \node at (-2,-2) {$2$};
    \node at (-1,-2) {$2$};
    \node at (0,-2) {$2$};
    \node at (1,-2) {$0$};
    \node at (2,-2) {$2$};
    \node at (3,-2) {$2$};
    \node at (4,-2) {$4$};
    \node at (5,-2) {$4$};
    \node at (6,-2) {$6$};
    \node at (7,-2) {$6$};

    \node at (-7,-3) {$9$};
    \node at (-6,-3) {$8$};
    \node at (-5,-3) {$7$};
    \node at (-4,-3) {$6$};
    \node at (-3,-3) {$5$};
    \node at (-2,-3) {$4$};
    \node at (-1,-3) {$3$};
    \node at (0,-3) {$2$};
    \node at (1,-3) {$3$};
    \node at (2,-3) {$0$};
    \node at (3,-3) {$3$};
    \node at (4,-3) {$2$};
    \node at (5,-3) {$5$};
    \node at (6,-3) {$4$};
    \node at (7,-3) {$7$};

    \node at (-7,-4) {$4$};
    \node at (-6,-4) {$4$};
    \node at (-5,-4) {$4$};
    \node at (-4,-4) {$4$};
    \node at (-3,-4) {$4$};
    \node at (-2,-4) {$4$};
    \node at (-1,-4) {$4$};
    \node at (0,-4) {$4$};
    \node at (1,-4) {$2$};
    \node at (2,-4) {$4$};
    \node at (3,-4) {$0$};
    \node at (4,-4) {$4$};
    \node at (5,-4) {$2$};
    \node at (6,-4) {$6$};
    \node at (7,-4) {$4$};

    \node at (-7,-5) {$11$};
    \node at (-6,-5) {$10$};
    \node at (-5,-5) {$9$};
    \node at (-4,-5) {$8$};
    \node at (-3,-5) {$7$};
    \node at (-2,-5) {$6$};
    \node at (-1,-5) {$5$};
    \node at (0,-5) {$4$};
    \node at (1,-5) {$5$};
    \node at (2,-5) {$2$};
    \node at (3,-5) {$5$};
    \node at (4,-5) {$0$};
    \node at (5,-5) {$5$};
    \node at (6,-5) {$2$};
    \node at (7,-5) {$7$};

    \node at (-7,-6) {$6$};
    \node at (-6,-6) {$6$};
    \node at (-5,-6) {$6$};
    \node at (-4,-6) {$6$};
    \node at (-3,-6) {$6$};
    \node at (-2,-6) {$6$};
    \node at (-1,-6) {$6$};
    \node at (0,-6) {$6$};
    \node at (1,-6) {$4$};
    \node at (2,-6) {$6$};
    \node at (3,-6) {$2$};
    \node at (4,-6) {$6$};
    \node at (5,-6) {$0$};
    \node at (6,-6) {$6$};
    \node at (7,-6) {$2$};       
    \node at (-7,-7) {$13$};
    \node at (-6,-7) {$12$};
    \node at (-5,-7) {$11$};
    \node at (-4,-7) {$10$};
    \node at (-3,-7) {$9$};
    \node at (-2,-7) {$8$};
    \node at (-1,-7) {$7$};
    \node at (0,-7) {$6$};
    \node at (1,-7) {$7$};
    \node at (2,-7) {$4$};
    \node at (3,-7) {$7$};
    \node at (4,-7) {$2$};
    \node at (5,-7) {$7$};
    \node at (6,-7) {$0$};
    \node at (7,-7) {$7$};
    \node at (-7,3) {$7$};
    \node at (-6,3) {$6$};
    \node at (-5,3) {$5$};
    \node at (-4,3) {$4$};
    \node at (-3,3) {$3$};
    \node at (-2,3) {$2$};
    \node at (-1,3) {$3$};
    \node at (0,3) {$0$};
    \node at (1,3) {$5$};
    \node at (2,3) {$2$};
    \node at (3,3) {$7$};
    \node at (4,3) {$4$};
    \node at (5,3) {$9$};
    \node at (6,3) {$6$};
    \node at (7,3) {$11$};
    \end{scope}
   \end{tikzpicture}    
    \caption{Number of real roots of $Q_{m,n}$ in the $(m,n)$-plane, for $-7\leq m,n\leq 7$, as well as curves (in green) along which real roots are interlaced as detailed in Corollary \ref{cor:interlacing}.}
    \label{fig:zeroestable}
\end{figure}
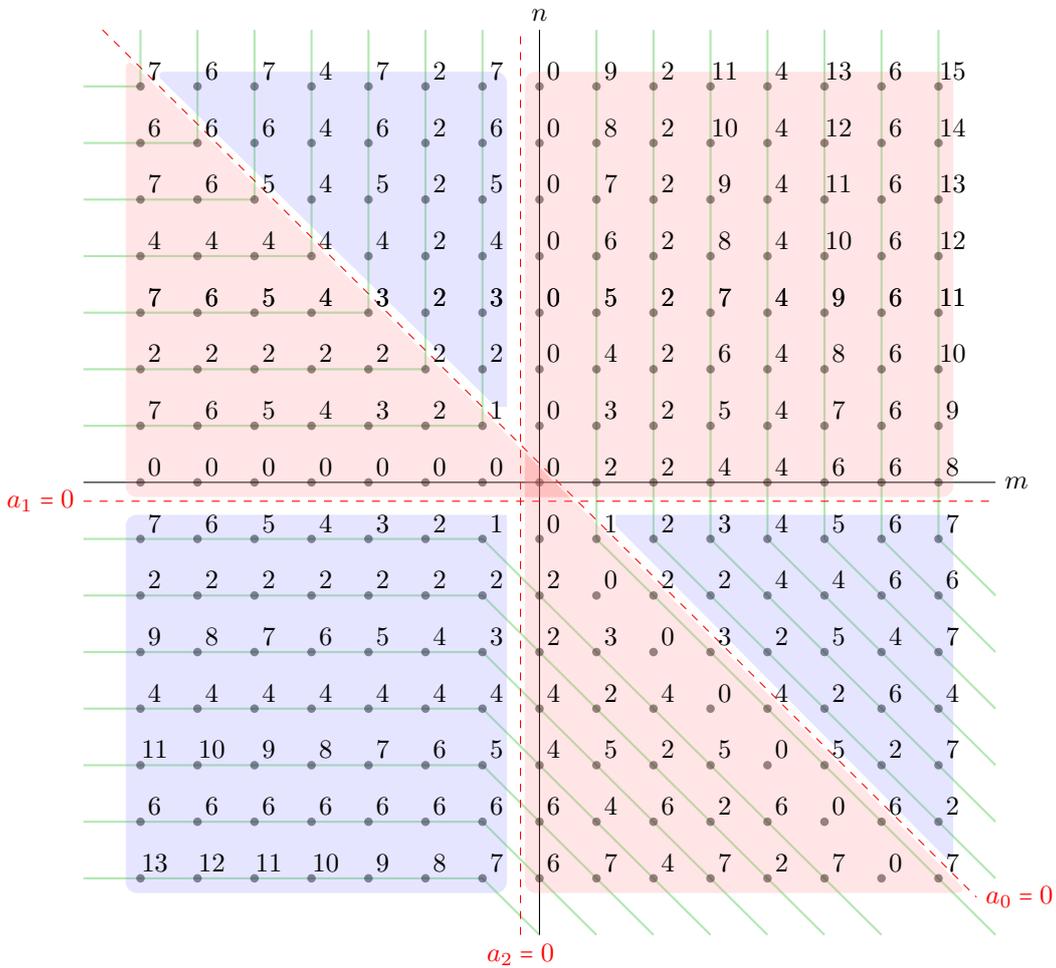

On the one hand, this provides us with the numbers of real and imaginary roots of the polynomials relevant to the rational solutions of the Sasa-Satsuma equation as in Corollary \ref{cor:numberofrootsYY}. 
On the other hand, this characterises precisely which of the generalised Okamoto rational solutions $q_{m,n}$ are pole-free on the real line or on the imaginary line.

\begin{corollary}
\begin{enumerate}[(a)]
\item The generalised Okamoto rational solution $q_{m,n}(t)$ has no poles on the real line if and only if
$n=0$ and $m\leq0$.
%
\item The generalised Okamoto rational solution $q_{m,n}(t)$ has no poles on the imaginary line if and only if
$m=0$ and $n\leq0$.
\end{enumerate}
\end{corollary}

We give some plots of generalised Okamoto rationals which are pole-free on the real line in Figure \ref{fig:polefreeokamotoratplots}.

\begin{figure}[htb]
    \centering
     \begin{subfigure}[b]{0.48\textwidth}
         \centering
         \includegraphics[width=\textwidth]{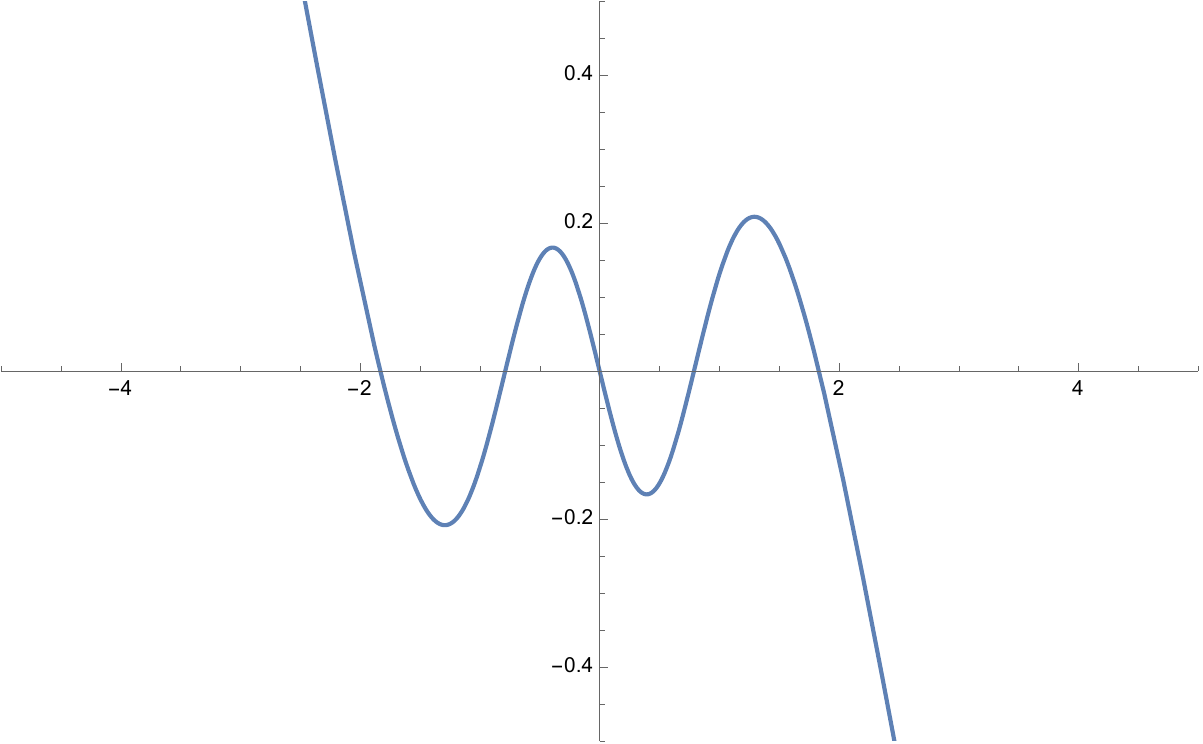}
         \caption{$q_{-2,0}(t)$}
         \label{fig:polefree2}
     \end{subfigure}
    \begin{subfigure}[b]{0.48\textwidth}
         \centering
         \includegraphics[width=\textwidth]{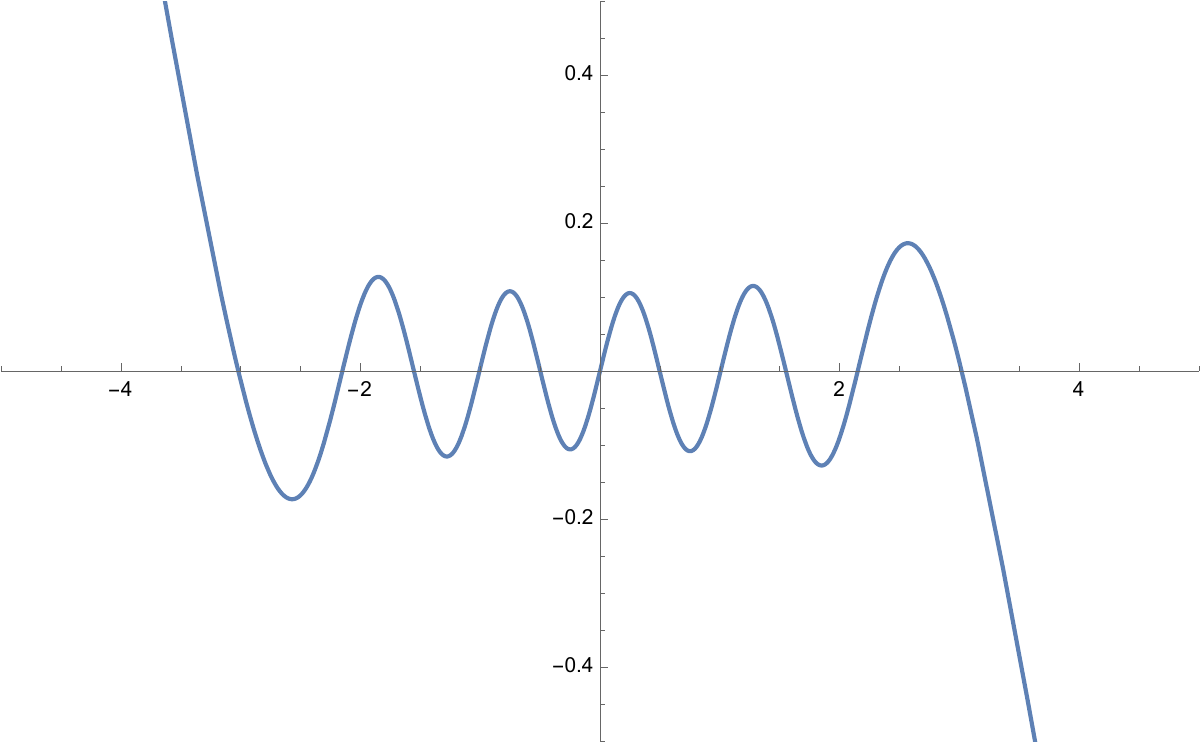}
         \caption{$q_{-5,0}(t)$}
         \label{fig:polefree5}
     \end{subfigure}
    \begin{subfigure}[b]{0.48\textwidth}
         \centering
         \includegraphics[width=\textwidth]{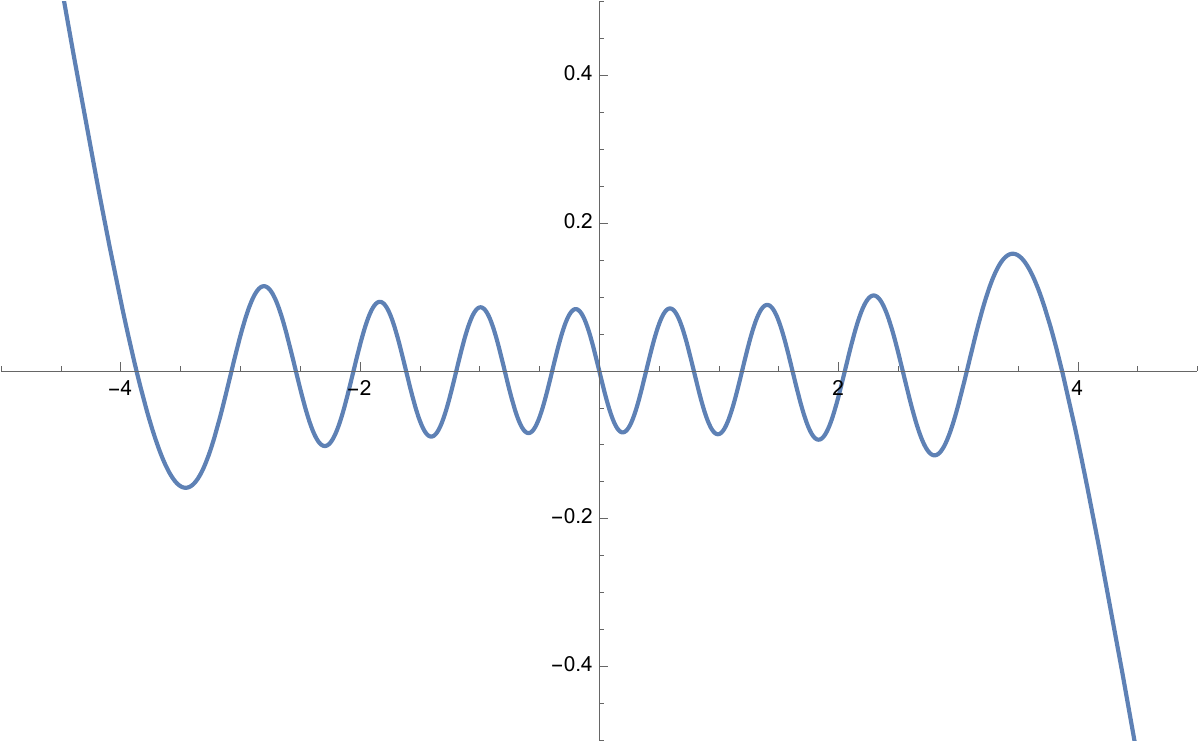}
         \caption{$q_{-8,0}(t)$}
         \label{fig:polefree8}
      \end{subfigure}
   \begin{subfigure}[b]{0.48\textwidth}
         \centering
         \includegraphics[width=\textwidth]{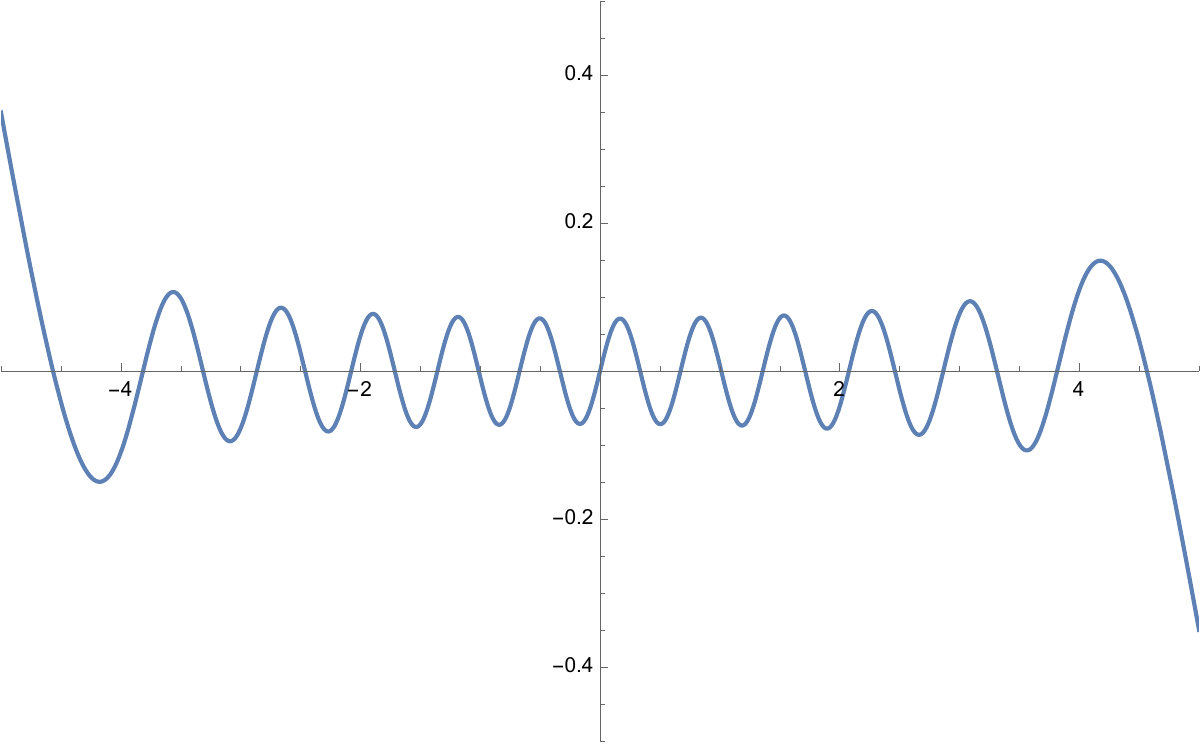}
         \caption{$q_{-11,0}(t)$}
         \label{fig:polefree11}
      \end{subfigure}
   \caption{Plots of some Okamoto rationals which are pole-free on the real line.}
    \label{fig:polefreeokamotoratplots}
\end{figure}


\begin{remark}\label{rem:marquette}
An application of Corollary \ref{cor:realrootfree} can be found in the construction of Hamiltonians constrained to fulfil a third-order shape-invariance condition \cite{HMZ22}.
The Hamiltonians in question take the form
\begin{equation*}
    \mathcal{H}=-\frac{d}{dt^2}+V(t),\quad V(t)=t^2-(q'-2t q-q^2)-1,
\end{equation*}
where $q=q(t)$ is a rational solution of the fourth Painlev\'e equation. Note that the potential $V(t)$ is singular only at plus poles of $q(t)$. Setting $q=q_{m,n}$, this is manifest in the following formula,
\begin{equation*}
    V(t)=t^2- \frac{4Q_{m-2,n}Q_{m,n}}{9Q_{m-1,n}^2}-4m-2n+1,
\end{equation*}
which is derived analogously to \cite[Eq. (2.9)]{HMZ22}.
By Corollary \ref{cor:realrootfree}, we thus obtain three families of smooth potentials, defined by the sets of indices
\begin{enumerate}
    \item $m=1, n\geq 0$.
    \item $m\leq 1, n=0$.
    \item $m\geq 1, n=-m$.
\end{enumerate}
Case (2) was given in \cite[\S 2]{HMZ22} and the corresponding potentials were described as ``the most general forms of the rationally extended potentials generated
from the `$-2x/3$' hierarchy.'' The potentials associated with (1) and (3) can be added to this description.
\end{remark}

From the theorems above we can determine whether the real roots of generalised Okamoto polynomials show any kind of interlacing, which is a well-known property of the roots of families of classical orthogonal polynomials.
\begin{definition}[interlacing of real roots]
Given polynomials $P$ and $Q$ with real coefficients and a connected subset $I\subseteq \mathbb{R}$,
we say that the real roots of $P$ and $Q$ are interlaced on $I$ if and only if, in between any two real roots of $P$ in $I$, lies a real root of $Q$ and, in between any two real roots of $Q$ in $I$, lies a real root of $P$.
\end{definition}


An interlacing property has been established  for the family of Yablonskii-Vorob'ev polynomials $(Y_n)_{n\geq 0}$, associated with rational solutions of $\pain{II}$ \cite{yablonskii, vorobev}, by Clarkson \cite{clarksonsurvey} (see also \cite{pieterYV}), namely that the real roots of $Y_{n-1}$ and $Y_{n+1}$ are interlaced on the real line.



Noting the correspondence between roots of generalised Okamoto polynomials and the different kinds of apparent singularities as in Remark \ref{remark:facts}, by inspection of $\mathfrak{S}(q_{m,n})$ in the theorems above we can derive interlacing properties between various generalised Okamoto polynomials.

\begin{corollary}\label{cor:interlacing}
We have the following interlacings of real roots of generalised Okamoto polynomials. 
\begin{enumerate}[(a)]
    \item Let $(m,n)\in \Z^2$ be in region $\mathrm{I}$ or region $\mathrm{II}$. Then the real roots of $Q_{m,n}$ and $Q_{m,n-1}$ are interlaced on the real line when $m$ is odd, and they are interlaced both on the positive and on the negative real line when $m$ is even.

    \item Let $(m,n)\in \Z^2$ be in region $\mathrm{III}$ or region $\mathrm{IV}$. Then the real roots of $Q_{m,n}$ and $Q_{m-1,n}$ are interlaced on the real line when $n$ is odd, and they are interlaced both on the positive and the negative real line when $n$ is even.

    \item Let $(m,n)\in \Z^2$ be in region $\mathrm{V}$ or region $\mathrm{VI}$. Then the real roots of $Q_{m,n}$ and $Q_{m-1,n+1}$ are interlaced on the real line when $m+n+1$ is odd, and they are interlaced both on the positive and the negative real line when $m+n+1$ is even.
\end{enumerate}
\end{corollary}
\begin{proof}
We give the proof of claim (a) in region $\mathrm{I}$. So, take $(m,n)$ in region $\mathrm{I}$.
Firstly, let us assume $m$ is odd. 
Recall that zeros of $Q_{m,n}$ and $Q_{m,n-1}$ correspond respectively to apparent singularities of type $p_-$ and $z_+$ of $q_{m,n}$. 
From Theorem \ref{th:region1} the singularity signature of $q_{m,n}$ for $(m,n)$ in region $\mathrm{I}$ is given for $n=2\nu$ even by 
    \begin{flalign*}
           ~&\qquad\qquad\qquad   \mathfrak{S}(q_{m,n}) = (p_-\,z_+\,z_-\,p_+)^{\mu} (p_-\,z_+)^{\nu}\,p_- \,\hat{z}_+\,p_-(z_+\,p_-)^{\nu}(p_+\,z_-\,z_+\,p_-)^{\mu}, &
    \end{flalign*}
and for $n=2\nu+1$ odd by 
    \begin{flalign*}
           ~&\qquad\qquad\qquad \mathfrak{S}(q_{m,n}) = (p_-\,z_+\,z_-\,p_+)^{\mu}(p_-\,z_+)^{\nu+1}\,\hat{p}_-\,(z_+\,p_-)^{\nu+1}(p_+\,z_-\,z_+\,p_-)^{\mu}. &
    \end{flalign*} 
In both of these formulas, we see that  between each pair of $p_-$'s there is a $z_+$ and vice versa, so the real roots of $Q_{m,n}$ and $Q_{m,n-1}$ are interlaced on the real line. 
Next, assume $m$ is even, then the singularity signature of $q_{m,n}$ is as follows.
For $n=2\nu$ even we have
        \begin{flalign*}
           ~&\qquad\qquad\qquad   \mathfrak{S}(q_{m,n}) = (p_-\,z_+\,z_-\,p_+)^{\mu}\,(z_-\,p_+)^{\nu}\,\hat{z}_-\,(p_+\,z_-)^{\nu}\,(p_+\,z_-\,z_+,\,p_-)^{\mu}, &
        \end{flalign*}        
and for $n=2\nu+1$ odd we have 
        \begin{flalign*}
           ~&\qquad\qquad\qquad \mathfrak{S}(q_{m,n}) =  (p_-\,z_+\,z_-\,p_+)^{\mu}\,(z_-\,p_+)^{\nu}\,z_-\,\hat{p}_+\, z_-\,(p_+\,z_-)^{\nu}\,(p_+\,z_-\,z_+\,p_-)^{\mu}. &
        \end{flalign*}
Noting that the apparent singularities $\hat{z}_-$ and $\hat{p}_+$ in the above formulas occur at the origin, it follows that $p_-$ and $z_+$ singularities alternate each other on the positive and on the negative real line. This means that the real roots of $Q_{m,n}$ and $Q_{m,n-1}$ are interlaced on the positive and on the negative real line. The other cases are proven analogously.
\end{proof}

See Figure \ref{fig:zeroestable} for a visual representation of the directions in the $(m,n)$ plane along which interlacing of real roots occurs as specified in the above corollary.

\begin{remark}
    In the process of proving the above theorems, some of the topological arguments allow us to obtain the following facts about the possible singularity signatures of general real solutions.
    Firstly, for real parameters $a$ in regions $\mathrm{I}$, $\mathrm{II}$ and $\mathrm{III}$, the singularity signature of any real solution is made of symbols taken alternatively out of $\{p_+,z_+\}$ and $\{p_-,z_-\}$. For example, the substring $p_+ \,p_-$ occurs in the singularity signatures of some Okamoto rationals, but $p_+\,z_+$ is not permitted for real solutions with $a$ in region $\mathrm{I}$.
We can also obtain facts about triples of symbols that can occur in the signatures of real solutions for different regions. 
For real parameters $a$ in regions $\mathrm{I}$ or  $\mathrm{II}$, any real solution of $\pain{IV}$ cannot have either of the following two triples of consecutive apparent singularities in its signature,
    \begin{equation*}
        z_-\, z_+\, z_-,\qquad p_-\, p_+\, p_-.
    \end{equation*}
Results regarding triples which cannot occur in signatures of solutions for parameters in other regions can be obtained along the same lines.
\end{remark}



\subsection{Comparison with the literature}
In \cite{twiton2}, the authors developed methods to classify real solutions of $\pain{IV}$ for generic real parameters, according to their asymptotic behaviour and what we call singularity signature.
For a solution $q$ with known singularity signature, the authors also gave a procedure for deducing that of a solution $w \,q$ related by a B\"acklund transformation corresponding to an element $w$ of the extended affine Weyl group $\widetilde{W}(A_2^{(1)})\cong W(A_2^{(1)}) \rtimes \Z_3$, subject to an assumption that the asymptotic behaviour of $q$ as $t\rightarrow \pm \infty$ is that of type $C$ as defined in \cite{twiton1,twiton2}.
This procedure was given using generators denoted by $\tau$, $\sigma$ in \cite{twiton2}, which correspond to a reflection in a simple root and a rotation of the Dynkin diagram respectively. 
The translations $T_1$, $T_2$ are given in terms of these generators by 
\begin{equation*}
    T_1 = 
    \tau \sigma \tau \sigma, 
    \quad T_2 
    = \sigma \tau \sigma^2 \tau \sigma
\end{equation*}

We note that the asymptotic behaviours of the Okamoto rational solutions are of type $C$, so by composing the actions of $\tau$ and $\sigma$ on singularity signatures as given in \cite{twiton2} one could perform the inductive steps required in the proofs of the formulas for $\mathfrak{S}(q_{m,n})$.

Furthermore, it can be shown that our formulas for $\mathfrak{S}(q_{m,n})$ persist for a whole two-parameter family of solutions, for any fixed generic parameters, from results in \cite{twiton2}.
Namely, denote the space of real parameter values for the system \eqref{eq:systfg} by
\begin{equation*}
    V \defeq \left\{a = (a_0,a_1,a_2)\in \mathbb{R}^3\, :\, a_0+a_1+a_2=1 \right\},
\end{equation*}
and consider the triangular lattice generated by $T_1$ and $T_2$, starting from the triangle
\begin{equation*}
    \triangle:=\{(a_0,a_1,a_2)\in V \,:\, a_0\geq 0, a_1\geq 0, a_2 \geq 0\}.
\end{equation*}
Then, one of the main results of \cite{twiton2} says that, for any generic parameters $a$, that is $a_0,a_1,a_2\not \in \Z$, there exists
a unique finite singularity signature for which $\pain{IV}$ has a two-parameter family of real solutions with this singularity signature. 
Furthermore, this singularity signature is constant as parameters are varied within any face of the triangular lattice and equal to the singularity signature of the Okamoto rational that lies at the center of that face.

\section{Tools} \label{sec:tools}

In this section we review the construction of Okamoto's space of initial conditions for the system \eqref{eq:systfg} and the accompanying suite of tools which we will use in the proofs of our results.



\subsection{Okamoto's space}

In the seminal paper \cite{OKAMOTO1979}, Okamoto constructed an augmented phase space for a polynomial Hamiltonian form of each Painlev\'e equation on which its singularities are resolved and in which all solutions stay.
Later, Takano and collaborators showed that for $\pain{II}$-$\pain{VI}$ the space characterises the equation as the unique  globally regular meromorphic Hamiltonian system on the relevant space \cite{takano1,takano2,takano3} (see \cite{chiba, iwasakiokada} for the case of $\pain{I}$).
Okamoto's space will be our main tool in the proofs of our main results so we review its construction now for the system \eqref{eq:systfg}.

Let $B = \C_t$ be the independent variable space for system \eqref{eq:systfg} (on which the coefficients are analytic), so the phase space for system \eqref{eq:systfg} can be taken initially as the trivial bundle $\C^2_{f,g} \times B$ over $B$, where here and later subscripts indicate coordinates.
The system \eqref{eq:systfg} defines regular initial value problems everywhere on $\C^2_{f,g} \times B$, but the fact that solutions can develop movable poles means that we cannot analytically continue all solutions globally in this bundle. 
The construction of Okamoto's space proceeds via compactification of the fibres, blow-ups, then the removal of certain curves. 
The result is a complex analytic fibre bundle $\pi : M \rightarrow B$ such that $M$ contains $\C^2_{f,g} \times B$ as a subbundle and the flow of the differential system extended to $M$ defines a \emph{uniform foliation} or \emph{complete Ehresmann connection}\footnote{We thank N. Nikolaev for pointing out the relation between terminologies of uniform foliations and complete Ehresmann connections.}, i.e.
\begin{itemize}
\item the foliation of $M$ is nonsingular and each leaf is transverse to the fibres,
\item each leaf intersects the subbundle $\C^2_{f,g} \times B$, and
\item for any point $p_0$ in the fibre over $t_*$, i.e $\pi(p_0)=t_* \in B$, and any path $\ell$ in $B$ with starting point $t_*$, the solution $p(t)$ of the system satisfying $p(t_*) = p_0$ can be holomorphically continued in $M$ over $\ell$.
\end{itemize}
Each fibre parametrises the set of solutions of the equation and is called the \emph{initial value space} or \emph{space of initial conditions}.

To construct the space for system \eqref{eq:systfg} we begin by compactifying $\C^2_{f,g}$ to $\p^1 \times \p^1$, where $\p^1 = \p^1(\C)$, but remark that other choices of compactification are possible, e.g. $\p^2$ or a Hirzebruch surface as initially used by Okamoto. 
We introduce coordinates to cover $\p^1 \times \p^1$ by the four charts 
\begin{equation*}
\p^1 \times \p^1 = \C^2_{f,g} \cup \C^2_{F,g}\cup \C^2_{f,G} \cup \C^2_{F,G},
\end{equation*}
with gluing defined by $F = 1/f$, $G=1/g$. 
Extending the system \eqref{eq:systfg} to $(\p^1 \times\p^1) \times B$ is done by using the gluing as a change of variables, so the system is given in each of the charts as follows:
\begin{equation*}
\begin{aligned}
(f,g) ~:~ &\left\{ 
	\begin{aligned}
		 f' &= 2 f g -f^2 - 2 t f - 2 a_1,  \\
		g'&= 2 f g - g^2 + 2 t g + 2 a_2,
 	\end{aligned}
\right.
&\quad
(F,g) &~:~ \left\{ 
	\begin{aligned}
		 F' &= - 2 F g + 1 + 2 t F + 2 a_1 F^2, \\
		 g' &=  \frac{2 g}{F} - g^2 + 2 t g + 2 a_2, 
	\end{aligned}
\right.
\\
(f,G) ~:~ &\left\{ 
	\begin{aligned}
 		f' &= \frac{2 f}{G} - f^2 - 2 t f - 2 a_1, \\
		 G' &= - 2 f G + 1 - 2 t G - 2 a_2 G^2, 
	\end{aligned}
\right.
&\quad
(F,G) &~:~ \left\{ 
	\begin{aligned}
 		F' &= - \frac{2 F}{G} + 1 + 2 t F + 2 a_1 F^2, \\
 		G' &= - \frac{2G }{F} +1 - 2 t G - 2 a_2 G^2. 
	\end{aligned}
\right.
\end{aligned}
\end{equation*}
The rational vector field on $(\p^1 \times\p^1) \times B$ corresponding to the system has indeterminacies at the points in the fibre given in coordinates by
\begin{equation*}
(F,g) = (0,0), \qquad (f,G) = (0,0), \qquad (F,G) = (0,0),
\end{equation*}
which correspond to points where the infinite families of series solutions cross at the same movable singularity, namely $p_+$, $z_+$ and $p_-$ respectively from Lemma \ref{lem:singularities_expansions}.

We will resolve these indeterminacies and separate the solutions passing through the same point in the fibre using the blow-up technique.
In introducing coordinate charts to cover the exceptional divisor arising from each blowup, we use the following convention:
after blowing up a point $b_i $ given in some chart $(x,y)$ by
\begin{equation*}
b_i : (x,y) = (x_*, y_*),
\end{equation*}
the exceptional divisor $E_i \cong \mathbb{P}^1$ replacing $b_i$ is covered by two coordinate charts $(u_i,v_i)$ and $(U_i,V_i)$ given by 
\begin{equation*}
\begin{aligned}
x&= u_i v_i+ x_{*}  , 			& y &= v_i + y_{*}, \\
u_i &= \frac{x- x_{*}}{y- y_{*} },	 & v_i  &=y-y_{*},
\end{aligned}
\qquad \text{and} \qquad 
\begin{aligned}
x &= V_i+ x_{*}, 			& y&= U_i V_i +y_{*},\\
U_i &= \frac{y-y_*}{x-x_*},		& V_i &= x - x_*.
\end{aligned}
\end{equation*}
In particular in these charts the exceptional divisor $E_i$ has local equation $v_i=0$, respectively $V_i=0$.

It takes eight blowups of the fibre over $t \in B$ to resolve these indeterminacies, of points $b_i$ given (in coordinates introduced according to the convention above) in Figure \ref{fig:points}. 
We denote the exceptional divisor of the blowup of $b_i$, or more precisely its total transform under any further blowups, by $E_i$, $i=1,\dots,8$.
In Figure \ref{fig:points} we use arrows to indicate when a point lies on an exceptional divisor of a previous blow-up, e.g. $b_1 : (F,g) = (0,0)  \leftarrow b_{2} : (U_1, V_1) = ( -2a_2,0)$ indicates that $b_2$ lies on the exceptional divisor $E_1$, in coordinates $(U_1, V_1)$ defined by $F=V_1$, $g = U_1 V_1$. 

\begin{figure}[htb]
\begin{equation*}
\begin{aligned}
&b_1 : (F,g) = (0,0)  &\leftarrow     &\quad b_{2} : (U_1, V_1) = ( -2a_2,0) \\
&b_3: (f,G) = (0, 0)  &\leftarrow     &\quad b_{4} : (u_3, v_3) = ( 2a_1 ,0)\\
&b_5 : (F,G) = (0,0) &\leftarrow     &\quad b_{6} : (U_5, V_5) = (1,0) \\
&   & &\quad   \uparrow  & \\
&   & &\quad  b_{7} : (u_{6},v_{6})  = (-2 t, 0)  \\
&   & &\quad   \uparrow  & \\
&   & &\quad  b_{8} : (u_{7},v_{7}) = \left( 2(1 - a_1 - a_2 + 2 t^2) , 0\right)  \\
\end{aligned}
\end{equation*}
\caption{Blow-up points for Okamoto's space for system \eqref{eq:systfg}}
\label{fig:points}
\end{figure}

After these eight blow-ups of the $\p^1 \times \p^1$ fibre over $t\in B$ we obtain a complex rational surface which we denote $X_t$. 
We give a schematic representation of the configuration of these points and the resulting surface in Figure \ref{fig:surfaceokamoto}, where $E_i$ indicates the exceptional divisor arising from the blow-up of $b_i$ (or more precisely its proper transform or strict transform under any further blowups) and $D_i$ are the irreducible components of the unique effective anticanonical divisor of $X_t$.
The vector field diverges on these components and they are inaccessible to the flow of the system, so are called \emph{inaccessible divisors} or \emph{vertical leaves}.
We remove the inaccessible divisors from each fibre to arrive at the space $M$, with fibre over $t\in B$ being $M_t = X_t \backslash \cup_i D_i$, on which the system of differential equations extended from \eqref{eq:systfg} is regular. 
The flow of the system then defines a uniform foliation of $M$, and in particular the families of solutions with movable singularities at the same $t=t_*$ are separated once lifted to $M$.

\begin{figure}[htb]
\centering
	\begin{tikzpicture}[scale=.85,>=stealth,basept/.style={circle, draw=red!100, fill=red!100, thick, inner sep=0pt,minimum size=1.2mm}]
		\begin{scope}[xshift = -4cm]
			\draw [black, line width = 1pt] 	(4.1,2.5) 	-- (-0.5,2.5)	node [left]  {$g=\infty$} node[pos=0, right] {};
			\draw [black, line width = 1pt] 	(0,3) -- (0,-1)			node [below] {$f=0$}  node[pos=0, above, xshift=-7pt] {} ;
			\draw [black, line width = 1pt] 	(3.6,3) -- (3.6,-1)		node [below]  {$f=\infty$} node[pos=0, above, xshift=7pt] {};
			\draw [black, line width = 1pt] 	(4.1,-.5) 	-- (-0.5,-0.5)	node [left]  {$g=0$} node[pos=0, right] {};

			\node (p1) at (3.6,-.5) [basept,label={[xshift=-10pt, yshift = 0 pt] $b_{1}$}] {};
			\node (p2) at (4.3,0.3) [basept,label={[xshift=10pt, yshift = -10 pt] $b_{2}$}] {};
			\node (p3) at (0,2.5) [basept,label={[yshift=-20pt, xshift=+10pt] $b_{3}$}] {};
			\node (p4) at (0.65,3.3) [basept,label={[xshift=0pt, yshift = 0 pt] $b_{4}$}] {};
			\node (p5) at (3.6,2.5) [basept,label={[xshift=-10pt,yshift=0pt] $b_{5}$}] {};
			\node (p6) at (4.1,3.3) [basept,label={[xshift=0pt, yshift = 0 pt] $b_{6}$}] {};
			\node (p7) at (4.8,3.3) [basept,label={[xshift=0pt, yshift = 0 pt] $b_{7}$}] {};
			\node (p8) at (5.5,3.3) [basept,label={[xshift=0pt, yshift = 0 pt] $b_{8}$}] {};
			\draw [line width = 0.8pt, ->] (p2) -- (p1);
			\draw [line width = 0.8pt, ->] (p4) -- (p3);
			\draw [line width = 0.8pt, ->] (p6) -- (p5);
			\draw [line width = 0.8pt, ->] (p7) -- (p6);
			\draw [line width = 0.8pt, ->] (p8) -- (p7);
			\node (P1P1) at (1.8, -1.5) [below] {$\p^1 \times \p^1$};
		\end{scope}
	
		\draw [->] (3.75,1.5)--(1.75,1.5) node[pos=0.5, below] {$\text{Bl}_{b_1\dots b_8}$};
	
		\begin{scope}[xshift = 6.5cm, yshift= .5cm]
			\draw [blue, line width = 1pt] 	(2.8,2.5) 	-- (-0.2,2.5)	node [pos = .5, below]  {
			$D_4$
			} node[pos=0, right] {};
			\draw [blue, line width = 1pt] 	(3.6,1.7) -- (3.6,-1)		node [pos = .5, left]  {
			$D_2$
			} node[pos=0, above, xshift=7pt] {};

			\draw [blue, line width = 1pt] 	(-1.3,0.9) 	-- (.7, 2.9)	 node [left] {} node[pos = 0, left] {
			$D_5$
			};
			\draw [red, line width = 1pt] 	(-.25,1.55) 	-- (-1.3,2.6)	 node [left] {} node[pos=1, left] {$E_4$};
			
			\draw [blue, line width = 1pt] 	(4,-.3) node[left]{}	-- (2,-2.3)	 node [below left] {} node[below] {
			$D_1$
			};

			\draw [red, line width = 1pt] 	(2.8,-1.1) 	-- (3.8,-2.1)	 node [left] {} node[pos=1, right] {$E_2$};
			
			\draw [blue, line width = 1pt]	(4,.9) -- (2,2.9) node[ pos=0, right] {
			$D_3$
			};
			\draw [blue, line width = 1pt]	(2.8,1.7) -- (4.2,3.1) node [right] {
			$D_6$
			} ;
			\draw [blue, line width = 1pt]	(4.2,2.7) -- (2.8,4.1) node [above] {
			$D_0$
			};
			\draw [red, line width = 1pt]	(2.8,3.7) -- (3.8,4.7)  node [below right] {$E_8$};

			\draw [black, dotted, line width = 1pt]  (-1,1.6) -- (-1,-2.3);
			\draw [black, dotted, line width = 1pt]  (-1.25,-2) 	-- (2.7,-2);

%

		\end{scope}
	\end{tikzpicture}
	\caption{Surface $X_t$ for the initial value space of system \eqref{eq:systfg}}
	\label{fig:surfaceokamoto}
\end{figure}
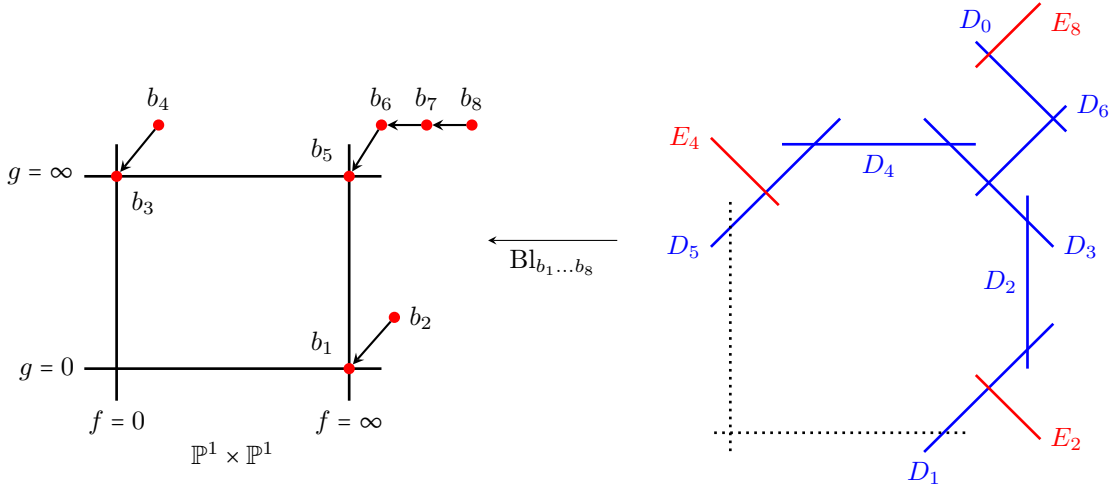

\subsection{Apparent singularities and exceptional curve crossings}

On the space of initial conditions, the movable singularities of solutions of $\pain{IV}$ discussed in Subsection \ref{subsec:apparentsingularities} correspond to crossings of the exceptional divisors $E_2$, $E_4$ and $E_8$.
These arise from the last in the sequences of blowups required to resolve $b_1$, $b_3$ and $b_5$ respectively, and are not removed from the fibre in the last step in the construction of $M$ above.

The coordinate charts covering the parts of $E_2$, $E_4$ and $E_8$ away from the inaccessible divisors $D_i$ which remain in the fibre $M_t$ are given as follows:
\begin{itemize}
    \item The part of $E_2$ away from $D_1$ is covered by coordinates $(u_2,v_2)$ defined by
        \begin{equation*}
        f = \frac{1}{v_2}, \quad g = v_2 ( -2 a_2 + u_2 v_2),
        \end{equation*}
    in which $E_2$ has local equation $v_2 = 0$.
   
   \item The part of $E_4$ away from $D_5$ is covered by coordinates $(u_4,v_4)$ defined by
        \begin{equation*}
        f = v_4 (2 a_1 + u_4 v_4), \quad g = \frac{1}{v_4},
        \end{equation*}
    in which $E_4$ has local equation $v_4 = 0$.
   
   \item The part of $E_8$ away from $D_0$ is covered by coordinates $(u_8,v_8)$ defined by
        \begin{equation*}
        f = \frac{1}{v_8}, \quad g = \frac{1}{v_8(1- 2 t v_8 + 2(1 + 2 t^2 - a_1 - a_2)v_8^2 + u_8 v_8^3 },
        \end{equation*}
    and $E_8$ has local equation $v_8 = 0$.
\end{itemize}

Lifting solutions $(f(t),g(t))$ of the system \eqref{eq:systfg} with expansions about movable singularities as in Lemma \ref{lem:singularities_expansions} to the space $M$, we have the following:

\begin{lemma} \label{lem:resonantparametersexceptionallinecrossings}
    \begin{enumerate}[(a)]
        \item If a solution $(f(t),g(t))$ of the system \eqref{eq:systfg} has a plus pole at $t=t_*$, then $(f(t_*),g(t_*))$ lies on the exceptional curve $E_2$, given in coordinates by
        \begin{equation*}
            u_2(t_*) = 2 \eta - t_*, \qquad v_2(t_*) = 0. 
        \end{equation*}

        \item If a solution $(f(t),g(t))$ of the system \eqref{eq:systfg} has a minus pole at $t=t_*$, then $(f(t_*),g(t_*))$ lies on the exceptional curve $E_8$, given in coordinates by
        \begin{equation*}
            u_8(t_*) = 2 \eta + t_* (7 + 8 t_*^2- 8 a_1 - 8 a_2), \qquad v_8(t_*) = 0. 
        \end{equation*}

        \item If a solution $(f(t),g(t))$ of the system \eqref{eq:systfg} has a plus zero at $t=t_*$, then $(f(t_*),g(t_*))$ lies on the exceptional curve $E_4$, given in coordinates by
        \begin{equation*}
            u_4(t_*) = \eta + 2 t_* a_1, \qquad v_4(t_*) = 0. 
        \end{equation*}

        \item If a solution $(f(t),g(t))$ of the system \eqref{eq:systfg} has a minus zero at $t=t_*$, then $(f(t_*),g(t_*))$ lies on the proper transform of the curve ${f=0}$, given in coordinates by
        \begin{equation*}
            f(t_*) = 0, \qquad g(t_*) = t_* - \frac{\eta}{2a_1}. 
        \end{equation*}        
    \end{enumerate}
\end{lemma}

\subsection{Translations as isomorphisms of surfaces}

Each of the translations $T_1$, $T_2$ used to generate the hierarchy of Okamoto rational solutions is provided by a transformation of the variables $(f,g)$ which becomes an isomorphism when lifted under blowups to the space of initial conditions.
These transformations are as follows.

\begin{align}
&T_1 : \left\{ 
    \begin{aligned}
        a=(a_0,a_1,a_2)&\mapsto \widehat{a}=(a_0-1,a_1+1,a_2) \\
        (f,g)&\mapsto (\widehat{f},\widehat{g})
    \end{aligned}
 \qquad 
    \begin{aligned}
\widehat{f} &= -2 t + \widehat{g} - \frac{2 (a_1 + a_2)}{f + \widehat{g}}, \\
\widehat{g} &= - \frac{f(f+g+2a_2)}{f g - 2 a_1},
    \end{aligned}
\right. \label{BTT1}
\\
&T_2 : \left\{ 
    \begin{aligned}
        a=(a_0,a_1,a_2)&\mapsto \widetilde{a}=(a_0-1,a_1,a_2+1) \\
        (f,g) &\mapsto (\widetilde{f},\widetilde{g})
    \end{aligned}
 \qquad 
    \begin{aligned}
\widetilde{f} &= - \frac{g(f g - 2a_1)}{f g + 2 a_2}, \\
\widetilde{g} &= 2 t + \widetilde{f} + \frac{2(a_1+a_2)}{\widetilde{f} +g}.
    \end{aligned}
\right. \label{BTT2}
\end{align}

Denote the compact surface obtained by the eight blowups as above with parameters $a = (a_0,a_1,a_2)$ and $t$ fixed by $X_{t;a}$.
Then regarding $(\widehat{f},\widehat{g})$ as affine coordinates for $X_{t;\widehat{a}}$ corresponding to $(f,g)$, the transformation $(f,g)\mapsto (\widehat{f},\widehat{g})$ becomes an isomorphism when extended to $\p^1\times\p^1$ then lifted under the blowups: 
\begin{equation*}
    \begin{aligned}
        \varphi_{T_1} : X_{t;a} &\rightarrow X_{t;\widehat{a}},  \\
                    (f,g) &\mapsto (\widehat{f},\widehat{g}),
    \end{aligned}
\end{equation*}
and similarly from $T_2$ we have the isomorphism
\begin{equation*}
    \begin{aligned}
        \varphi_{T_2} : X_{t;a} &\rightarrow X_{t;\widetilde{a}},  \\
                    (f,g) &\mapsto (\widetilde{f},\widetilde{g}).
    \end{aligned}
\end{equation*}
The movement of exceptional divisors under these isomorphisms will be used extensively in what follows and it is convenient to describe this in terms of the maps induced on the Picard groups of the surfaces.

The Picard group of the smooth rational surface $X_{t;a}$ is isomorphic to the group of divisors modulo linear equivalence, so it can be written as 
\begin{equation} \label{eq:picdef}
\Pic(X_{t;a}) = \Z \h_f + \Z \h_g + \Z \E_1 + \Z \E_2 + \Z \E_3 + \Z \E_4 + \Z \E_5 + \Z \E_6 + \Z \E_7 + \Z \E_8,
\end{equation}
where $\h_f$ and $\h_g$ correspond to classes of total transforms of divisors of constant $f$ and $g$ respectively, and for $i=1,\dots,8$, $\E_i$ is the class of the exceptional divisor $E_i$. 
The Picard group is equipped with the symmetric bilinear form corresponding to intersection numbers of divisors, which we denote $~\cdot~$,
defined by
\begin{equation*}
\h_f \cdot \h_g = 1, \quad \h_f \cdot \h_f = \h_g \cdot \h_g = \h_f \cdot \E_i = \h_g \cdot \E_i = 0, \quad \E_i \cdot \E_j = - \delta_{ij}, \quad i,j =1,\dots,8.
\end{equation*}

The isomorphisms $\varphi_{T_1}$ and $\varphi_{T_2}$ induce, by pushforward and pullback, transformations
\begin{equation*}
T_1 = (\varphi_{T_1})_* : \Pic(X_{t;a}) \rightarrow \Pic(X_{t;\widehat{a}}), \quad T_1^{-1} = (\varphi_{T_1})^* : \Pic(X_{t;\widehat{a}}) \rightarrow \Pic(X_{a,t}),
\end{equation*}
and
\begin{equation*}
T_2 = (\varphi_{T_2})_* : \Pic(X_{t;a}) \rightarrow \Pic(X_{t;\widetilde{a}}), \quad T_2^{-1} = (\varphi_{T_2})^* : \Pic(X_{t;\widetilde{a}}) \rightarrow \Pic(X_{t;a}),
\end{equation*}
which are $\Z$-linear and preserve the intersection form.
These will be given explicitly below in terms of the generators of the Picard group as in \eqref{eq:picdef}.

\subsection{Movement of exceptional divisors under translations}

The main building blocks of our proofs in the following sections will be relations between apparent singularities of solutions related by the B\"acklund transformations $T_1$ and $T_2$, which we deduce from the movement of exceptional divisors on the surfaces under these translations.

For example, the action of $T_1$ on the Picard group is given by
\begin{equation*}
    T_1 = (\varphi_{T_1})_*: \left\{ 
    \begin{aligned}
        \h_f &\mapsto \widehat{\h}_f +2 \widehat{\h}_g - \widehat{\E}_{5678}, \\
        \h_y &\mapsto 2\widehat{\h}_f + 3\widehat{\h}_g - \widehat{\E}_{12556678}, \\
        \E_1 &\mapsto \widehat{\h}_f + \widehat{\h}_g - \widehat{\E}_{568}, \\
        \E_2 &\mapsto \widehat{\h}_f + \widehat{\h}_g - \widehat{\E}_{567}, 
    \end{aligned}
    \qquad
        \begin{aligned}
        \E_3 &\mapsto \widehat{\h}_f + 2\widehat{\h}_g - \widehat{\E}_{25678}, \\
        \E_4 &\mapsto \widehat{\h}_f + 2\widehat{\h}_g - \widehat{\E}_{15678}, \\
        \E_5 &\mapsto \widehat{\h}_g - \widehat{\E}_{6}, \\
        \E_6 &\mapsto \widehat{\h}_g - \widehat{\E}_{5}, 
        \end{aligned}
        \qquad
     \begin{aligned}
        \E_7 &\mapsto \widehat{\E}_{3}, \\
        \E_8 &\mapsto \widehat{\E}_{4},
    \end{aligned}
    \right.
\end{equation*}
where $\widehat{\h}_f$, $\widehat{\h}_g$, and $\widehat{\E}_k$, $k=1,\dots,8$ are the generators of $\Pic(X_{t;\widehat{a}})$ defined in the natural way in terms of the blowups of $\p^1 \times \p^1$ at the points $b_1,\dots,b_8$ but with parameters $\widehat{a}$ in place of $a$. 
Here and from this point onwards we use shorthand for sums along the lines of $\E_{ijk} = \E_i+\E_j+\E_k$.

As a first illustration of the ways in which we will use $T_1$ and $T_2$ as maps on Picard groups to extract information about singularity signatures of solutions, consider the action of $T_1$ on the class $\E_8$:
\begin{equation*}
T_1(\E_8) = \widehat{\E}_{4}.
\end{equation*}
We note that since $E_4$ and $E_8$ are exceptional curves of the first kind, this implies in particular that $\varphi_{T_1}$ gives an isomorphism $E_8  \rightarrow E_4$.
Since $E_4$ and $E_8$ are crossed by solutions when they have a plus zero and minus pole respectively, this tells us that $(f(t),g(t))$ has a minus pole at $t=t_*$ if and only if $(\widehat{f}(t),\widehat{g}(t))$ defined by the B\"acklund transformation 
\eqref{BTT1} has a plus zero at $t=t_*$.

On the other hand if we consider the action on the class $\E_4$, we see that $E_4$ is sent by $\varphi_{T_1}$ to the curve on $X_{t;\widehat{a}}$ with class 
\begin{equation*}
T_1(\E_4) = \widehat{\h}_f + 2\widehat{\h}_g - \widehat{\E}_{15678}.
\end{equation*}
This curve is also exceptional of the first kind, and its equation in the coordinates $(\widehat{f},\widehat{g})$ can be computed to be 
\begin{equation} \label{eq:imagecurveT1E4}
\widehat{g}^2 - \widehat{g}( \widehat{f}+ 2 t) + 2(1- \widehat{a}_1-\widehat{a}_2) = 0.
\end{equation}
Therefore we deduce that $(f(t),g(t))$ has a plus zero at $t=t_*$ if and only $(\widehat{f}(t_*),\widehat{g}(t_*))$ lie in the curve on $X_{t_*;\widehat{a}}$ defined by \eqref{eq:imagecurveT1E4}.
In particular $(\widehat{f}(t_*),\widehat{g}(t_*))$ could lie on the part of this curve visible in the affine $(\widehat{f},\widehat{g})$ chart, or at its intersection with $\widehat{E}_8$ since $T_1(\E_4) \cdot \widehat{\E}_8=1$. 
So $\widehat{f}(t_*),\widehat{g}(t_*)$ are either finite satisfying \eqref{eq:imagecurveT1E4} with $t=t_*$, or $(\widehat{f}(t),\widehat{g}(t))$ has a minus pole at $t=t_*$.
In the latter case we can determine the special kind of minus pole that arises by computing the intersection of the curve defined by \eqref{eq:imagecurveT1E4} and the exceptional divisor $\widehat{E}_8$.
This is given in coordinates $\widehat{u}_8,\widehat{v}_8$ on $X_{t;\widehat{a}}$ covering $\widehat{E}_8$, by 
\begin{equation*}
\widehat{u}_8 = -4 t(3 +2 t^2 - 3 \widehat{a}_1 - 3 \widehat{a}_2), \qquad \widehat{v}_8 = 0.
\end{equation*}
We can then compute the precise value of the resonant parameter $\eta$ in the expansion of $(\widehat{f}(t), \widehat{g}(t))$ about this minus pole as in \eqref{lem:resonantparametersexceptionallinecrossings}.

Below we collect all such curves coming from the actions of $T_1$ and $T_2$ on exceptional divisors corresponding to apparent singularities.
\subsubsection{Curves related to $T_1$}
The maps $T_1$ and $T_1^{-1}$ are given by
\begin{equation*}
    T_1 = (\varphi_{T_1})_*: \left\{ 
    \begin{aligned}
        \h_f &\mapsto \widehat{\h}_f +2 \widehat{\h}_g - \widehat{\E}_{5678}, \\
        \h_g &\mapsto 2\widehat{\h}_f + 3\widehat{\h}_g - \widehat{\E}_{12556678}, \\
        \E_1 &\mapsto \widehat{\h}_f + \widehat{\h}_g - \widehat{\E}_{568}, \\
        \E_2 &\mapsto \widehat{\h}_f + \widehat{\h}_g - \widehat{\E}_{567}, \\
        \E_3 &\mapsto \widehat{\h}_f + 2\widehat{\h}_g - \widehat{\E}_{25678}, \\
        \E_4 &\mapsto \widehat{\h}_f + 2\widehat{\h}_g - \widehat{\E}_{15678}, \\
        \E_5 &\mapsto \widehat{\h}_g - \widehat{\E}_{6}, \\
        \E_6 &\mapsto \widehat{\h}_g - \widehat{\E}_{5}, \\
        \E_7 &\mapsto \widehat{\E}_{3}, \\
        \E_8 &\mapsto \widehat{\E}_{4},
    \end{aligned}
    \right.
    \quad
    T_1^{-1} = (\varphi_{T_1})^*: \left\{ 
    \begin{aligned}
        \widehat{\h}_f &\mapsto 3\h_f +2 \h_g - \E_{12334456}, \\
        \widehat{\h}_g &\mapsto 2\h_f + \h_g - \E_{1234}, \\
        \widehat{\E}_1 &\mapsto \h_f - \E_{4}, \\
        \widehat{\E}_2 &\mapsto \h_f - \E_{3}, \\
        \widehat{\E}_3 &\mapsto \E_7, \\
        \widehat{\E}_4 &\mapsto \E_8, \\
        \widehat{\E}_5 &\mapsto 2\h_f + \h_g - \E_{12346}, \\
        \widehat{\E}_6 &\mapsto 2\h_f + \h_g - \E_{12345}, \\
        \widehat{\E}_7 &\mapsto \h_f + \h_g - \E_{234}, \\
        \widehat{\E}_8 &\mapsto \h_f + \h_g - \E_{134}.
    \end{aligned}
    \right.
\end{equation*}
Similarly to above, from this we can deduce the behaviour of apparent singularities under the translation, which we present below. 
When the translation sends an apparent singularity to a crossing of a curve visible in the $(\widehat{f},\widehat{g})$-chart for $X_{t;\widehat{a}}$,
 we present the defining equation for the corresponding curve on $X_{t;a}$ in the coordinates $f,g$, for simplicity of notation.
%


\begin{lemma}\label{lem:behaviourT1}
The behaviour of apparent singularities under $T_1$ is as follows:
\begin{equation*}
   T_1 : \left\{
    \begin{aligned}
        p_+ &\mapsto C_{T_1(p_+)} :
        f - g + 2t = 0, \\
        p_- &\mapsto z_+, \\
        z_+ &\mapsto C_{T_1(z_+)} :
        g^2 - g( f+ 2 t) + 2a_0= 0, \\
        z_- &\mapsto p_+.
    \end{aligned}
    \right.
\end{equation*}
Similarly the behaviour of apparent singularities under $T_1^{-1}$ is as follows:
\begin{equation*}
   T_1^{-1} : \left\{
    \begin{aligned}
        p_+ &\mapsto z_-, \\
        p_- &\mapsto C_{T_1^{-1}(p_-)} : 
        f g - 2 a_1=0, \\
        z_+ &\mapsto p_-, \\
        z_- &\mapsto C_{T_1^{-1}(z_-)} : 
        f^3 g - f^2 ( g^2 - 2 t g - 2 a_2) + 4 a_1 f ( g - t) - 4 a_1^2 = 0. 
    \end{aligned}
    \right.
\end{equation*}
\end{lemma}

\subsubsection{Curves related to $T_2$}
The maps $T_2$ and $T_2^{-1}$ are given by
\begin{equation*}
    T_2 = (\varphi_{T_2})_*: \left\{ 
    \begin{aligned}
        \h_f &\mapsto 3\widetilde{\h}_f +2 \widetilde{\h}_g - \widetilde{\E}_{34556678}, \\
        \h_g &\mapsto 2\widetilde{\h}_f + \widetilde{\h}_g - \widehat{\E}_{5678}, \\
        \E_1 &\mapsto 2\widetilde{\h}_f + \widetilde{\h}_g - \widetilde{\E}_{45678}, \\
        \E_2 &\mapsto 2\widetilde{\h}_f + \widetilde{\h}_g - \widetilde{\E}_{35678}, \\
        \E_3 &\mapsto \widetilde{\h}_f + \widetilde{\h}_g - \widetilde{\E}_{568}, \\
        \E_4 &\mapsto \widetilde{\h}_f + \widetilde{\h}_g - \widetilde{\E}_{567}, \\
        \E_5 &\mapsto \widetilde{\h}_f - \widetilde{\h}_{6}, \\
        \E_6 &\mapsto \widetilde{\h}_f - \widetilde{\h}_{5}, \\
        \E_7 &\mapsto \widetilde{\h}_{1}, \\
        \E_8 &\mapsto \widetilde{\h}_{2},
    \end{aligned}
    \right.
    \quad
    T_2^{-1} = (\varphi_{T_1})^*: \left\{ 
    \begin{aligned}
        \widetilde{\h}_f &\mapsto \h_f +2 \h_g - \E_{1234}, \\
        \widetilde{\h}_g &\mapsto 2\h_f + 3\h_g - \E_{11223456}, \\
        \widetilde{\E}_1 &\mapsto \E_{7}, \\
        \widetilde{\E}_2 &\mapsto \E_{8}, \\
        \widetilde{\E}_3 &\mapsto \h_g - \E_2, \\
        \widetilde{\E}_4 &\mapsto \h_g - \E_1, \\
        \widetilde{\E}_5 &\mapsto \h_f + 2\h_g - \E_{12346}, \\
        \widetilde{\E}_6 &\mapsto \h_f + 2\h_g - \E_{12345}, \\
        \widetilde{\E}_7 &\mapsto \h_f + \h_g - \E_{124}, \\
        \widetilde{\E}_8 &\mapsto \h_f + \h_g - \E_{123}.
    \end{aligned}
    \right.
\end{equation*}
\begin{lemma} \label{lem:behaviourT2}
The behaviour of apparent singularities under $T_2$ is as follows:
\begin{equation*}
   T_2 : \left\{
    \begin{aligned}
        p_+ &\mapsto C_{T_2(p_+)} : 
        f^2 + f( 2t - g) - 2a_0= 0, \\
        p_- &\mapsto p_+, \\
        z_+ &\mapsto C_{T_2(z_+)} : 
        f - g + 2t = 0, \\
        z_- &\mapsto C_{T_2(z_-)} : 
        f^2 + f(2t - g) + 2 a_1 = 0.
    \end{aligned}
    \right.
\end{equation*}
Similarly the behaviour of apparent singularities under $T_2^{-1}$ is as follows:
\begin{equation*}
   T_2^{-1} : \left\{
    \begin{aligned}
        p_+ &\mapsto p_-, \\
        p_- &\mapsto C_{T_2^{-1}(p_-)} : 
        f g + 2 a_2 = 0, \\
        z_+ &\mapsto C_{T_2^{-1}(z_+)} : g = 0, \\
        z_- &\mapsto C_{T_2^{-1}(z_-)} :  
        f g - 2 a_1 = 0.
    \end{aligned}
    \right.
\end{equation*}
\end{lemma}

The curves in Lemmas \ref{lem:behaviourT1} and \ref{lem:behaviourT2} will form the basis of topological arguments in the following sections. 

\subsection{Behaviour of $q_{m,n}(t)$ at the origin and as $t\to \infty$}

It turns out that some of the significant points in the surface $X_{t;a}$ coincide when $t=0$, so we can describe the behaviour of the generalised Okamoto rationals at the origin exactly based on the movements of these points under $T_1$ and $T_2$.

\begin{definition}[special singularities at the origin]
We define the following special apparent singularities, which solutions may have at $t=0$.
\begin{description}
    \item[$\bullet$ $p_-'$ (\emph{special minus pole at the origin})]
      If a solution $(f(t),f(t))$ of the system \eqref{eq:systfg} has a minus pole at $t=0$ given by the expansion in Lemma \ref{lem:singularities_expansions} with $t_*=0$ and $\eta = 0$, we denote this by $p_-'$. 
     At $t=0$ such a solution passes through the point on the exceptional divisor $E_8$ in the surface $X_{t=0;a}$ given by 
     \begin{equation*}
         (u_8,v_8)=(0,0).
     \end{equation*}  
     
     \item[$\bullet$ $p_+'$ (special plus pole at the origin)]
      If a solution $(f(t),f(t))$ of the system \eqref{eq:systfg} has a plus pole at $t=0$ given by the expansion in Lemma \ref{lem:singularities_expansions} with $t_*=0$ and $\eta = 0$, we denote this by $p_+'$. 
     At $t=0$ such a solution passes through the point on the exceptional divisor $E_2$ in the surface $X_{t=0;a}$ given by 
     \begin{equation*}
         (u_2,v_2)=(0,0).
     \end{equation*}  
     
     \item[$\bullet$ $z_-'$ (special minus zero at the origin)]
     If a solution $(f(t),f(t))$ of the system \eqref{eq:systfg} has a minus zero at $t=0$ given by the expansion in Lemma \ref{lem:singularities_expansions} with $t_*=0$ and $\eta = 0$, we denote this by $z_-'$. 
     At $t=0$ such a solution passes through the point in the surface $X_{t=0;a}$ given by 
     \begin{equation*}
         (f,g)=(0,0).
     \end{equation*}
     
    \item[$\bullet$ $z_+'$ (special plus zero at the origin)]
     If a solution $(f(t),f(t))$ of the system \eqref{eq:systfg} has a plus zero at $t=0$ given by the expansion in Lemma \ref{lem:singularities_expansions} with $t_*=0$ and $\eta = 0$, we denote this by $z_+'$. 
     At $t=0$ such a solution passes through the point on the exceptional divisor $E_4$ in the surface $X_{t=0;a}$ given by 
     \begin{equation*}
         (u_4,v_4)=(0,0).
    \end{equation*}    
\end{description}

\end{definition}

We study the movement of the points corresponding to these special singularities at the origin under the translations when $t=0$ and parameters $a$ are arbitrary. 
Calculations in coordinates using the expressions \eqref{BTT1} and \eqref{BTT2} for the maps $\varphi_{T_1}$ and $\varphi_{T_2}$ show that the movement of points and behaviour of special singularities at the origin is as follows:
\begin{equation} \label{cd:originsingmovement}
\begin{tikzcd}
z_-'  \arrow[r, "T_1"] \arrow[d, "T_2"] 		&  p_+'  \arrow[r, "T_1"] \arrow[d, "T_2"]  		&  z_-'  \arrow[d, "T_2"]  \\
z_+'  \arrow[r, "T_1"] \arrow[d, "T_2"] 	&  p_-'  \arrow[r, "T_1"] \arrow[d, "T_2"]  		&  z_+'   \arrow[d, "T_2"]  \\
z_-'  \arrow[r, "T_1"]					 &  p_+'  \arrow[r, "T_1"] 	&  z_-' 
\end{tikzcd}
\end{equation}
From this we deduce that if a solution $(f,g)$ has a special apparent singularity of type $z_-',z_+',p_-'$ or $p_+'$ at the origin, then so will that obtained from it by applying $T_1^nT_2^m$, of a type which can be determined exactly using \eqref{cd:originsingmovement}.
Regarding the Okamoto rationals, this leads to the following proposition.

\begin{proposition} \label{prop:behaviouratorigin}
   The Okamoto rational $q_{m,n}(t)$ has the following behaviour at the origin $t=0$ :
   \begin{itemize}
        \item If $n$ is even and $m$ is even, a special minus zero $z_-'$: $q_{m,n}(t) = -2 ( \frac{1}{3} + n ) t + \mathcal{O}(t^2)$, which corresponds to the expansions
        \begin{equation*}
			f_{m,n}(t)=  -2 ( n + \tfrac{1}{3}) t + \mathcal{O}(t^3), \qquad g_{m,n}(t)=  2 ( m + \tfrac{1}{3}) t + \mathcal{O}(t^3).
        \end{equation*}
        \item If $n$ is even and $m$ is odd, a special plus zero $z_+'$: $q_{m,n}(t) = 2 ( \frac{1}{3} + n ) t + \mathcal{O}(t^2)$, 
which corresponds to the expansions
        \begin{equation*}
			f_{m,n}(t)=  2 ( n + \tfrac{1}{3}) t + \mathcal{O}(t^3), \qquad g_{m,n}(t)=  + \tfrac{1}{t} + \tfrac{2}{3} (m+2n+2) t + \mathcal{O}(t^3).
        \end{equation*}
        \item  If $n$ is odd and $m$ is even, a special plus pole $p_+'$: $q_{m,n}(t) = + \frac{1}{t} + \mathcal{O}(1)$, which corresponds to the expansions
        \begin{equation*}
			f_{m,n}(t)=  +\tfrac{1}{t} - \tfrac{2}{3}(2m +  n + 2) t + \mathcal{O}(t^3), \qquad g_{m,n}(t)=  -\tfrac{2}{3}(2m+1) t 
			+  \mathcal{O}(t^3).
        \end{equation*}
  
        \item If $n$ is odd and $m$ is odd, a special minus pole $p_-'$: $q_{m,n}(t) = - \frac{1}{t} + \mathcal{O}(1)$, which corresponds to the expansions
        \begin{equation*}
			f_{m,n}(t)=  -\tfrac{1}{t} + \tfrac{2}{3}(2m +  n -2) t + \mathcal{O}(t^3), \qquad g_{m,n}(t)=  - \tfrac{1}{t} -\tfrac{2}{3}(m+2n-2) t 
			+  \mathcal{O}(t^3).
        \end{equation*}

           \end{itemize}

\end{proposition}
\begin{proof}
Considering the seed solution $q_{0,0}(t) = -\tfrac{2}{3}t$, we correspondingly have $(f,g) = (f_{0,0}(t),g_{0,0}(t))$, where
\begin{equation*}
f_{0,0}(t) = - \tfrac{2}{3}t, \qquad g_{0,0}(t) = \tfrac{2}{3}t, 
\end{equation*}
which has a $z_-'$ at the origin.

Now, recalling how special apparent singularities at the origin map under translations $T_1$ and $T_2$, as summarised in equation
\eqref{cd:originsingmovement}, we see that each Okamoto rational has a special apparent singularity at the origin, with type as specified in the proposition. The corresponding expansions around $t=0$ follow directly from the corresponding expansions in Lemma \ref{lem:singularities_expansions}, which concludes the proof of the proposition.
\end{proof}

The behaviour of the generalised Okamoto rationals at the origin established above will act as one boundary condition in our proofs of their singularity signatures, with the other being their behaviour as $t$ approaches infinity, which is as follows:


\begin{lemma} \label{lem:behaviouratinfinity}
    For any $m,n \in \Z$, let the functions $f,g$ defined by the Okamoto rational $q_{m,n}$ according to \eqref{fgtoq} be $f_{m,n},g_{m,n}$. Then these functions have the following asymptotic behaviour as $t\rightarrow \infty$:
    \begin{equation*}
    \begin{aligned}
        f_{m,n}(t)  &= - \frac{2t }{3} - \frac{2m+n}{t} + \frac{3(n^2- 2 m n + n - 2 m^2)}{2 t^3} + \mathcal{O}(t^{-5}), \\
        g_{m,n}(t)  &= \frac{2t }{3} - \frac{m+2n}{t} + \frac{3(2 n^2+ 2 m n - m -  m^2)}{2 t^3} + \mathcal{O}(t^{-5}), \\
    \end{aligned}
    \end{equation*}
\end{lemma}
    \begin{proof}
By definition of the Okamoto rational $q_{m,n}$, see equation \eqref{qrat}, we have
\begin{equation*}
    q_{m,n}=-\tfrac{2}{3}t+\frac{Q_{m-1,n}'}{Q_{m-1,n}}-\frac{Q_{m,n}'}{Q_{m,n}}=-\tfrac{2}{3}t+\mathcal{O}(t^{-1}),
\end{equation*}
as $t\rightarrow \infty$. Since $q_{m,n}$ is meromorphic at $t=\infty$ and an odd function of $t$, see \eqref{eq:okamotorationaloddness}, it follows that $q_{m,n}$ has an asymptotic expansion of the form
\begin{equation*}
    q_{m,n}=-\tfrac{2}{3}t+\frac{u_1}{t}+\frac{u_3}{t^3}+\mathcal{O}(t^{-5}),
\end{equation*}
as $t\rightarrow \infty$. Substitution into the $\pain{IV}$ equation and comparing terms yields $u_1=-(2m+n)$ and $u_3=\tfrac{3}{2}(n^2- 2 m n + n - 2 m^2)$. The expansions in the lemma now follow directly from this under the change of variables \eqref{fgtoq}.
    \end{proof}


\subsection{Transitions between regions and nodal curves}

The reason that we have different formulas and inductive proofs for the singularity signatures of the Okamoto rationals for the regions I-VI in the $(m,n)$-plane as in Figure \ref{fig:regions} is that topological changes occur in certain curves relevant to the proofs when passing from one region to another.
These changes occur when the parameters $(a_0,a_1,a_2)$ are such that nodal curves appear on the surface which are not components of its unique effective anticanonical divisor, which we explain now.

The surface $X_{t;a}$ is a \emph{generalised Halphen surface}, which was defined by Sakai \cite{SAKAI2001} as a smooth complex projective rational surface $X$ which has an effective anticanonical divisor $D \in | \mathcal{K}_X|$ of canonical type, i.e if $D = \sum_{i} m_i D_i$, $m_i > 0$, is its decomposition into irreducible components then $[D_i] \cdot \mathcal{K}_X = 0$ for all $i$.

With the parameter normalisation $a_0 + a_1 + a_2 = 1$, the surface $X = X_{t;a}$ has a unique effective anticanonical divisor
\begin{equation*}
D = D_0 + D_1 + 2 D_2 + 3 D_3 + 2 D_4 + D_5 + 2 D_6,
\end{equation*}
where the irreducible components are the inaccessible divisors $D_i$ which were removed from the fibre in the construction of Okamoto's space.

The classes $\delta_i=[D_i]\in \Pic(X)$ of the components are given by 
\begin{equation*}
\begin{aligned}
\delta_0 &= \E_7 - \E_8, 			&\quad &\delta_1 = \E_1 - \E_2,  &\quad &\delta_2 = \h_f - \E_{15}, &\quad &\delta_3 = \E_5 - \E_6, \\
\delta_4 &= \h_g - \E_{35}, 		&\quad &\delta_5 = \E_3 - \E_4,  &\quad &\delta_6 = \E_6 - \E_7,
\end{aligned}
\end{equation*} 
so we have the following expression for the anticanonical class 
\begin{equation*}
- \mathcal{K}_{X} = \delta_0 + \delta_1 + 2 \delta_2 + 3 \delta_3 + 2 \delta_4 + \delta_5 + 2 \delta_6 = 2 \h_f + 2 \h_g - \E_{12345678}.
\end{equation*}
The span $Q = \Z \delta_0 + \Z \delta_1 + \dots + \Z \delta_6$ of the classes of the components of $D$ 
is isomorphic, when equipped with the symmetric bilinear form $(\F_1,\F_2) = - \F_1 \cdot \F_2$, to the root lattice of an affine root system of type $E_6^{(1)}$. 
In particular, $\delta_i\cdot \delta_i=-2$ for all $i$ and the curves $D_i$ are of self-intersection $-2$.

The orthogonal complement $Q^{\perp} = \left\{ \F \in \Pic(X) ~|~ \F \cdot \delta_i = 0, \text{ for all } i=0,\dots, 6\right\}$ of $Q$ is another root lattice, which can be written as 
\begin{equation*}
\begin{gathered}
Q^{\perp} = \Z \alpha_0 + \Z \alpha_1 + \Z \alpha_2,\\
\alpha_0 = \h_f + \h_g - \E_{5678}, \quad \alpha_1 = \h_f - \E_{34}, \quad \alpha_2 = \h_g - \E_{12}. 
\end{gathered}
\end{equation*}
Then $Q^{\perp}$ is isomorphic to the root lattice of the affine root system of type $A_2^{(1)}$.
For generic values of the parameters $a$, the real roots of this root system will not be classes of effective divisors but in special cases, in particular when any of $a_0,a_1,a_2=0$, some will become effective, and the set 
\begin{equation*}
\Delta^{\operatorname{nod}} \defeq \left\{ \F \in \Pic(X) ~|~ \F\cdot\F=-2,~~\F\cdot \delta_i=0,~~\F = [C] \text{ for some irreducible curve } C\right\},
\end{equation*}
of classes of nodal curves disjoint from the components of the anticanonical divisor will be nonempty. 
This is due to the parameters $a_0,a_1,a_2$ being values on $\alpha_0, \alpha_1, \alpha_2$ of the period map of $X_{t;a}$ defined in terms of a rational two-form whose divisor of poles is $D$ (see \cite[Proposition 22]{SAKAI2001}
 for details).

For generic values of the parameters $a$, the curves appearing in Lemmas \ref{lem:behaviourT1} and \ref{lem:behaviourT2} are irreducible. 
However, for special values of $a$ some of these curves decompose with one component being a nodal curve. We have summarised this in the following proposition.

\begin{proposition}
At special values of parameters $a$, the curves in Lemmas \ref{lem:behaviourT1} and \ref{lem:behaviourT2} decompose as follows:
\begin{description}
\item[$\bullet$ $C_{T_1(z_+)}$]
When $a_0 = 1 -a_1 -a_2=0$, the defining equation factors as
\begin{equation*}
g^2 - g(f+2t)+ 2 a_0 \xrightarrow{~a_0=0~} ( f - g + 2t) g, 
\end{equation*}
which corresponds to the decomposition
\begin{equation*}
\begin{gathered}
[C_{T_1(z_+)}] =  \h_f + 2 \h_g - \E_{15678} = (\h_f + \h_g - \E_{5678}) + (\h_g- \E_1), \\
 \h_f + \h_g - \E_{5678} = \alpha_0 \in \Delta^{\nod} \text{ when } a_0=0.
\end{gathered}
\end{equation*}

\item[$\bullet$ $C_{T_1^{-1}(p_-)}$]
When $a_1 = 0$, the defining equation factors as
\begin{equation*}
f g - 2a_1 \xrightarrow{~a_1=0~}  f g, 
\end{equation*}
which corresponds to the decomposition
\begin{equation*}
\begin{gathered}
[C_{T_1^{-1}(p_-)}] =  \h_f + \h_g - \E_{134} = (\h_f - \E_{34}) + (\h_g- \E_1), \\
\h_f - \E_{34} = \alpha_1 \in \Delta^{\nod} \text{ when } a_1=0.
\end{gathered}
\end{equation*}

\item[$\bullet$ $C_{T_1^{-1}(z_-)}$]
When $a_1 = 0$, the defining equation factors as
\begin{equation*}
f^3 g - f^2 (g^2 - 2 t g - 2a_2)+ 4 a_1 f(g-t) - 4 a_1^2 \xrightarrow{~a_1=0~}  -f^2 (g^2 - f g - 2 t g - 2 a_2) , 
\end{equation*}
which corresponds to the decomposition
\begin{equation*}
\begin{gathered}
[C_{T_1^{-1}(z_-)}] =  3 \h_f + 2 \h_g - \E_{123344567} = 2(\h_f - \E_{34}) + (\h_f +2 \h_g - \E_{12567}), \\
\h_f - \E_{34} = \alpha_1  \in \Delta^{\nod} \text{ when } a_1=0.
\end{gathered}
\end{equation*}
Similarly when $a_1+a_2 = 0$, the defining equation factors as
\begin{equation*}
f^3 g - f^2 (g^2 - 2 t g - 2a_2)+ 4 a_1 f(g-t) - 4 a_1^2 \xrightarrow{~a_1+a_2=0~}  - (f^2 - f g + 2 t f + 2a_1)(f g - 2a_1), 
\end{equation*}
which corresponds to the decomposition
\begin{equation*}
\begin{gathered}
[C_{T_1^{-1}(z_-)}] =  3 \h_f + 2 \h_g - \E_{123344567} =  (\h_f +\h_g - \E_{1234}) + (2\h_f + \h_g - \E_{34567}), \\
\h_f +\h_g - \E_{1234}= \alpha_1 + \alpha_2 \in \Delta^{\nod} \text{ when } a_1+a_2=0.
\end{gathered}
\end{equation*}

\item[$\bullet$ $C_{T_2(p_+)}$]
When $a_0=1-a_1-a_2 = 0$, the defining equation factors as
\begin{equation*}
f^2 +f(2t-g) - 2 a_0 \xrightarrow{~a_0=0~} f (f-g+2t) , 
\end{equation*}
which corresponds to the decomposition
\begin{equation*}
\begin{gathered}
[C_{T_2(p_+)}] =  2 \h_f +  \h_g - \E_{35678} =  (\h_f+\h_g - \E_{5678}) + (\h_f - \E_{3}), \\
\h_f+\h_g - \E_{5678} = \alpha_0 \in \Delta^{\nod} \text{ when } a_0=0.
\end{gathered}
\end{equation*}

\item[$\bullet$ $C_{T_2(z_-)}$]
When $a_1= 0$, the defining equation factors as
\begin{equation*}
f^2 +f(2t-g) +2 a_1\xrightarrow{~a_1=0~}  f (f-g+2t) , 
\end{equation*}
which corresponds to the decomposition
\begin{equation*}
\begin{gathered}
[C_{T_2(p_+)}] =  2 \h_f +  \h_g - \E_{34567} =   (\h_f - \E_{34}) + (\h_f + \h_g - \E_{567}), \\
\h_f - \E_{34} =\alpha_1 \in \Delta^{\nod} \text{ when } a_1=0.
\end{gathered}
\end{equation*}

\item[$\bullet$ $C_{T_2^{-1}(p_-)}$]
When $a_2= 0$, the defining equation factors as
\begin{equation*}
f g + 2a_2 \xrightarrow{~a_2=0~}   f g , 
\end{equation*}
which corresponds to the decomposition
\begin{equation*}
\begin{gathered}
[C_{T_2^{-1}(p_-)}] =  \h_f +  \h_g - \E_{123} =   (\h_g - \E_{12}) + (\h_f - \E_{3}), \\
\h_g - \E_{12} = \alpha_2 \in \Delta^{\nod} \text{ when } a_2=0.
\end{gathered}
\end{equation*}

\end{description}
\end{proposition}
\begin{proof}
This is a calculation which we demonstrate only in the first case of $C_{T_1(z_+)}$, since the others are analogous.
Consider the curve on $\p^1 \times \p^1$ given in the chart $(f,g)$ by the defining equation 
\begin{equation} \label{eq:prop4.7fgcurve}
g^2-g(f+2t)+2a_0=0.
\end{equation}
On $\p^1 \times \p^1$ this is a smooth, and thus irreducible, curve as long as $a_0\neq 0$.
When $a_0 \neq 0$ and the other parameters are also generic,  its proper transform under the blowups of $b_1,\dots,b_8$ is an irreducible curve on $X_{t,a}$. 
Its class in $\Pic(X_{t;a})$
\begin{equation*}
    [C_{T_1(z_+)}] = \h_f + 2 \h_g - \E_{15678},
\end{equation*}
can be computed by calculation in charts. 
The coefficients of $\h_f$ and $\h_g$ are deduced from the bidegree of \eqref{eq:prop4.7fgcurve} as a curve on $\p^1 \times\p^1$, with the remaining ones computable by checking multiplicities of $E_i$ in the total transform of the curve \eqref{eq:prop4.7fgcurve} under the blowups of $b_1,\dots,b_8$.


However, when $a_0=0$, the curve \eqref{eq:prop4.7fgcurve} in the affine $(f,g)$-plane factorises into two curves, defined  respectively by $f-g+2t=0$ and $g=0$. 
We compute the classes of the corresponding curves in $\Pic(X_{t;a})$ when $a_0=0$, using the coordinates introduced for the blowups of $b_1,\dots,b_8$ with $a_0$ set to zero. 
The class of the proper transform of the curve defined by $g=0$ is $\h_g -\E_1$, while the class of the proper transform of the curve $f-g+2t=0$ is $\h_f + \h_g - \E_{5678}$, and these provide the decomposition 
$[C_{T_1(z_+)}] = (\h_f + \h_g - \E_{5678}) + (\h_g- \E_1)$ in the case $a_0=0$.
The element $\h_f + \h_g - \E_{5678} \in \Pic(X_{t;a})$ is of self-intersection $-2$ and there does not exist an irreducible curve representing it in the case of generic $a$. 
But when $a_0=0$ it is the class of a $-2$ curve disjoint from the components of the anticanonical divisor $D$, which is precisely the proper transform of that defined by $f-g+2t=0$.
Since this is disjoint from the components of $D$ its class lies in $Q^{\perp}$, and it is given by $\alpha_0$, which lies in $\Delta^{\operatorname{nod}}$ when $a_0=0$.
\end{proof}

%
%
%

We will encounter topological changes in the real $(f,g)$-plane in many of the curves in Lemmas \ref{lem:behaviourT1} and \ref{lem:behaviourT2} occurring precisely when the parameters cross hyperplanes on which the curves decompose and nodal curves appear as above, which will be demonstrated in the proofs in the sections that follow.
%
%
%
%


\section{Proof for region $\mathrm{I}$} \label{sec:region1}

In this section we give the inductive proof of Theorem \ref{th:region1} in full detail.
The base case will be $m=n=0$, and we first establish the topology in the real $(f,g)$-plane for parameters in region $\mathrm{I}$ of the curves in Lemmas \ref{lem:behaviourT1} and \ref{lem:behaviourT2}. 
These curves are crossed by solutions, obtained by applying $T_1$ or $T_2$ respectively to a solution $(f(t),g(t))$, at the values of $t$ where the original solution has apparent singularities. 
After establishing this, we perform four inductive steps for each of $T_1$ and $T_2$ of the formulas in Theorem \ref{th:region1} to account for all combinations of parities of $m$ and $n$.



\subsection{Translation $T_1$ in region $\mathrm{I}$}

Recall that if a solution $(f,g)$ for parameters $a$ has an apparent singularity at $t=t_*$, then the solution $(\widehat{f}(t),\widehat{g}(t))$ for parameters $\widehat{a}$ will have either an apparent singularity or a crossing of one of the curves in Lemma \ref{lem:behaviourT1} when $t=t_*$, according to the following:
\begin{equation} \label{singsmovementT1}
p_+ \xrightarrow{~T_1~} C_{T_1(p_+)}, \qquad p_- \xrightarrow{~T_1~} z_+, \qquad z_+ \xrightarrow{~T_1~} C_{T_1(z_+)}, \qquad z_- \xrightarrow{~T_1~} p_+.
\end{equation}
For brevity we will refer to the latter situation as the solution $(\widehat{f}(t),\widehat{g}(t))$ having a \emph{curve crossing} at $t=t_*$.

\subsubsection{Curve crossings in the real $(f,g)$-plane relevant to $T_1$ in region $\mathrm{I}$}

For parameters relevant to our proof for region $\mathrm{I}$, the configuration of the curves $C_{T_1(p_+)}$ and $C_{T_1(z_+)}$ in the real $(f,g)$-plane is provided by the following lemma.
\begin{lemma} \label{lem:regionsT1}
For parameters $a$ in region $\mathrm{I}$ such that $a_1>1$, $a_2>0$, and $t<0$, the curves 
\begin{equation*}
\begin{aligned}
C_{T_1(p_+)} &: P_{T_1(p_+)} \defeq f - g +2 t = 0, \\
C_{T_1(z_+)} &: P_{T_1(z_+)} \defeq g^2 - g(f+2 t)+ 2 a_0 = 0, \\
C_{z_-} &: f=0
\end{aligned}
\end{equation*}
 divide the real $(f,g)$-plane into the eight regions, $L_1,\dots, L_4$, $R_1,\dots,R_4$, shown in Figure \ref{fig:regionsT1}, defined by the signs of the functions as in Table \ref{table:sgns_regions}.

\renewcommand{\arraystretch}{1.1}
\begin{table}[ht]
\centering
\begin{tabular}{|c || c | c | c | c | c | c | c | c | c |} 
 \hline
& $L_1$ & $L_2$ & $L_3$ & $L_4$ & $R_1$ & $R_2$ & $R_3$ & $R_4$\\ \hline
$f$ & $-$ & $-$ & $-$ & $-$ & $+$ & $+$ & $+$ & $+$\\
$P_{T_1(p_+)}$ & $-$ & $-$ & $+$ & $+$ & $-$ & $-$ & $+$ & $+$\\
$P_{T_1(z_+)}$ & $+$ & $-$ & $-$ & $+$ & $+$ & $-$ & $-$ & $+$\\
 \hline
\end{tabular}
\caption{Signs of functions $f$, $P_{T_1(p_+)}$ and $P_{T_1(z_+)}$ in the eight regions $L_k$, $R_k$, $1\leq k\leq 4$. 
}
\label{table:sgns_regions}
\end{table}
\end{lemma}



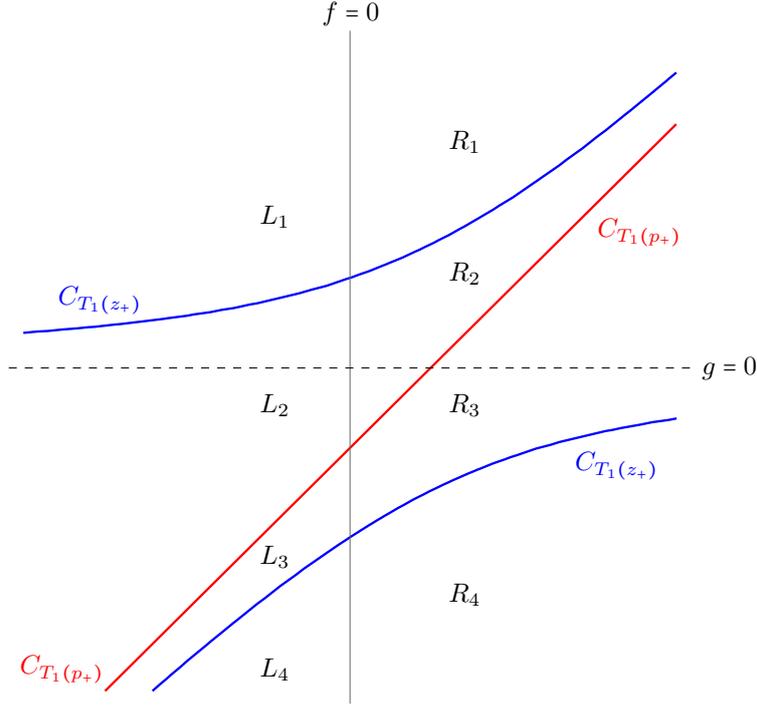
\begin{figure}[htb]
    \centering
    \begin{tikzpicture}
    \node (pic) at (0,0) {\includegraphics[width=.6\textwidth]{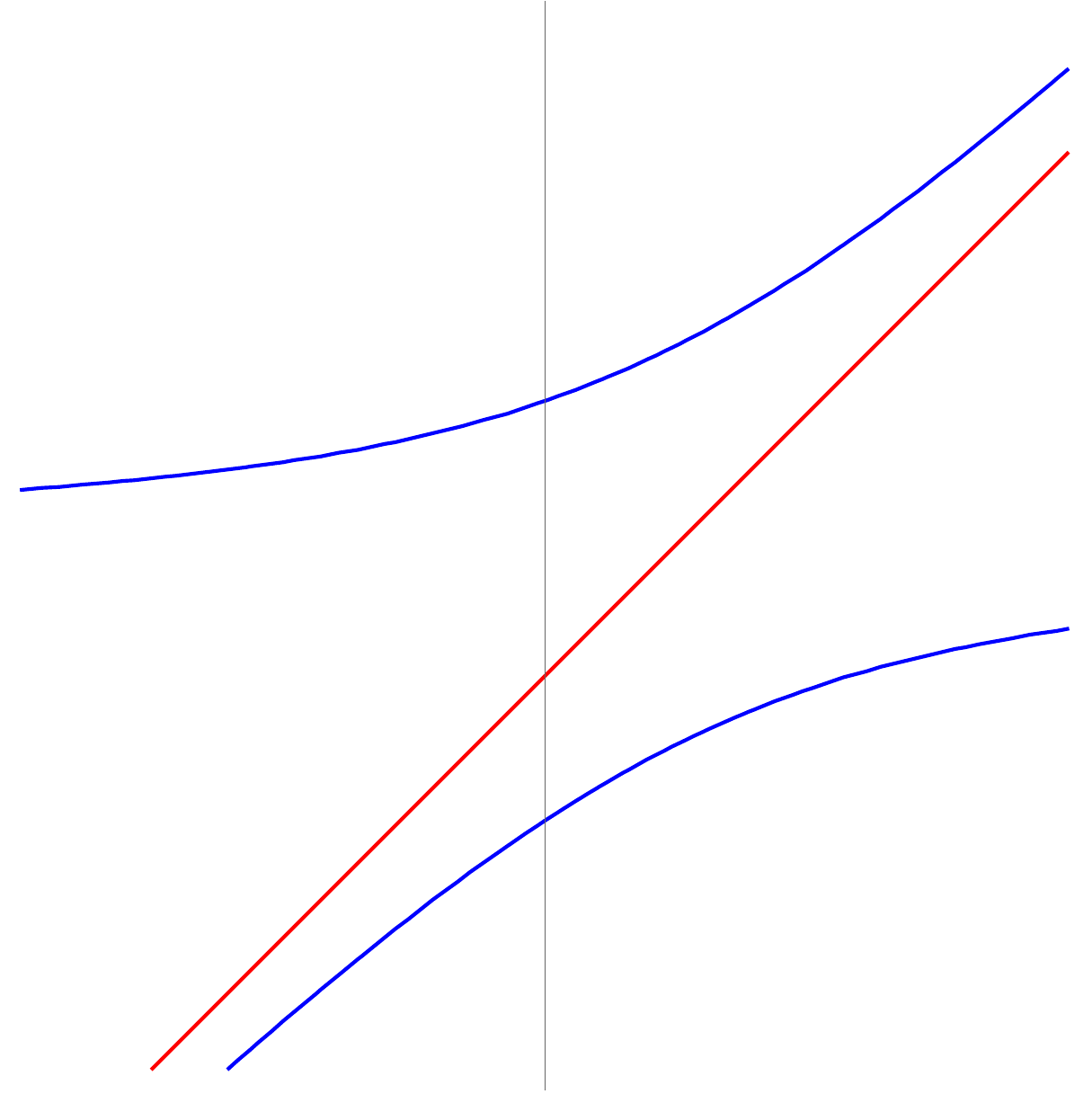}};
    \node (L1) at (-1,2) {$L_1$};
    \node (L2) at (-1,-.5) {$L_2$};
    \node (L3) at (-1,-2.5) {$L_3$};
    \node (L4) at (-1,-4) {$L_4$};    
    \node (R1) at (1.5,3) {$R_1$};
    \node (R2) at (1.5,1.25) {$R_2$};
    \node (R3) at (1.5,-.5) {$R_3$};
    \node (R4) at (1.5,-3) {$R_4$};    
    \node (bluecurvelabelL) at (-3.3,.9) {{\color{blue}$C_{T_1(z_+)}$}};    
    \node (bluecurvelabelR) at (3.5,-1.3) {{\color{blue}$C_{T_1(z_+)}$}};    
    \node (redcurvelabelL) at (-3.8,-4) {{\color{red}$C_{T_1(p_+)}$}};    
    \node (redcurvelabelR) at (3.8,1.8) {{\color{red}$C_{T_1(p_+)}$}};    
    \node (f0label) at (0,4.7) {$f=0$};    
    \draw[dashed] (-4.5,0)  -- (4.5,0) node[right] {$g=0$}; 
\end{tikzpicture}
    \caption{Regions and curves relevant to $T_1$ in region $\mathrm{I}$ for $t<0$.}
    \label{fig:regionsT1}
\end{figure}

When a real solution $(f(t),g(t))$ has an apparent singularity, then the solution curve has the following behaviour relative to the curves above, which we show in Figure \ref{fig:masterpictureT1}.

\begin{lemma} \label{lem:crossingsT1}
For parameters $a$ such that $a_1>1$, $a_2>0$, apparent singularities of real solutions at $t=t_*<0$ correspond to the following behaviours in the real $(f,g)$-plane:
\begin{itemize}
    \item $z_+$ is a crossing from $L_4$ to $R_1$.
    \item $z_-$ is a crossing from $\cup_i R_i$ to $\cup_i L_i$.
    \item $p_+$ is a crossing from $L_2$ to $R_3$.
    \item $p_-$ is either
    \begin{itemize}
        \item a crossing from $R_1$ to $L_3$,
        \item a crossing from $R_1$ to $L_4$, with the curve $C_{T_1(z_+)}$ on $X$ crossed simultaneously,
        \item a crossing from $R_2$ to $L_4$,
    \end{itemize}
\end{itemize}
\end{lemma}
\begin{proof}
Some of the above follow immediately from the forms of the Laurent expansions about movable singularities as in Lemma \ref{lem:singularities_expansions} without the need for calculation, for example a $z_+$ being a crossing from $L_4$ to $R_1$.
 
In other cases calculation using Laurent series expansions is required to determine the placement of the solution near an apparent singularity relative to the curve.
For example, to determine where a solution in the neighbourhood of a $p_+$ will be placed relative to the curve $C_{T_1(z_+)}$, we substitute the relevant Laurent expansions as $t \to t_*$ into the defining equation of the curve: 
\begin{equation*}
P_{T_1(z_+)}  = 2(a_1-1) + (2 \eta - (1+4a_2)t_*) (t-t_*) + \mathcal{O}( (t-t_*)^2).
\end{equation*}
Noting the signs of the defining polynomial $P_{T_1(z_+)}$ in different regions given in Lemma \ref{lem:regionsT1}, we can deduce the placement of a solution curve in the neighbourhood of a $p_+$ relative to $C_{T_1(z_+)}$ as in Figure \ref{fig:masterpictureT1} for $a_1>1$, which in particular is the case for any parameters $a=(\tfrac{1}{3}-m-n,\tfrac{1}{3}+n,\tfrac{1}{3}+m)$ corresponding to the generalised Okamoto rationals in region $\mathrm{I}$ with  $n\geq1$.

In cases when the exceptional divisor corresponding to apparent singularities intersects, on the surface $X$, one of the curves bounding the regions then its placement may depend on the free parameter in its Laurent series expansion.
This is the case for $p_-$ relative to $C_{T_1(p_+)}$, since a $p_-$ corresponds to a crossing of the exceptional divisor $E_8$, which intersects this curve: $[C_{T_1(p_+)}]\cdot \E_8 = (\h_f + 2 \h_g - \E_{15678})\cdot \E_8 = 1$.
Substituting Laurent expansions we find
\begin{equation*}
    P_{T_1(z_+)} = \left( 2 \eta + t_*(1+4 a_0) \right) (t-t_*) - \left(t_*^2 + 2 t_* \eta + 4 a_0( t_*^2 - 1 + a_1) \right) (t-t_*)^2 +  \mathcal{O}\left( (t-t_*)^3 \right).
\end{equation*}
The regions in which the solution lies for $t$ in the neighbourhood of $t_*$ depend on the value of $\eta$, and can be determined from the sign of $P_{T_1(z_+)}$. 
In particular there can be a change of sign of $P_{T_1(z_+)}$ as the solution has a $p_-$, or $\eta$ could take the special value that leads to $P_{T_1(z_+)} =\mathcal{O}\left( (t-t_*)^2\right)$, in which case the solution crosses $C_{T_1(z_+)}\cap E_8$ when lifted to $X_{t_*;a}$ and we must determine the sign of the coefficient of $(t-t_*)^2$ to deduce its placement relative to $C_{T_1(z_+)}$ for $t$ close to $t_*$. 
When $\eta = - \tfrac{1}{2} t_* (1+ 4 a_0)$ we calculate that
\begin{equation*}
    P_{T_1(z_+)} = 4 (a_1-1)(a_1+a_2-1)  (t-t_*)^2 +  \mathcal{O}( (t-t_*)^3),
\end{equation*}
so in particular for $a_1>1$ and $a_2> 0$ such a solution
will be in region $R_1$, respectively $L_4$, for $(t-t_*)$ small and negative, respectively positive.
%
%
The remaining claims are proved similarly.
\end{proof}

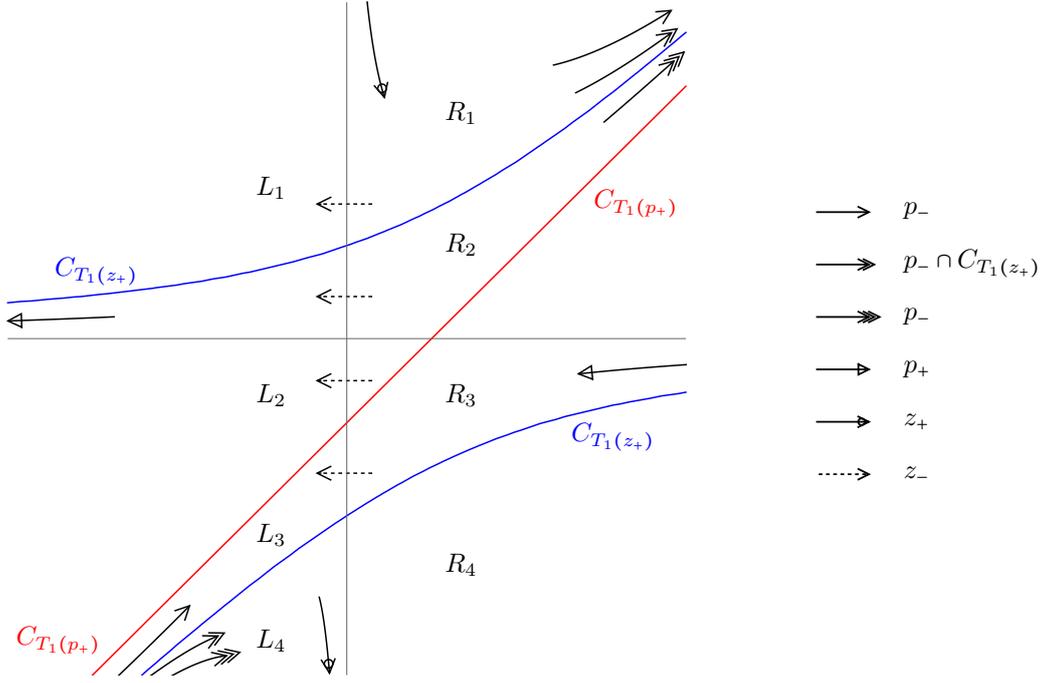
\begin{figure}[htb]
    \centering
	\begin{tikzpicture}
    \node (pic) at (0,0) {\includegraphics[width=.6\textwidth]{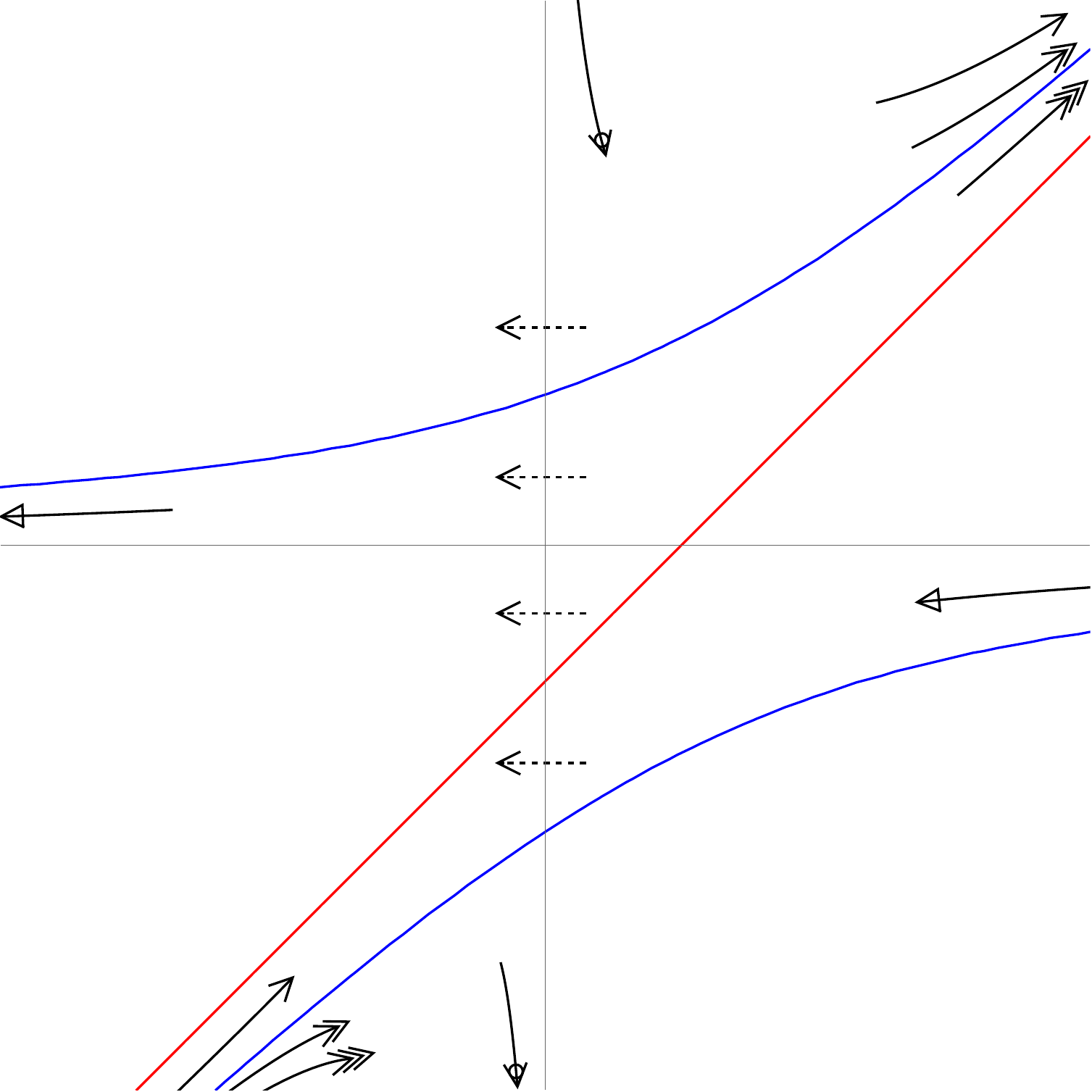}};
    \node (L1) at (-1,2) {$L_1$};
    \node (L2) at (-1,-.75) {$L_2$};
    \node (L3) at (-1,-2.6) {$L_3$};
    \node (L4) at (-1,-4) {$L_4$};    
    \node (R1) at (1.5,3) {$R_1$};
    \node (R2) at (1.5,1.25) {$R_2$};
    \node (R3) at (1.5,-.75) {$R_3$};
    \node (R4) at (1.5,-3) {$R_4$};    
    \node (bluecurvelabelL) at (-3.3,.9) {{\color{blue}$C_{T_1(z_+)}$}};    
    \node (bluecurvelabelR) at (3.5,-1.3) {{\color{blue}$C_{T_1(z_+)}$}};    
    \node (redcurvelabelL) at (-3.8,-4) {{\color{red}$C_{T_1(p_+)}$}};    
    \node (redcurvelabelR) at (3.8,1.8) {{\color{red}$C_{T_1(p_+)}$}};    
    \node (legend) at (6.6,-0.05) {\includegraphics[width=.069\textwidth]{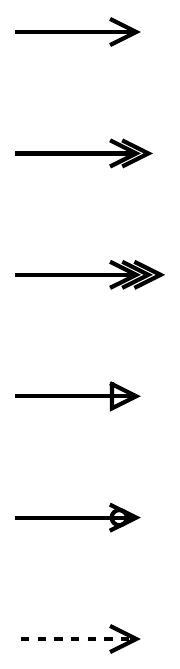}};
    \node (pminus1) at (7.2,1.7) [right] {$p_-$};
    \node (pminus2) at (7.2,1) [right] {$p_- \cap C_{T_1(z_+)}$};
    \node (pminus3) at (7.2,.3) [right] {$p_-$};
    \node (pplus) at (7.2,-.4) [right] {$p_+$};
    \node (zplus) at (7.2,-1.1) [right] {$z_+$};
    \node (zplus) at (7.2,-1.8) [right] {$z_-$};
	\end{tikzpicture}
    \caption{Apparent singularities and curves relevant to $T_1$ in region $\mathrm{I}$ for $t<0$.}
    \label{fig:masterpictureT1}
\end{figure}

\subsubsection{Boundary conditions at $t=0$ and $t\rightarrow -\infty$}
We will also require the behaviour of a solution relative to the curves as $t\to -\infty$, and when $t$ has a special apparent singularity at the origin as in Proposition \ref{prop:behaviouratorigin}.
The following two lemmas are established by substitution of the Laurent series expansions from Lemma \ref{lem:behaviouratinfinity} and Proposition \ref{prop:behaviouratorigin} respectively into the defining equations of the relevant curves.

\begin{lemma} \label{lem:asymptoticregionT1}
For sufficiently large negative $t$, the Okamoto rational solution $(f_{m,n}(t),g_{m,n}(t))$ lies in region $R_2$.
\end{lemma}
\begin{proof}
Recall that region $R_2$ is characterised by $f>0$, $P_{T_1(p_+)} < 0$, and $P_{T_1(z_+)} < 0$.
Indeed with the asymptotic series in Lemma \ref{lem:behaviouratinfinity} we have that as $t\to \infty$
\begin{equation*}
f = - \tfrac{2}{3} t + \mathcal{O}\left(\tfrac{1}{t}\right), \qquad P_{T_1(p_+)} = \tfrac{2}{3} t + \mathcal{O}\left(\tfrac{1}{t}\right), \qquad P_{T_1(z_+)} = - \tfrac{4}{9} t^2 + \mathcal{O}\left(1\right),
\end{equation*}
so for sufficiently large negative $t$ the solution lies in $R_2$ as claimed.
\end{proof}

\begin{lemma}
If a solution has a special apparent singularity at the origin $t=0$ of type $p_{\pm}'$ or $z_{\pm}'$, then for sufficiently small negative $t$ the solution $(f(t),g(t))$ will lie in the following regions:
\begin{equation*}
z_-'  : R_2, \qquad z_+' : L_4, \qquad p_-' : R_1, \qquad p_+' : L_2.
\end{equation*}
\end{lemma}
\begin{proof}
The proof is analogous to the proof of Lemma \ref{lem:asymptoticregionT1}, by substituting the asymptotic expansions in Proposition \ref{prop:behaviouratorigin} into the polynomials defining the different regions in Lemma \ref{lem:regionsT1} and looking at the signs.
\end{proof}

\subsection{Inductive arguments for $T_1$ in region $\mathrm{I}$}

We are now ready to put together the inductive steps for $T_1$.
In what follows we will refer to $q_{m,n}$ as the `old' solution and $q_{m,n+1}$ obtained by applying $T_1$ as the `new' solution, and to their apparent singularities similarly. 
For example, from \eqref{singsmovementT1} we know that an old $p_-$ is a new $z_+$. 
Since we have different formulas for different parities of $m$ and $n$, the inductive step consists of four different cases, the first of which is the following.

\begin{lemma}[$m$ even, $n$ even $\xrightarrow{T_1}$ $m$ even, $n$ odd]
\label{lem:eveneven1}
Under the inductive hypothesis
   \begin{equation*}
        \mathfrak{S}(q_{2\mu,2\nu}) = (p_-\,z_+\,z_-\,p_+)^{\mu}\,(z_-\,p_+)^{\nu}\,\hat{z}_-\,(p_+\,z_-)^{\nu}\,(p_+\,z_-\,z_+,\,p_-)^{\mu},
    \end{equation*}        
we have:
    \begin{equation*}
        \mathfrak{S}(q_{2\mu,2\nu+1}) =  (p_-\,z_+\,z_-\,p_+)^{\mu}\,(z_-\,p_+)^{\nu}\,z_-\,\hat{p}_+\, z_-\,(p_+\,z_-)^{\nu}\,(p_+\,z_-\,z_+\,p_-)^{\mu}.
    \end{equation*}
\end{lemma}

To prove this and the other inductive steps we will require several smaller results related to the behaviour of subsequences in the singularity signatures under $T_1$.
We point out here that since we are applying $T_1$ to a solution in region $\mathrm{I}$ the new solution will be for parameters $a$ with $a_1>1$, $a_2>0$. 
Therefore Lemmas \ref{lem:regionsT1} and \ref{lem:crossingsT1} apply and we can argue using Figure \ref{fig:masterpictureT1}.

Note that the part of the singularity signature of $q_{2\mu,2\nu}$ corresponding to apparent singularities at non-positive values of $t$, factorises as $A B C$, where 
\begin{equation*}
A=(p_-\,z_+\,z_-\,p_+)^{\mu-1}, \quad B=p_-\, z_+\,z_-, \quad C = p_+(z_- \,p_+)^{\nu}\hat{z}_-,
\end{equation*}
where $\hat{z}_-$ indicates the apparent singularity at the origin. 
We deal with these factors individually in the following lemmas.

\begin{lemma} \label{lem:asymptoticT1}
A new $p_-$ must occur before the first new $z_+$ (first old $p_-$).
This corresponds to the following movement of the new solution curve through the regions indicated above:
\begin{equation*}
    \begin{tikzcd}
        R_2 \arrow[r,"p_-"] & L_4 \arrow[r,"z_+"] & R_1.
    \end{tikzcd}
\end{equation*}
\end{lemma}
\begin{proof}
    From Lemma \ref{lem:asymptoticregionT1}, we have that for large negative $t$ the new solution is in $R_2$. 
    From the assumption that the first apparent singularity of the old solution is a $p_-$, the new solution must pass from $R_2$ to $L_4$ in order to have the first new $z_+$, without first crossing $C_{T_1(z_+)}$ or $C_{T_1(p_+)}$ in the process and also without having a new $p_+$ since this would correspond to an old $z_-$.
%
This means that there are only two possible moves for the new solution, either to cross from $L_2$ to $R_2$ via a new $z_-$, or to cross from $R_2$ to $L_4$ via a new $p_-$. As to the first possibility, we note that the solution curve is then stuck in $L_2$ under the restrictions above and in particular not able to reach $L_4$.
Therefore, the only possibility is for the new solution to cross from $R_2$ to $L_4$ via a new $p_-$, after which the first new $z_+$ occurs.

We will make extensive use of such arguments throughout the remainder of the paper, and we summarise the above as follows with $\lightning$ indicating a contradiction and $\checkmark$ indicating a path that is possible under the inductive assumption:
\begin{itemize}
    \item $R_2 \xrightarrow{z_-} L_2 \quad \lightning$
    \item $R_2 \xrightarrow{p_-} L_4\xrightarrow{z_+} R_1 \quad \checkmark$ 
\end{itemize}
As argued above, these two cases exhaust all possible paths before reaching the endpoint or being stuck, and thus the lemma follows.
\end{proof}

We will also use of arguments relating to substrings of the singularity signature of a new solution based on those of an old solution, so it will be convenient to introduce notation as follows.
Suppose an old solution has singularity signature which contains the string $(s_1 \, \dots \, s_N)$ corresponding to apparent singularities $s_1,\dots,s_N$ at times $t_1<\dots<t_N$. 
We write 
\begin{equation} \label{eq:5.10}
(s_1 \, \dots \, s_N) \xrightarrow{~T~}   (\tilde{s}_1 \, \dots \, \tilde{s}_{\tilde{N}}) ,
\end{equation}
to indicate that the new solution obtained by applying the translation $T$ to the old one has the apparent singularities $\tilde{s}_1,\dots,\tilde{s}_N$ (and only these apparent singularities) on the interval $[t_1,t_N]$ in this order, with $\tilde{s}_1$ and $\tilde{s}_{\tilde{N}}$ at $t_1$ and $t_N$ respectively.


\begin{lemma} \label{lem:T1:4cycles1}
We have the following behaviour of subsequences of singularity signatures under $T_1$ in region $\mathrm{I}$
\begin{equation*}
    (p_-\, z_+\, z_-\, p_+\, p_-) \xrightarrow{T_1} (z_+\, z_-\, p_+\, p_-\, z_+),
\end{equation*}
using the notation introduced in \eqref{eq:5.10}.
In particular the new solution must pass through regions via apparent singularities and curve crossings as follows, where the branching indicates different possibilities but these do not affect the singularity signature of the new solution: 
\begin{equation} \label{regionsCD:T1:4cycles1}
    \begin{tikzcd}
                    &          & L_1  \arrow[dr, "C_{T_1(z_+)}"]    &       &       &       &       &  &\\ 
~ \arrow[r, "z_+"]  & R_1 \arrow[ur,"z_-"] \arrow[dr,swap, "C_{T_1(z_+)}"] \arrow[rr,"z_- \cap C_{T_1(z_+)}"]     &          & L_2 \arrow[r,"p_+"]  & R_3  \arrow[r,"C_{T_1(p_+)}"] & R_2 \arrow[r,"p_-"]  & L_4 \arrow[r,"z_+"]  & R_1  \arrow[r,"z_+"] &~\\
                    &          & R_2  \arrow[ur,swap, "z_-"]    &       &       &       &       &  & 
    \end{tikzcd}
\end{equation}
\end{lemma}
\begin{proof}
The old $p_-$ at the start of the subsequence gives a new $z_+$, after which the new solution is in $R_1$. 
Then the new solution needs to cross $C_{T_1(z_+)}$ before any of $C_{T_1(p_+)}$, $z_+$ or $p_+$, then have a new $p_+$ (old $z_-$) before any of $C_{T_1(p_+)}$, $C_{T_1(z_+)}$ or $z_+$.
Therefore starting in $R_1$, the new solution must reach $L_2$, on the way only crossing $C_{T_1(z_+)}$ exactly once with no other crossings of $C_{T_1(p_+)}$, $C_{T_1(z_+)}$, $p_+$ or $z_+$, then pass from $L_2$ to $R_3$ via a new $p_+$.
Similarly to in the proof of the previous lemma, we summarise the results of playing this `maze game' with Figure \ref{fig:masterpictureT1} as follows:
\begin{itemize}
    \item $R_1 \xrightarrow{z_-} L_1 \xrightarrow{C_{T_1(z_+)}} L_2 \xrightarrow{p_+} R_3 \quad \checkmark$
    \item $R_1 \xrightarrow{C_{T_1(z_+)}} R_2 \xrightarrow{z_-} L_2\xrightarrow{p_+} R_3 \quad \checkmark$
    \item $R_1 \xrightarrow{C_{T_1(z_+)}\cap z_-}  L_2\xrightarrow{p_+} R_3 \quad \checkmark$
    \item $R_1 \xrightarrow{p_-} L_3 \xrightarrow{C_{T_1(z_+)}} L_4  \quad \lightning$
    \item $R_1 \xrightarrow{p_-\cap C_{T_1(z_+)}} L_4 \quad \lightning$
    \item $R_1 \xrightarrow{C_{T_1(z_+)}} R_2 \xrightarrow{p_-} L_4 \quad \lightning$
\end{itemize}
Therefore we have the first part of \eqref{regionsCD:T1:4cycles1}, meaning that after the new $z_+$ (old $p_-$) there must be a new $z_-$ before the next new $p_+$ (old $z_-$). Note that the new $z_-$ could occur before, after or at the same $t$ as the old $z_+$, as reflected in the branching in the schematic representation \eqref{regionsCD:T1:4cycles1} of movement of the new solution through regions.

After this new $p_+$ (old $z_-$), the new solution must cross $C_{T_1(p_+)}$ (old $p_+$) once before any of $z_+$, $p_+$, or $C_{T_1(z_+)}$, then have a new $z_+$ (old $p_-$).
In terms of regions, the new solution must pass from $R_3$ to $L_4$ and in the process only cross $C_{T_1(p_+)}$ once and not cross $z_+$, $p_+$, or $C_{T_1(z_+)}$.
Similarly we find the following:
\begin{itemize}
    \item $R_3 \xrightarrow{C_{T_1(p_+)}} R_2 \xrightarrow{p_-} L_4 \xrightarrow{z_+} R_1 \quad \checkmark$
    \item $R_3 \xrightarrow{C_{T_1(p_+)}} R_2 \xrightarrow{z_-} L_2 \quad \lightning$
    \item $R_3 \xrightarrow{z_-} L_3 \xrightarrow{C_{T_1(p_+)}} L_2 \quad \lightning$
\end{itemize}
Therefore after the new $p_+$ (old $z_-$) there must be a new $p_-$ after the old $p_+$ but before the next new $z_+$ (old $p_-$), and we have obtained \eqref{regionsCD:T1:4cycles1}.
\end{proof}

The topological arguments in the proofs of the following few lemmas are similar to those above and can be recovered easily using Figure \ref{fig:masterpictureT1}, so we only summarise them using similar notation.


\begin{lemma} \label{lem:T1:phasetransition}
For parameters in region $\mathrm{I}$,
\begin{equation*}
    (p_-\, z_+\, z_-) \xrightarrow{T_1} (z_+\, z_- \,p_+),
\end{equation*}
using the notation introduced in \eqref{eq:5.10}. 
The new solution passes through regions as follows: 
\begin{equation*}
    \begin{tikzcd}
                    &       & R_2 \arrow[dr,"z_-"]  &          &  \\
    ~ \arrow[r,"z_+"] & R_1 \arrow[ur,"C_{T_1(z_+)}"]\arrow[rr,"z_-\cap C_{T_1(z_+)}"]\arrow[dr,swap,"z_-"] &       & L_2  \arrow[r,"p_+"]    & ~\\
                    &       & L_1 \arrow[ur,swap, "C_{T_1(z_+)}"]  &          &
    \end{tikzcd}
\end{equation*}
\end{lemma}
\begin{proof}
We have the following possible paths,
    \begin{itemize}
    \item $R_1 \xrightarrow{z_-} L_1 \xrightarrow{C_{T_1(z_+)}} L_2 \xrightarrow{p_+} \quad \checkmark$
    \item $R_1 \xrightarrow{C_{T_1(z_+)}} R_2 \xrightarrow{z_-} L_2 \xrightarrow{p_+} \quad \checkmark$
    \item $R_1 \xrightarrow{C_{T_1(z_+)} \cap z_-} L_2 \xrightarrow{p_+} \quad \checkmark$
    \item $R_1 \xrightarrow{p_-} L_3 \xrightarrow{C_{T_1(z_+)}} L_4 \quad \lightning$
    \item $R_1 \xrightarrow{p_-\cap C_{T_1(z_+)}} L_4  \quad \lightning$
    \item $R_1 \xrightarrow{C_{T_1(z_+)}} R_2 \xrightarrow{p_-} L_4 \quad \lightning$
\end{itemize}
In each of the three admissible paths that do not lead to a contradiction, the singularity signature of the new solution is $(z_+\, z_- \,p_+)$ on the relevant interval and thus the lemma follows.
\end{proof}

\begin{lemma}\label{lem:T1:2cycles1}
For parameters in region $\mathrm{I}$,
\begin{equation*}
    (z_-\, p_+\, z_-) \xrightarrow{T_1} (p_+\, z_-\, p_+).
\end{equation*}
using the notation introduced in \eqref{eq:5.10}. 
The new solution passes through regions as follows: 
\begin{equation*}
    \begin{tikzcd}
                    &       & R_2 \arrow[dr,"z_-"]  &          &  \\
    ~ \arrow[r,"p_+"] & R_3 \arrow[ur,"C_{T_1(p_+)}"]\arrow[rr,"z_- \cap C_{T_1(p_+)}"]\arrow[dr,swap,"z_-"] &       & L_2  \arrow[r,"p_+"]    & ~\\
                    &       & L_3 \arrow[ur,swap,"C_{T_1(p_+)}"]  &          &
    \end{tikzcd}
\end{equation*}
\end{lemma}

\begin{proof}
%
%
We have the following possible paths,
\begin{itemize}
    \item $R_3 \xrightarrow{C_{T_1(p_+)}} R_2 \xrightarrow{p_-} L_4 \quad \lightning$
    \item $R_3 \xrightarrow{C_{T_1(p_+)}} R_2 \xrightarrow{z_1} L_2 \xrightarrow{p_+}  R_3\quad \checkmark$
    \item $R_3 \xrightarrow{C_{T_1(p_+)} \cap z_-} L_2 \xrightarrow{p_+}  R_3\quad \checkmark$
    \item $R_3 \xrightarrow{z_-} L_3 \xrightarrow{C_{T_1(p_+)}} L_2 \xrightarrow{p_+}  R_3\quad \checkmark$
\end{itemize}
In each of the three admissible paths that do not lead to a contradiction, the singularity signature of the new solution is $(p_+\, z_- \,p_+)$ on the relevant interval and thus the lemma follows.
\end{proof}

With the above four lemmas in hand, we are ready to prove the first inductive step.
\begin{proof}[Proof of Lemma \ref{lem:eveneven1}]
    To begin, apply Lemma \ref{lem:asymptoticT1} for the first $p_-$ of the new solution, then Lemma \ref{lem:T1:4cycles1} $\mu-1$ times.
    Then for the transition of the new solution's signature from cycles $(p_-\,z_+\,z_-\,p_+)$ to $(z_-\,p_+)$ we use Lemma \ref{lem:T1:phasetransition}. 
    Then finally we apply Lemma \ref{lem:T1:2cycles1} repeatedly $\nu$ times to arrive at the origin, which we know to be a new $p_+$. 
    The remainder of the singularity signature follows from the fact that $q_{m,n}(-t)=- q_{m,n}(t)$.
\end{proof}

The second combination of parities of $m$ and $n$ for which we consider the inductions step, is the following. 

\begin{lemma}[$m$ even, $n$ odd $\xrightarrow{T_1}$ $m$ even, $n$ even]
\label{lem:evenodd1}
Under the inductive hypothesis
       \begin{equation*}
        \mathfrak{S}(q_{2\mu,2\nu+1}) =  (p_-\,z_+\,z_-\,p_+)^{\mu}\,(z_-\,p_+)^{\nu}\,z_-\,\hat{p}_+\, z_-\,(p_+\,z_-)^{\nu}\,(p_+\,z_-\,z_+\,p_-)^{\mu},
    \end{equation*}
we have:
    \begin{equation*}
        \mathfrak{S}(q_{2\mu,2\nu+2}) = (p_-\,z_+\,z_-\,p_+)^{\mu}\,(z_-\,p_+)^{\nu+1}\,\hat{z}_-\,(p_+\,z_-)^{\nu+1}\,(p_+\,z_-\,z_+,\,p_-)^{\mu}.
    \end{equation*}
\end{lemma}
To prove this lemma, we require the following
\begin{lemma}\label{lem:T1:2cycles2}
For parameters in region $\mathrm{I}$,
\begin{equation*}
    (p_-\, z_+\, p_-) \xrightarrow{T_1} (z_+\, p_-\, z_+),
\end{equation*}
using the notation introduced in \eqref{eq:5.10}. 
The new solution passes through regions as follows: 
\begin{equation*}
    \begin{tikzcd}
                    &       & L_3 \arrow[dr,"C_{T_1(z_+)}"]  &         ~   \\
    ~ \arrow[r,"z_+"] & R_1 \arrow[ur,"p_-"]\arrow[rr,"p_-\cap C_{T_1(z_+)}"]\arrow[dr,swap,"C_{T_1(z_+)}"] &       & L_4  \arrow[r,"z_+"]   &~  \\
                    &       & R_2 \arrow[ur,swap, "p_-"]  &          ~
    \end{tikzcd}
\end{equation*}

\end{lemma}
\begin{proof}
We have the following possible paths,
\begin{itemize}
    \item $R_1 \xrightarrow{C_{T_1(z_+)}} R_2 \xrightarrow{z_-} L_2  \quad \lightning$
    \item $R_1\xrightarrow{z_- \cap C_{T_1(z_+)}} L_2   \quad \lightning$
    \item $R_1\xrightarrow{z_-} L_1  \xrightarrow{C_{T_1(z_+)}} L_2   \quad \lightning$
    \item $R_1\xrightarrow{p_-} L_3 \xrightarrow{C_{T_1(z_+)}} L_4   \quad \checkmark$
    \item $R_1\xrightarrow{p_- \cap C_{T_1(z_+)}} L_4   \quad \checkmark$
    \item $R_1 \xrightarrow{C_{T_1(z_+)}} R_2  \xrightarrow{p_-} L_4  \quad \checkmark$
\end{itemize}
In each of the three admissible paths that do not lead to a contradiction, the singularity signature of the new solution is $(z_+\, p_- \,z_+)$ on the relevant interval and thus the lemma follows.
\end{proof}


\begin{proof}[Proof of Lemma \ref{lem:evenodd1}]
    This is largely the same as the above, using Lemma \ref{lem:asymptoticT1} for the first $p_-$, then Lemma \ref{lem:T1:4cycles1} $\mu-1$ times, followed by Lemma \ref{lem:T1:phasetransition} for the transition then Lemma \ref{lem:T1:2cycles1} until we arrive at the origin.    
\end{proof}

The following will also require one more lemma.
\begin{lemma}[$m$ odd, $n$ even $\xrightarrow{T_1}$ $m$ odd, $n$ odd]
\label{lem:oddeven1}
Under the inductive hypothesis
       \begin{equation*}
        \mathfrak{S}(q_{2\mu+1,2\nu}) = (p_-\,z_+\,z_-\,p_+)^{\mu} (p_-\,z_+)^{\nu}\,p_- \,\hat{z}_+\,p_-(z_+\,p_-)^{\nu}(p_+\,z_-\,z_+\,p_-)^{\mu},
    \end{equation*}
we have:
    \begin{equation*}
        \mathfrak{S}(q_{2\mu+1,2\nu+1}) = (p_-\,z_+\,z_-\,p_+)^{\mu}(p_-\,z_+)^{\nu+1}\,\hat{p}_-\,(z_+\,p_-)^{\nu+1}(p_+\,z_-\,z_+\,p_-)^{\mu}.
    \end{equation*}
\end{lemma}

\begin{lemma} \label{lem:T1:phasetransition2}
For parameters in region $\mathrm{I}$,
\begin{equation*}
    (z_-\, p_+\, p_-) \xrightarrow{T_1} (p_+\, p_- \,z_+),
\end{equation*}
using the notation introduced in \eqref{eq:5.10}. 
The new solution passes through regions as follows: 
\begin{equation*}
    \begin{tikzcd}
    ~ \arrow[r,"p_+"] & R_3 \arrow[rr,"C_{T_1(p_+)}"] &      & R_2  \arrow[r,"p_-"]    & L_4 \arrow[r,"z_+"] &~\\
    \end{tikzcd}
\end{equation*}
\end{lemma}
\begin{proof}
We have the following possible paths,
    \begin{itemize}
    \item $R_3 \xrightarrow{z_-} L_3 \xrightarrow{C_{T_1(p_+)}} L_2  \quad \lightning$
    \item $R_3  \xrightarrow{C_{T_1(p_+)}} R_2 \xrightarrow{z_-} L_2 \quad \lightning$
    \item $R_3  \xrightarrow{C_{T_1(p_+)}} R_2 \xrightarrow{p_-} L_4 \quad \checkmark$
\end{itemize}
In the only admissible path that does not lead to a contradiction, the singularity signature of the new solution is $(p_+\, p_- \,z_+)$ on the relevant interval and thus the lemma follows.
\end{proof}

\begin{proof}[Proof of Lemma \ref{lem:oddeven1}]
Apply Lemma \ref{lem:asymptoticT1} for the first $p_-$ of the new solution, then Lemma \ref{lem:T1:4cycles1} $\mu-1$ times.
    Then for the transition we use Lemma \ref{lem:T1:phasetransition} and \ref{lem:T1:phasetransition2}. 
    Then apply Lemma \ref{lem:T1:2cycles2} $\nu+1$ times to arrive at the new $p_-$ at the origin.    
\end{proof}

\begin{lemma}[$m$ odd, $n$ odd $\xrightarrow{T_1}$ $m$ odd, $n$ even]
\label{lem:oddodd1}
Under the inductive hypothesis
           \begin{equation*}
        \mathfrak{S}(q_{2\mu+1,2\nu+1}) = (p_-\,z_+\,z_-\,p_+)^{\mu}(p_-\,z_+)^{\nu+1}\,\hat{p}_-\,(z_+\,p_-)^{\nu+1}(p_+\,z_-\,z_+\,p_-)^{\mu},
    \end{equation*} 
we have:   
    \begin{equation*}
        \mathfrak{S}(q_{2\mu+1,2\nu+2}) = (p_-\,z_+\,z_-\,p_+)^{\mu} (p_-\,z_+)^{\nu+1}\,p_- \,\hat{z}_+\,p_-(z_+\,p_-)^{\nu+1}(p_+\,z_-\,z_+\,p_-)^{\mu}.
    \end{equation*}
\end{lemma}

\begin{proof}[Proof of Lemma \ref{lem:oddodd1}]

Again this relies only on Lemmas \ref{lem:asymptoticT1}, \ref{lem:T1:4cycles1}, \ref{lem:T1:phasetransition},  \ref{lem:T1:2cycles2}, \ref{lem:T1:phasetransition2} and the knowledge of the singularity at the origin.

\end{proof}

\subsection{Translation $T_2$ in region $\mathrm{I}$}

Here and in later sections we will reuse notation for regions in the real $(f,g)$-plane. 
In cases where the proofs are straightforward applications of similar techniques to those explained in detail above we omit details.
Recall from Section \ref{sec:tools} that $T_2$ acts on apparent singularities as follows:
\begin{equation*}
p_+ \xrightarrow{~T_2~} C_{T_2(p_+)}, \qquad p_- \xrightarrow{~T_2~} p_+, \qquad z_+ \xrightarrow{~T_2~} C_{T_2(z_+)}, \qquad z_- \xrightarrow{~T_2~} C_{T_2(z_-)}.
\end{equation*}

\subsubsection{Curve crossings in the real $(f,g)$-plane relevant to $T_2$ in region $\mathrm{I}$}


\begin{lemma} \label{lem:regionsT2}
For parameters $a$ in region $\mathrm{I}$ such that $a_2>1$, $a_1>0$, and $t<0$, the curves 
\begin{equation*}
\begin{aligned}
    & C_{T_2(p_+)} : P_{T_2(p_+)} \defeq f^2 + f( 2t - g) - 2(a_1+a_2-1)= 0, 	\\
    &C_{T_2(z_+)} : P_{T_2(z_+)} \defeq f - g + 2t = 0, \\
    &C_{T_2(z_-)} : P_{T_2(z_-)} \defeq f^2 + f(2t - g) + 2 a_1 = 0,	\\
    &C_{z_-} : f=0,
\end{aligned}
\end{equation*}
 divide the real $(f,g)$-plane into the eight regions $L_1,\dots, L_4$, $R_1,\dots,R_4$ shown in Figure \ref{fig:masterpictureT2}.
\end{lemma}

\begin{lemma} \label{lem:crossingsT2}
When the parameters $a$ lie in region $\mathrm{I}$ with $a_2>1,a_1>0$, apparent singularities of real solutions at $t=t_*<0$ correspond to the following behaviours in the real $(f,g)$-plane as shown in Figure \ref{fig:masterpictureT2}:
\begin{itemize}
    \item $z_+$ is either
    \begin{itemize}
        \item a crossing from $L_3$ to $R_3$,
        \item a crossing from $L_3$ to $R_2$, with the curve $C_{T_2(z_-)}$ on $X$ crossed simultaneously,
        \item a crossing from $L_2$ to $R_2$.
    \end{itemize}
    \item $z_-$ is a crossing from $R_3 \cup E_4$ to $ L_1\cup L_2$.
    \item $p_+$ is a crossing from $L_1$ to $R_4$.
    \item $p_-$ is either
    \begin{itemize}
        \item a crossing from $R_2$ to $L_4$,
        \item a crossing from $R_2$ to $L_3$, with the curve $C_{T_2(p_+)}$ on $X$ crossed simultaneously,
        \item a crossing from $R_1$ to $L_3$.
    \end{itemize}
\end{itemize}
\end{lemma}





\begin{figure}[htb]
    \centering
	\begin{tikzpicture}
    \node (pic) at (0,0) {\includegraphics[width=.6\textwidth]{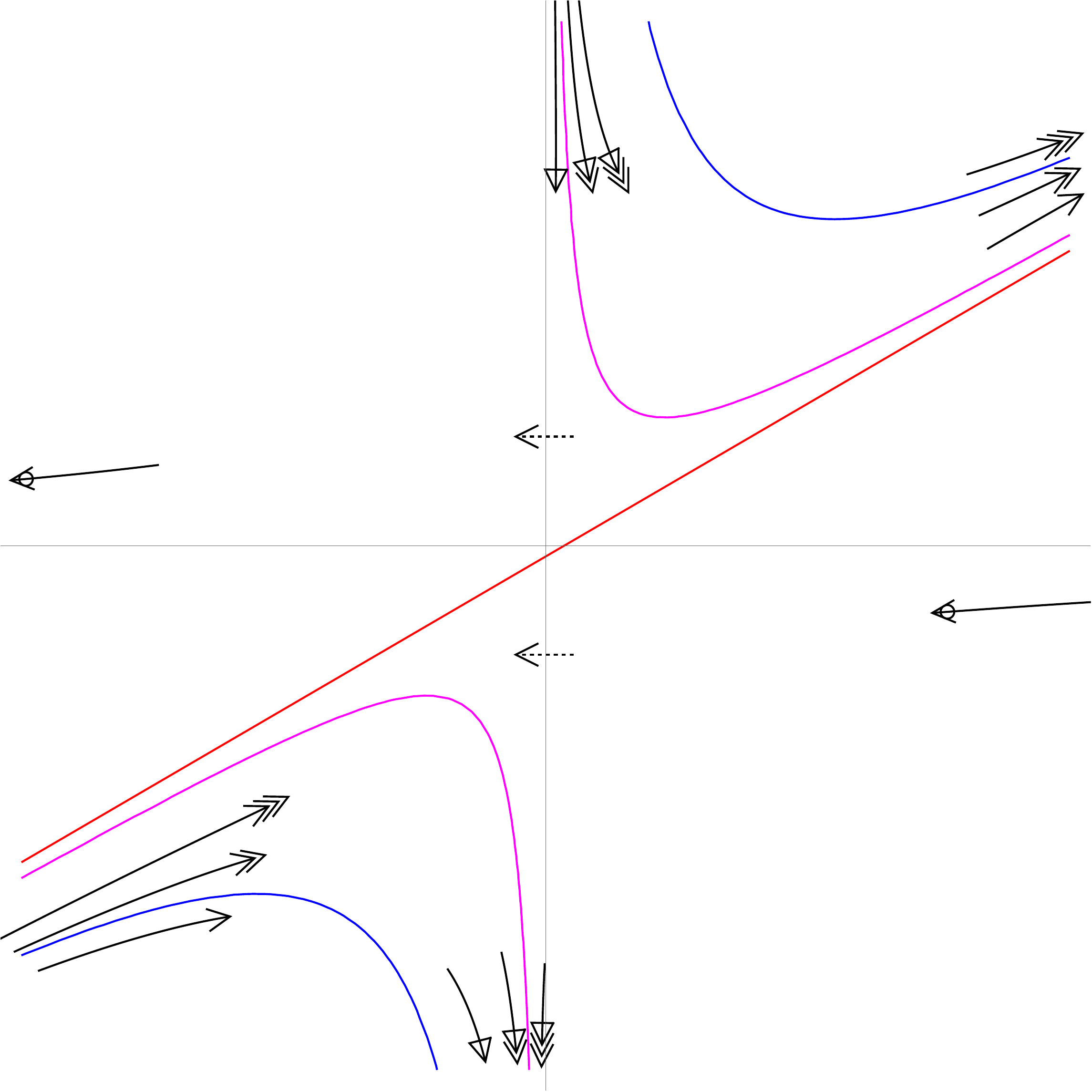}};
    \node (L1) at (-1.5,1.5) {$L_1$};
    \node (L2) at (-.6,-.95) {$L_2$};
    \node (L3) at (-1.25,-2.25) {$L_3$};
    \node (L4) at (-2,-4) {$L_4$};    
    \node (R1) at (2.5,3.75) {$R_1$};
    \node (R2) at (1.5,2) {$R_2$};
    \node (R3) at (.75,.75) {$R_3$};
    \node (R4) at (1.5,-1.5) {$R_4$};    
   \node (bluecurvelabelL) at (-1,-4.5) {{\color{blue}$C_{T_2(p_+)}$}};    
   \node (bluecurvelabelR) at (1,4.6) {{\color{blue}$C_{T_2(p_+)}$}};    
   \node (redcurvelabelL) at (-3,-1) {{\color{red}$C_{T_2(z_+)}$}};    
   \node (redcurvelabelR) at (3,1) {{\color{red}$C_{T_2(z_+)}$}};    
   \node (magentacurvelabelL) at (-5,-2.75) {{\color{magenta}$C_{T_2(z_-)}$}};    
   \node (magentacurvelabelR) at (5.2,2.75) {{\color{magenta}$C_{T_2(z_-)}$}};    
\begin{scope}[xshift=.2cm]       
   \node (legend) at (6.6,-0.05) {\includegraphics[width=.09\textwidth]{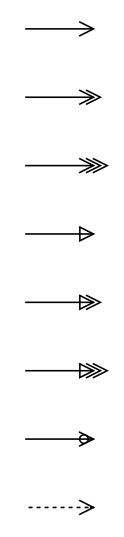}};
    \node (pminus1) at (7.2,2.55) [right] {$p_-$};
    \node (pminus2) at (7.2,1.8) [right] {$p_- \cap C_{T_2(p_+)}$};
    \node (pminus3) at (7.2,1.05) [right] {$p_-$};
    \node (zplus) at (7.2,0.3) [right] {$z_+$};
    \node (zplus2) at (7.2,-0.4) [right] {$z_+ \cap C_{T_2(z_-)}$};
    \node (zplus3) at (7.2,-1.15) [right] {$z_+$};
    \node (pplus) at (7.2,-1.95) [right] {$p_+$};
    \node (zminus) at (7.2,-2.65) [right] {$z_-$};
\end{scope}
	\end{tikzpicture}
    \caption{Apparent singularities and curves relevant to $T_2$ in region $\mathrm{I}$ for $t<0$.}
    \label{fig:masterpictureT2}
\end{figure}
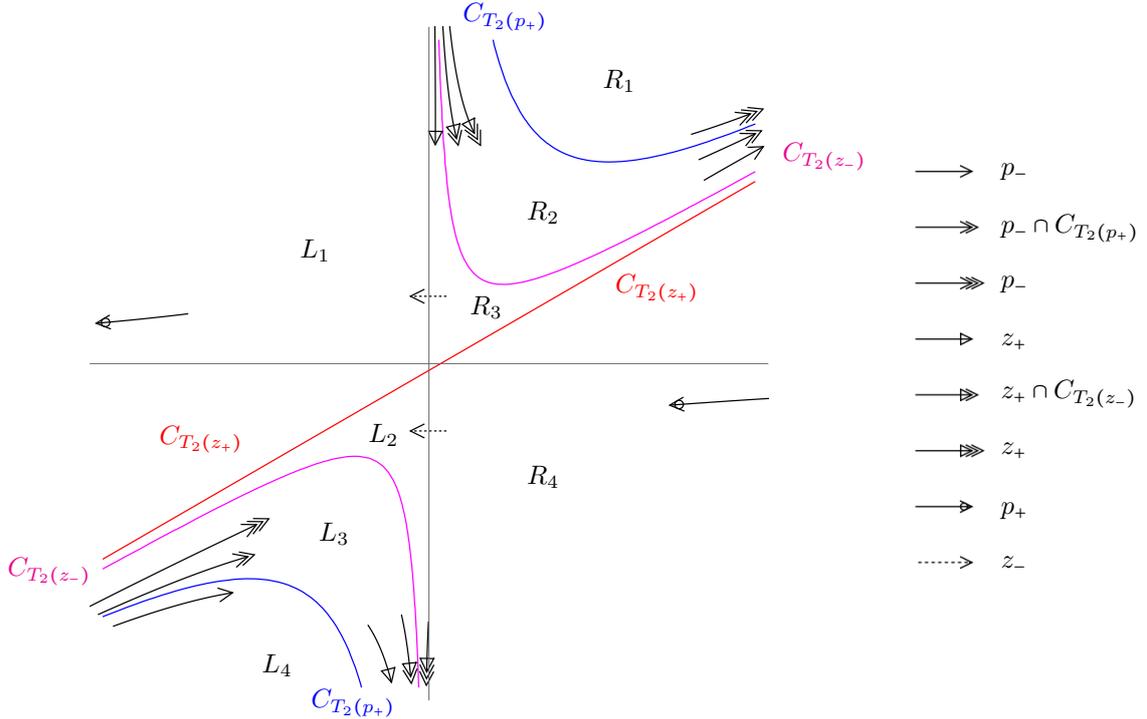

\subsubsection{Boundary conditions at $t=0$ and $t\rightarrow -\infty$}

\begin{lemma} \label{lem:asymptoticregionT2}
For parameters sufficiently large negative $t$, the generalised Okamoto rational solution $(f_{m,n}(t),g_{m,n}(t))$ lies in region $R_1$.
\end{lemma}

\begin{lemma}
For parameters $a$ in region $\mathrm{I}$, if a solution has a special apparent singularity at the origin $t=0$ of type $p_{\pm}'$ or $z_{\pm}'$, then for sufficiently small negative $t$ the solution $(f(t),g(t))$ will lie in the following regions:
\begin{equation*}
z_-'  : R_4, \qquad z_+' : L_3, \qquad p_-' : R_2, \qquad p_+' : L_1.
\end{equation*}
\end{lemma}

\subsection{The inductive arguments for $T_2$ in region $\mathrm{I}$}

Similarly to the case of $T_1$ above, we note that since we are applying $T_2$ to an old solution in region $\mathrm{I}$ the new solution will be for parameters $a$ with $a_2>1$, $a_1>0$ so we can argue using Figure \ref{fig:masterpictureT2}.

\begin{lemma} \label{lem:asymptoticT2}
A new sequence $(p_-\,z_+\,z_-)$ must occur before the first new $p_+$ (first old $p_-$).
This corresponds to the following movement of the solution curve through the regions indicated above:
\begin{equation} \label{eq:regionsequencestartoT2}
    \begin{tikzcd}
        R_1 \arrow[r,"p_-"] & L_3 \arrow[r,"z_+"] & R_3 \arrow[r,"z_-"] & L_1 
    \end{tikzcd}
\end{equation}
\end{lemma}
\begin{proof}
    Note that by Lemma \ref{lem:asymptoticregionT2} the new solution starts in region $R_1$, then has to make its way to $L_1$ in order to have the first new $p_+$ coming from the first old $p_-$, without crossing any of the curves $C_{T_2(p_+)}$, $C_{T_2(z_-)}$, $C_{T_2(z_+)}$.
    It is straightforward to check that the  sequence of regions in \eqref{eq:regionsequencestartoT2} is the only possibility.
\end{proof}

\begin{lemma} \label{lem:T2region1:4cycles1}
For parameters in region $\mathrm{I}$,
\begin{equation*}
    (p_-\, z_+\, z_-\, p_+\, p_-) \xrightarrow{T_2} (p_+\, p_-\, z_+\, z_- \,p_+),
\end{equation*}
using the notation introduced in \eqref{eq:5.10}. 
The new solution passes through regions as follows:
\begin{equation*} 
    \begin{tikzcd}
    & & & & R_1 \arrow[dr, "p_-"] & & & & \\
~ \arrow[r, "p_+"]  & R_4 \arrow[r,"C_{T_2(z_+)}"]  & R_3 \arrow[r,"C_{T_2(z_-)}"] & R_2 \arrow[rr,"p_-\cap C_{T_2(p_+)}"] \arrow[ur,"C_{T_2(p_+)}"] \arrow[dr,swap, "p_-"] &  & L_3 \arrow[r,"z_+"] & R_3 \arrow[r,"z_-"] & L_1 \arrow[r,"p_+"] &~ \\
    & & & & L_4 \arrow[ur, swap, "C_{T_2(p_+)}"]& & & & 
    \end{tikzcd}
\end{equation*}
\end{lemma}
\begin{proof}
We have the following possible paths,
\begin{itemize}
    \item $R_4 \xrightarrow{C_{T_2(z_+)}} R_3 \xrightarrow{z_-} L_1 \quad \lightning$
    \item $R_4 \xrightarrow{z_-} L_2 \xrightarrow{C_{T_2(z_+)}} L_1 \quad \lightning$
    \item $R_4 \xrightarrow{C_{T_2(z_+)}} R_3 \xrightarrow{C_{T_2(z_-)}} R_2 \xrightarrow{C_{T_2(p_+)}} R_1 \xrightarrow{p_-} L_3 \xrightarrow{z_+} R_3 \xrightarrow{z_-} L_1\quad \checkmark$  
    \item $R_4 \xrightarrow{C_{T_2(z_+)}} R_3 \xrightarrow{C_{T_2(z_-)}} R_2 \xrightarrow{p_- \cap C_{T_2(p_+)}} L_3 \xrightarrow{z_+} R_3 \xrightarrow{z_-} L_1 \quad \checkmark$  
    \item $R_4 \xrightarrow{C_{T_2(z_+)}} R_3 \xrightarrow{C_{T_2(z_-)}} R_2  \xrightarrow{p_-} L_4\xrightarrow{C_{T_2(p_+)}} L_3 \xrightarrow{z_+} R_3 \xrightarrow{z_-} L_1\quad \checkmark$  
\end{itemize}
In each of the three admissible paths that do not lead to a contradiction, the singularity signature is as claimed in the lemma.
\end{proof}

\begin{lemma} \label{lem:T2region1:2cycles1}
For parameters in region $\mathrm{I}$,
\begin{equation*}
    (p_-, z_+\, p_-) \xrightarrow{T_2} (p_+\,z_-,\,p_+),
\end{equation*}
using the notation introduced in \eqref{eq:5.10}. 
The new solution passes through regions as follows:
\begin{equation*}
    \begin{tikzcd}
                    &       & L_2 \arrow[dr,"C_{T_2(z_+)}"]  &         ~   \\
    ~ \arrow[r,"p_+"] & R_4 \arrow[ur,"z_-"]\arrow[rr,"z_-\cap C_{T_2(z_+)}"]\arrow[dr,swap,"C_{T_2(z_+)}"] &       & L_1  \arrow[r,"p_+"]   &~  \\
                    &       & R_3 \arrow[ur,swap, "z_-"]  &          ~
    \end{tikzcd}
\end{equation*}
\end{lemma}

\begin{proof}
We have the following possible paths,
    \begin{itemize}
    \item $R_4 \xrightarrow{z_-} L_2 \xrightarrow{C_{T_2(z_+)}} L_1 \quad \checkmark$
    \item $R_4 \xrightarrow{C_{T_2(z_+)}\cap z_-} L_1 \quad \checkmark$
    \item $R_4 \xrightarrow{C_{T_2(z_+)}} R_3 \xrightarrow{z_-} L_1 \quad \checkmark$
    \item $R_4 \xrightarrow{z_-} L_2 \xrightarrow{z_+} R_2 \xrightarrow{p_-} L_4 \quad \lightning$
\end{itemize}
In each of the three admissible paths that do not lead to a contradiction, the singularity signature of the new solution is $(p_+\, z_- \,p_+)$ on the relevant interval and thus the lemma follows.
\end{proof}

With the above two lemmas in hand we can prove the first two of the inductive steps:
\begin{lemma}[$m$ odd, $n$ even $\xrightarrow{T_2}$ $m$ even, $n$ even]
\label{lem:oddeven2}
Under the inductive hypothesis
    \begin{equation*}
        \mathfrak{S}(q_{2\mu+1,2\nu}) =   (p_-\,z_+\,z_-\,p_+)^{\mu} (p_-\,z_+)^{\nu}\,p_- \,\hat{z}_+\,p_-(z_+\,p_-)^{\nu}(p_+\,z_-\,z_+\,p_-)^{\mu},    
    \end{equation*}
we have:
    \begin{equation*}
        \mathfrak{S}(q_{2\mu+2,2\nu}) =(p_-\,z_+\,z_-\,p_+)^{\mu+1}\,(z_-\,p_+)^{\nu}\,\hat{z}_-\,(p_+\,z_-)^{\nu}\,(p_+\,z_-\,z_+,\,p_-)^{\mu+1}.
    \end{equation*}
\end{lemma}
\begin{proof}
    Apply Lemma \ref{lem:asymptoticT2} to obtain the first substring $(p_- z_+ z_- p_+)$ for the new solution, with the $p_+$ occurring at the first old $p_-$.
    Then apply Lemma \ref{lem:T2region1:4cycles1} 
    $\mu$ times followed by Lemma \ref{lem:T2region1:2cycles1} $\nu+1$ times before arriving at the origin, which we know to be a new $z_-$.
\end{proof}

\begin{lemma}[$m$ odd, $n$ odd $\xrightarrow{T_2}$ $m$ even, $n$ odd]
\label{lem:oddodd2}
Under the inductive hypothesis
    \begin{equation*}
        \mathfrak{S}(q_{2\mu+1,2\nu+1}) =  (p_-\,z_+\,z_-\,p_+)^{\mu}(p_-\,z_+)^{\nu+1}\,\hat{p}_-\,(z_+\,p_-)^{\nu+1}(p_+\,z_-\,z_+\,p_-)^{\mu},    
    \end{equation*}
we have:
    \begin{equation*}
        \mathfrak{S}(q_{2\mu+2,2\nu+1}) =(p_-\,z_+\,z_-\,p_+)^{\mu+1}\,(z_-\,p_+)^{\nu}\,z_-\,\hat{p}_+\, z_-\,(p_+\,z_-)^{\nu}\,(p_+\,z_-\,z_+\,p_-)^{\mu+1}.
    \end{equation*}
\end{lemma}
\begin{proof}
    Similarly to the previous lemma this requires only Lemmas \ref{lem:asymptoticT2}, \ref{lem:T2region1:4cycles1} 
    and \ref{lem:T2region1:2cycles1} as well as the known behaviour of the new solution at the origin.
\end{proof}

We will need slightly more complicated lemmas to deal with the inductive steps for the remaining parities, since in the factor $(z_-,\,p_+)^{\nu}$ in both cases we do not have apparent singularities of old solutions that correspond directly to apparent singularities of new ones. 
Rather we have the old $z_-$ and $p_+$ corresponding to curve crossings, the locations of which in relation to the new apparent are not determined based only on the inductive hypothesis.
Therefore we deal with the two-cycles all at once rather than individually with the following two lemmas, which are proved along the same lines as the others.
\begin{lemma}  \label{lem:T2region1:zonelemma1}
For parameters in region $\mathrm{I}$, and positive integer $\ell$,
\begin{equation*}
    p_-\,z_+ (z_-\,p_+)^{\ell} \hat{z}_-
    \xrightarrow{T_2}
    p_+\,(p_-\,z_+)^{\ell-1} p_- \hat{z}_+,
\end{equation*}
using the notation introduced in \eqref{eq:5.10}. 
The new solution passes through regions as follows:
\begin{equation*} 
    \begin{tikzcd}
     & &    & & &  R_1 \arrow[dr,"p_-"]& &     \\
~\arrow[r,"p_+"] & R_4 \arrow[r,"C_{T_2(z_+)}"] & R_3 \arrow[r,"C_{T_2(z_-)}"] &\big[R_2 \arrow[r,"\mathcal{M}"]  & 
R_2\big]^{\ell-1} \arrow[rr,"p_-\cap C_{T_2(p_+)}"]  \arrow[ur,"C_{T_2(p_+)}"] \arrow[dr,swap,"p_-"] & & L_3  \arrow[r,"\hat{z}_+"] & ~\\
      & &    & & & L_4 \arrow[ur,swap,"C_{T_2(p_+)}"]& & 
    \end{tikzcd}
\end{equation*}
where we have used $R_2 \xrightarrow{\mathcal{M}} R_2$ to denote the following movement between regions, which includes a new $p_-$ followed by a new $z_+$: 
\begin{equation*}  
    \begin{tikzcd}[row sep=1cm,column sep=1cm]
   & & R_1 \arrow[dr,"p_-"]& & L_2 \arrow[dr,"z_+"] & ~ \\
\mathcal{M}: &  R_2 \arrow[rr,"p_-\cap C_{T_2(p_+)}"]  \arrow[ur,"C_{T_2(p_+)}"] \arrow[dr,swap,"p_-"] & & L_3 \arrow[rr,"z_+\cap C_{T_2(z_-)}"]  \arrow[ur,"C_{T_2(z_-)}"] \arrow[dr,swap,"z_+"] & & R_2 \\
  & & L_4 \arrow[ur,swap,"C_{T_2(p_+)}"]& & R_3 \arrow[ur,swap,"C_{T_2(z_-)}"] & ~
    \end{tikzcd}\qquad \quad
\end{equation*}
\end{lemma}

\begin{lemma} \label{lem:T2region1:zonelemma2}
    For parameters in region $\mathrm{I}$, and positive integer $\ell$,
\begin{equation*}
    p_-\,z_+(z_-\,p_+)^{\ell} z_- \hat{p}_+
    \xrightarrow{T_2}
    p_+\, (p_-\,z_+)^{\ell} \hat{p}_-,
\end{equation*}
using the notation introduced in \eqref{eq:5.10}. 
The new solution passes through regions as follows:
\begin{equation*} 
    \begin{tikzcd}
~\arrow[r,"p_+"] & R_4 \arrow[r,"C_{T_2(z_+)}"] & R_3 \arrow[r,"C_{T_2(z_-)}"] &\big[R_2 \arrow[r,"M"] & 
R_2\big]^{\ell} \arrow[r,"p_-'"] & ~
    \end{tikzcd}
\end{equation*}
\end{lemma}

\begin{lemma}[$m$ even, $n$ even $\xrightarrow{T_2}$ $m$ odd, $n$ even]
\label{lem:eveneven2}
Under the inductive hypothesis
    \begin{equation*}
        \mathfrak{S}(q_{2\mu,2\nu}) = (p_-\,z_+\,z_-\,p_+)^{\mu}\,(z_-\,p_+)^{\nu}\,\hat{z}_-\,(p_+\,z_-)^{\nu}\,(p_+\,z_-\,z_+,\,p_-)^{\mu},    
    \end{equation*}
we have:
    \begin{equation*}
        \mathfrak{S}(q_{2\mu+1,2\nu}) =(p_-\,z_+\,z_-\,p_+)^{\mu} (p_-\,z_+)^{\nu}\,p_- \,\hat{z}_+\,p_-(z_+\,p_-)^{\nu}(p_+\,z_-\,z_+\,p_-)^{\mu}.
    \end{equation*}
\end{lemma}
\begin{proof}
    After applying Lemma \ref{lem:asymptoticT2} to obtain the leading $(p_-\,z_+\,z_-)$ before the first $p_+$ of the new solution, we apply Lemma \ref{lem:T2region1:4cycles1} to the factor $(p_-\,z_+\,z_-\,p_+)^{\mu-1}$ in the singularity signature of the old solution.
    Then the remaining part of the singularity signature of the new solution up to and including the origin will come from the factor $p_-\,z_+\,z_-\,p_+\,(z_-\,p_+)^{\nu} \hat{z}_-$ of the old one.
    This we deal with using Lemma \ref{lem:T2region1:zonelemma1} with $\ell=\nu+1$ and we have the result.
    \end{proof}


\begin{lemma}[$m$ even, $n$ odd $\xrightarrow{T_2}$ $m$ odd, $n$ odd]
\label{lem:evenodd2}
Under the inductive hypothesis
    \begin{equation*}
        \mathfrak{S}(q_{2\mu,2\nu+1}) =  (p_-\,z_+\,z_-\,p_+)^{\mu}\,(z_-\,p_+)^{\nu}\,z_-\,\hat{p}_+\, z_-\,(p_+\,z_-)^{\nu}\,(p_+\,z_-\,z_+\,p_-)^{\mu},    
    \end{equation*}
we have:
    \begin{equation*}
        \mathfrak{S}(q_{2\mu+1,2\nu+1}) =(p_-\,z_+\,z_-\,p_+)^{\mu}(p_-\,z_+)^{\nu+1}\,\hat{p}_-\,(z_+\,p_-)^{\nu+1}(p_+\,z_-\,z_+\,p_-)^{\mu}.
    \end{equation*}
\end{lemma}
\begin{proof}
    By Lemmas \ref{lem:asymptoticT2} and \ref{lem:T2region1:4cycles1} as well as \ref{lem:T2region1:zonelemma2}.
\end{proof}

We have now completed all four of the inductive steps required for $(m,n)$ in region $\mathrm{I}$ and the proof of Theorem \ref{th:region1} is complete.

\section{Proof for region $\mathrm{II}$} \label{sec:region2}

We will give the inductive proof for region $\mathrm{II}$, where $a_1>0$, $a_2<0$, $a_1+a_2>1$, using $T_1$ and $\hat{T}_1 = T_1 T_{2}^{-1} : q_{m,n}\mapsto q_{m-1,n+1}$, with base case being $(m,n)=(-1,2)$.
Again we omit proofs in cases when they are analogous to those given in detail previously.


The solution
\begin{equation*}
    q_{-1,2}(t)=-\frac{2 t(16 t^8-504 t^4-567)}{3(4t^4+12t^2-9)(4t^4-12t^2-9)},
\end{equation*}
has seven real apparent singularities, $t_{-3}<t_{-2}<\ldots<t_2<t_3$, given by
\begin{align*}
    t_{\pm 3}&=\pm (\tfrac{63}{4}+9\sqrt{\tfrac{7}{2}})^{\frac{1}{4}},\\
    t_{\pm 2}&=\pm\sqrt{\tfrac{3}{2}(\sqrt{2}+1)},\\
    t_{\pm 1}&=\pm \sqrt{\tfrac{3}{2}(\sqrt{2}-1)},\\
    t_0&=0,
\end{align*}
with corresponding singularity signature
\begin{equation*}
    \mathfrak{S}(q_{-1,2})=
    z_-\, p_+\, p_-\, \hat{z}_+\, p_-\, p_+\, z_-,
\end{equation*}
see Figure \ref{fig:okamotorat-1,2}.

\begin{figure}[ht]
    \centering
\includegraphics[width=.8\textwidth]{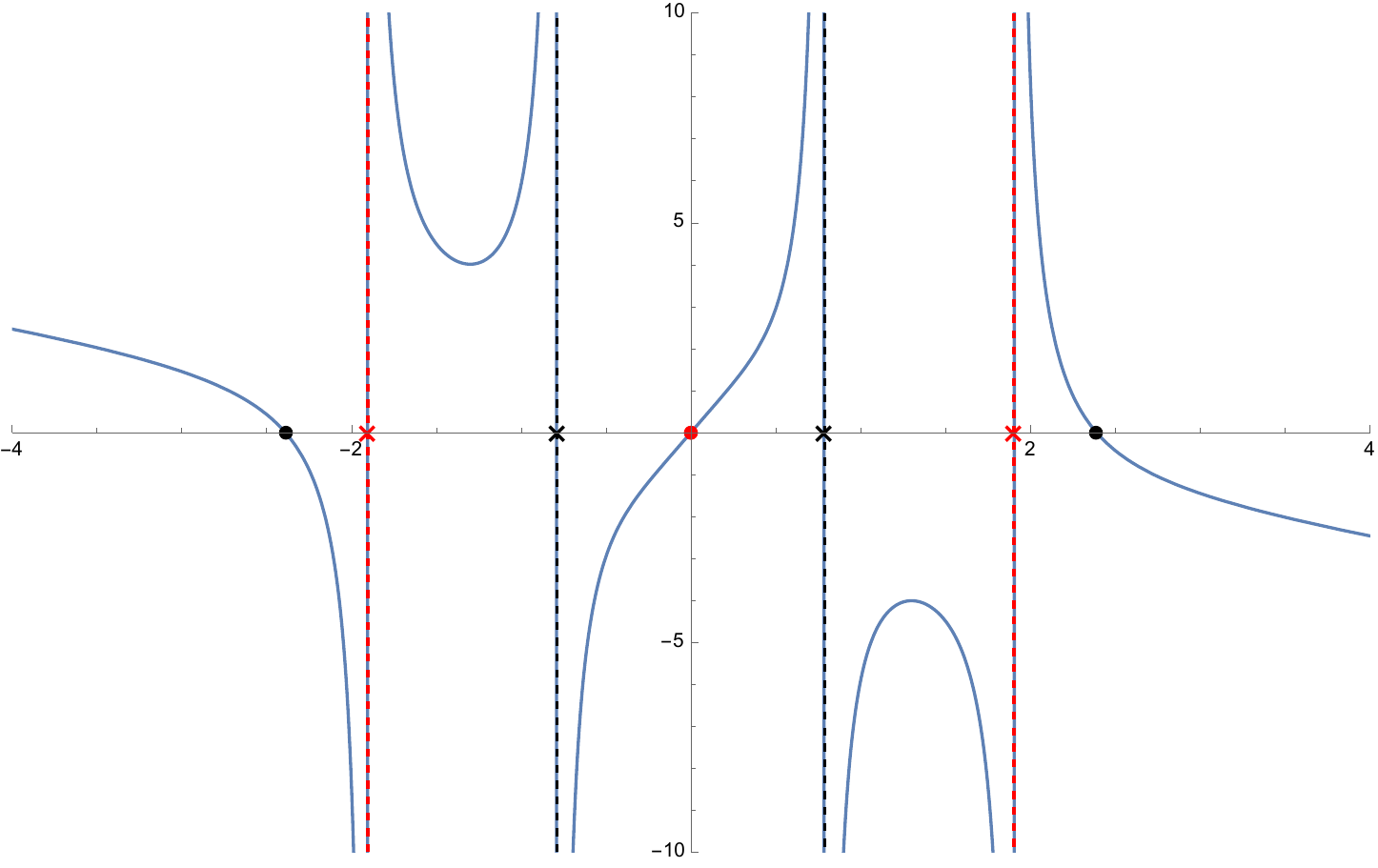}           
    \caption{Generalised Okamoto rational $q_{-1,2}(t)$}
    \label{fig:okamotorat-1,2}
\end{figure}

\subsection{Translation $\hat{T}_1$ in region $\mathrm{II}$}

The transformation $\hat{T}_1 = T_1 T_{2}^{-1}$ we will use in this section can be obtained by composing the B\"acklund transformations given by (\ref{BTT1}-\ref{BTT2}).
The pushforward $(\hat{T}_1)_*$ on the Picard groups can be calculated directly, and the action of $\hat{T}_1$ on apparent singularities is as follows:
\begin{equation*}
   \hat{T}_1 : \left\{
    \begin{aligned}
        p_+ &\mapsto z_+, \\
        p_- &\mapsto C_{\hat{T}_1(p_-)} : g = 0, \\
        z_+ &\mapsto C_{\hat{T}_1(z_+)} : g^2 - g(2t + f) - 2a_2 = 0, 
        \\
        z_- &\mapsto p_-.
    \end{aligned}
    \right.
\end{equation*}


 \begin{lemma} \label{lem:regionsT1hat}
For parameters $a$ in region $\mathrm{II}$, and $t<0$, the curves 
\begin{equation*}
\begin{aligned}
    C_{\hat{T}_1(p_-)} &: g = 0, 	\\
    C_{\hat{T}_1(z_+)} &: g^2 - g(2t + f) - 2a_2 = 0, \\
    C_{z_-} &: f=0,
\end{aligned}
\end{equation*}
 divide the real $(f,g)$-plane into the regions shown in Figure \ref{fig:regionsThat1}. 
 We separate cases depending on the sign of $t^2+2a_2$ because this determines whether $C_{\hat{T}_1(z_+)}$ intersects the vertical axis. 
 The case when $t^2+2a_2=0$ is similar. We formally set $L_4:=L_2$ and $R_4=\emptyset$ when $t^2+2a_2<0$.
\end{lemma}

\begin{figure}[htb]
 \begin{subfigure}{0.4\textwidth}
     \begin{tikzpicture}
    \node (pic) at (0,0) {\includegraphics[width=\textwidth]{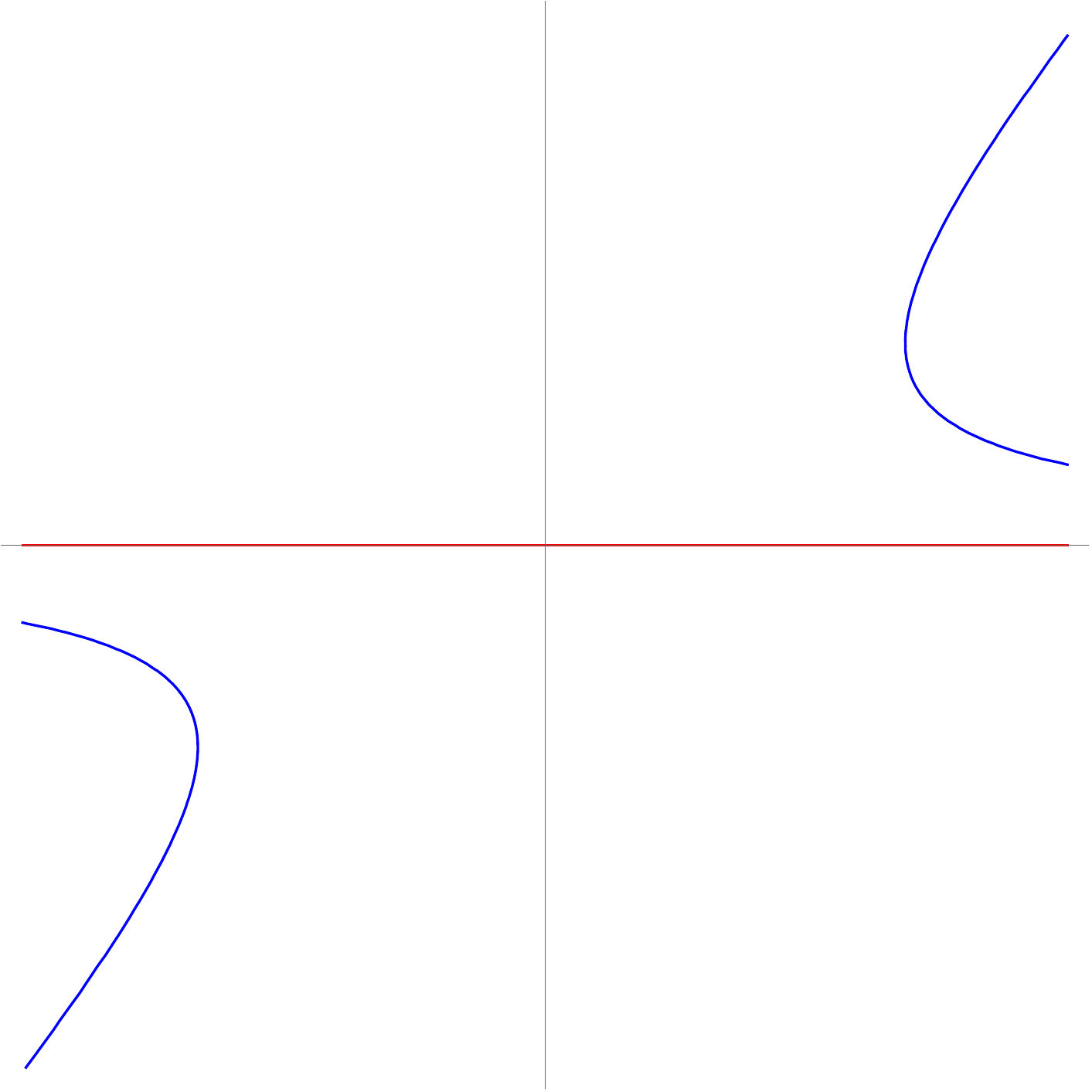}};
    \node (L1) at (-1.5,1.5) {$L_1$};
    \node (L2) at (-1.5,-1.5) {$L_2$};
    \node (L3) at (-2.75,-1.5) {$L_3$};
    \node (R1) at (2.75,1.5) {$R_2$};
    \node (R2) at (1.5,1.5) {$R_1$};
    \node (R3) at (1.5,-1.5) {$R_3$};
    \node (bluecurvelabelL) at (-2,-2.75) {{\color{blue}$C_{\hat{T}_1(z_+)}$}};    
    \node (bluecurvelabelR) at (2,2.75) {{\color{blue}$C_{\hat{T}_1(z_+)}$}};    
    \node (redcurvelabelL) at (-2.5,.5) {{\color{red}$C_{\hat{T}_1(p_-)}$}};
    \node (redcurvelabelR) at (2.5,-.5) {{\color{red}$C_{\hat{T}_1(p_-)}$}};
    \node (f0label) at (0,3.5) {$f=0$};    
    \end{tikzpicture}
\caption{If $t^2+2a_2 < 0$}
\label{subfig:regionsThat1easy}
 \end{subfigure}
\qquad
\begin{subfigure}{0.5\textwidth}
     \begin{tikzpicture}
    \node (pic) at (0,0) {\includegraphics[width=\textwidth]{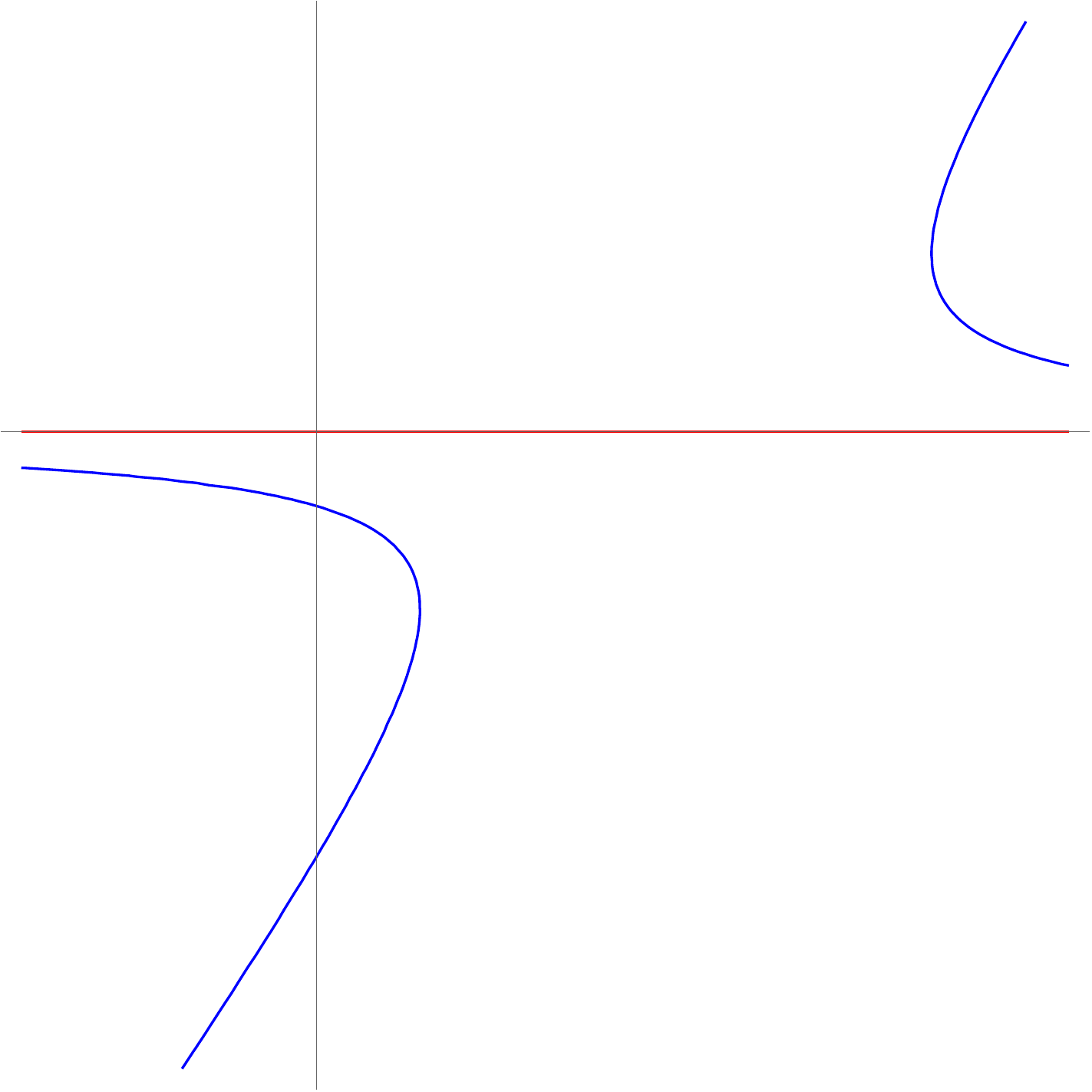}};
    \node (L1) at (-3,3) {$L_1$};
    \node (L2) at (-2.5,-.5) {$L_2$};
 \draw[->] (L2) -- (-1.8,.6) ; 
    \node (L3) at (-3,-1.5) {$L_3$};
    \node (L4) at (-2,-3.5) {$L_4$};    
    \node (R1) at (3.25,2) {$R_2$};
    \node (R2) at (1,2) {$R_1$};
    \node (R3) at (1,-1.5) {$R_3$};
    \node (R4) at (-1.25,-.5) {$R_4$};    
    \node (bluecurvelabelL) at (-3,-3.3) {{\color{blue}$C_{\hat{T}_1(z_+)}$}};    
    \node (bluecurvelabelR) at (2.25,3) {{\color{blue}$C_{\hat{T}_1(z_+)}$}};    
    \node (redcurvelabelL) at (-3,1.25) {{\color{red}$C_{\hat{T}_1(p_-)}$}};
    \node (redcurvelabelR) at (3,0.25) {{\color{red}$C_{\hat{T}_1(p_-)}$}};
    \end{tikzpicture}
\caption{If $t^2+2a_2 > 0$}
\label{subfig:regionsThat1hard}
 \end{subfigure}
    \caption{Regions for $\hat{T}_1$ with $a$ in region $\mathrm{II}$}
    \label{fig:regionsThat1}
\end{figure}
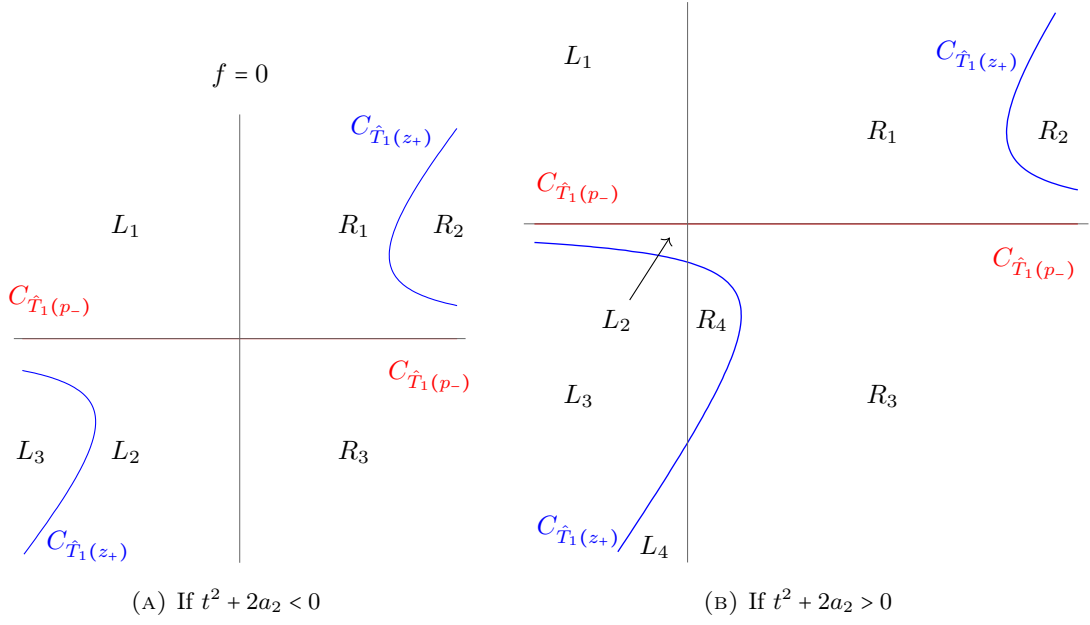

\subsubsection{Curve crossings in the real $(f,g)$-plane relevant to $\hat{T}_1$ in region $\mathrm{II}$}

\begin{lemma} \label{lem:crossingsT1hat}
When the parameters $a$ lie in region $\mathrm{II}$ with $a_1>1$, apparent singularities of real solutions at $t=t_*<0$ correspond to the following behaviours in the real $(f,g)$-plane, as shown in Figure \ref{fig:masterpictureT1hat}.

\begin{itemize}
    \item $z_+$ is a crossing from $L_4$ to $R_1$.
    \item $z_-$ is a crossing from $\cup_i R_i$ to $\cup_i L_i$.
    \item $p_+$ is either
    \begin{itemize}
        \item a crossing from $L_2$ to $R_2$,
        \item a crossing from $L_2$ to $R_1$, with the curve $C_{\hat{T}_1(z_+)}$ on $X$ crossed simultaneously,
        \item a crossing from $L_3$ to $R_1$.
    \end{itemize}
    \item $p_-$ is a crossing from $R_1$ to $L_4$.
\end{itemize}
\end{lemma}

\begin{figure}[htb]
    \begin{subfigure}{0.95\textwidth}
	\begin{tikzpicture}
    \node (pic) at (0,0) {\includegraphics[width=.6\textwidth]{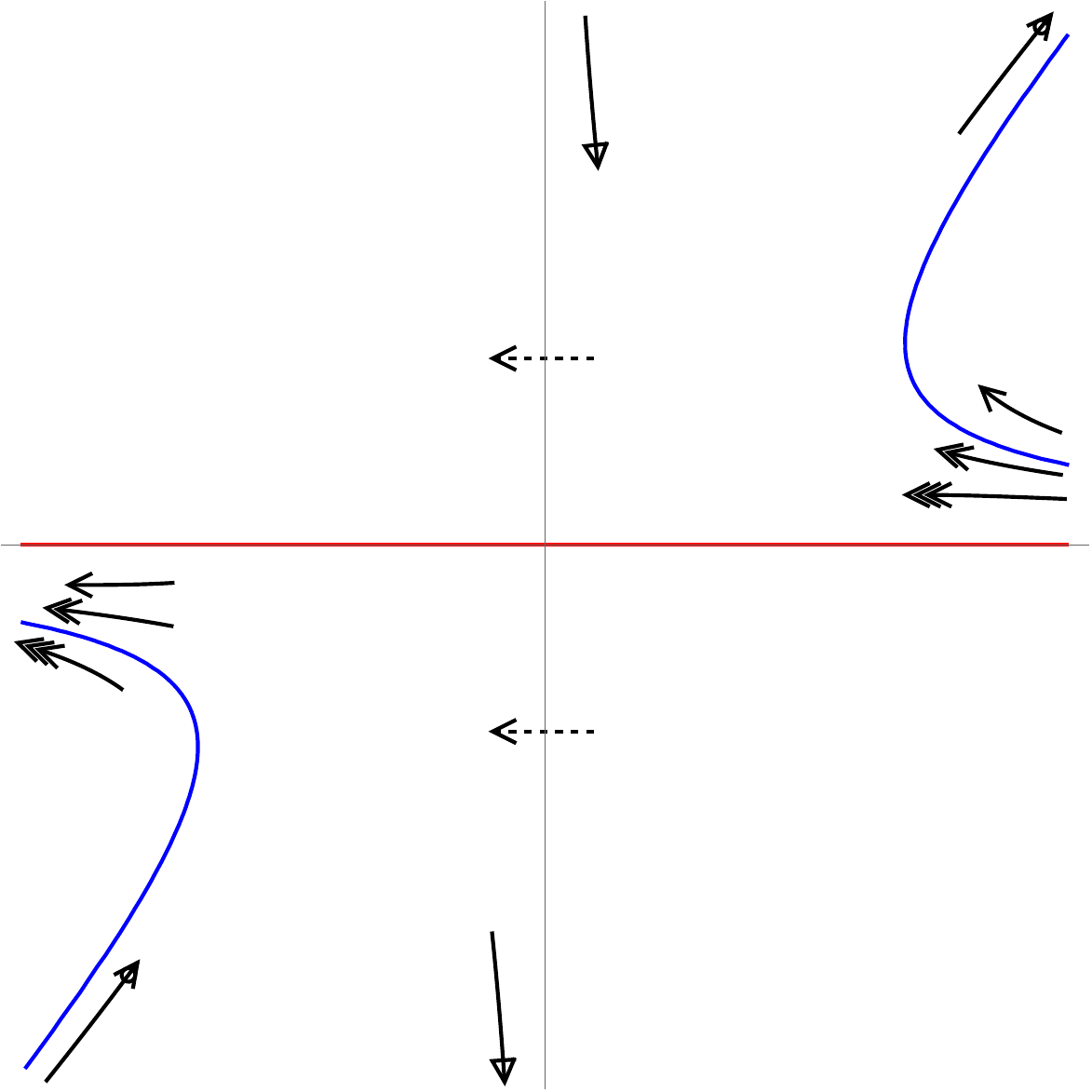}};
    \node (L1) at (-1.5,1.5) {$L_1$};
    \node (L2) at (-1.5,-1.5) {$L_2$};
    \node (L3) at (-3.5,-2) {$L_3$};
    \node (R1) at (3.5,2) {$R_2$};
    \node (R2) at (1.5,1.5) {$R_1$};
    \node (R3) at (1.5,-1.5) {$R_3$};
    \node (bluecurvelabelL) at (-2.5,-2.75) {{\color{blue}$C_{\hat{T}_1(z_+)}$}};    
    \node (bluecurvelabelR) at (2.5,2.75) {{\color{blue}$C_{\hat{T}_1(z_+)}$}};    
    \node (redcurvelabelL) at (-2.5,.5) {{\color{red}$C_{\hat{T}_1(p_-)}$}};
    \node (redcurvelabelR) at (2.5,-.5) {{\color{red}$C_{\hat{T}_1(p_-)}$}};
\begin{scope}[xshift=.2cm]       
    \node (legend) at (6.6,-0.05) {\includegraphics[width=.073\textwidth]{arrowslegendT1better.pdf}};
    \node (pplus1) at (7.2,1.7) [right] {$p_+$};
    \node (pplus2) at (7.2,1) [right] {$p_+ \cap C_{\hat{T}_1(z_+)}$};
    \node (pplus3) at (7.2,.3) [right] {$p_+$};
    \node (zplus) at (7.2,-.4) [right] {$z_+$};
    \node (pminus) at (7.2,-1.1) [right] {$p_-$};
    \node (zplus) at (7.2,-1.8) [right] {$z_-$};
\end{scope}
	\end{tikzpicture}
 \caption{$t^2+2a_2<0$}
 \label{subfig:masterpictureT1hateasy}
 \end{subfigure}
 \\
     \begin{subfigure}{0.95\textwidth}
	\begin{tikzpicture}
    \node (pic) at (0,0) {\includegraphics[width=.6\textwidth]{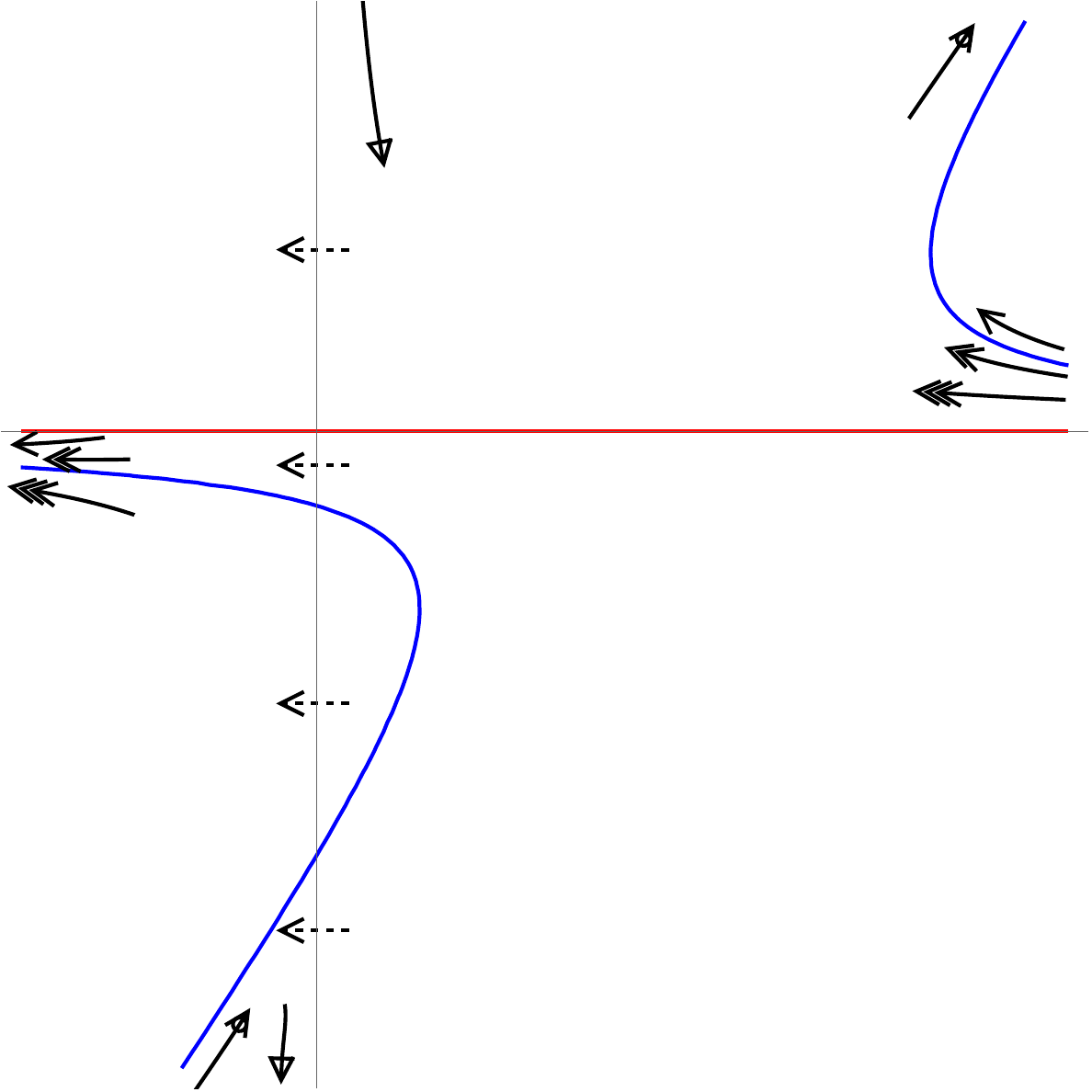}};
    \node (L1) at (-3,3) {$L_1$};
    \node (L2) at (-2.5,.68) {$L_2$};
    \node (L3) at (-3,-1.5) {$L_3$};
    \node (L4) at (-2,-3.3) {$L_4$};    
    \node (R1) at (3.5,2.5) {$R_2$};
    \node (R2) at (.5,2) {$R_1$};
    \node (R3) at (.5,-1.5) {$R_3$};
    \node (R4) at (-1.35,-.5) {$R_4$};    
    \node (bluecurvelabelL) at (-3,-3.3) {{\color{blue}$C_{\hat{T}_1(z_+)}$}};    
    \node (bluecurvelabelR) at (2.3,2.5) {{\color{blue}$C_{\hat{T}_1(z_+)}$}};    
    \node (redcurvelabelL) at (-3,1.25) {{\color{red}$C_{\hat{T}_1(p_-)}$}};
    \node (redcurvelabelR) at (3,0.5) {{\color{red}$C_{\hat{T}_1(p_-)}$}};
\begin{scope}[xshift=.2cm]       
    \node (legend) at (6.6,-0.05) {\includegraphics[width=.073\textwidth]{arrowslegendT1better.pdf}};
    \node (pplus1) at (7.2,1.7) [right] {$p_+$};
    \node (pplus2) at (7.2,1) [right] {$p_+ \cap C_{\hat{T}_1(z_+)}$};
    \node (pplus3) at (7.2,.3) [right] {$p_+$};
    \node (zplus) at (7.2,-.4) [right] {$z_+$};
    \node (pminus) at (7.2,-1.1) [right] {$p_-$};
    \node (zplus) at (7.2,-1.8) [right] {$z_-$};
\end{scope}
	\end{tikzpicture}
    \caption{$t^2+2a_2>0$}
  \label{subfig:masterpictureT1hathard}
   \end{subfigure}
    \caption{Apparent singularities relative to curves for $\hat{T}_1$, for $(m,n)$ in region $\mathrm{II}$ with $t<0$.}
    \label{fig:masterpictureT1hat}
\end{figure}

\subsubsection{Boundary conditions at $t=0$ and $t\rightarrow -\infty$}

\begin{lemma} \label{lem:asymptoticregionT1hat}
For sufficiently large negative $t$, the Okamoto rational solution $(f_{m,n}(t),g_{m,n}(t))$ lies in region $R_4$ on Figure \ref{subfig:regionsThat1hard}.
\end{lemma}

\begin{lemma} \label{lem:originT1hat}
If a solution has a special apparent singularity at the origin $t=0$ of type $p_{\pm}'$ or $z_{\pm}'$, then for sufficiently small negative $t$ the solution $(f(t),g(t))$ will lie in the following regions as indicated in Figure \ref{subfig:masterpictureT1hateasy}:
\begin{equation*}
z_-'  : R_1, \qquad z_+' : L_2, \qquad p_-' : R_1, \qquad p_+' : L_2.
\end{equation*}
\end{lemma}

\subsection{The inductive arguments for $\hat{T}_1$ in region $\mathrm{II}$}

\begin{lemma} \label{lem:asymptoticThat1}
A new sequence $(z_-\,p_+)$ must occur before the first new $p_-$ (first old $z_-$).
This corresponds to the following movement of the solution curve through the regions indicated above:
\begin{equation} \label{eq:regionsequencestartT1hat}
    \begin{tikzcd}
        R_4 \arrow[r,"z_-"] & L_3 \arrow[r,"p_+"] & R_1 
    \end{tikzcd}
\end{equation}
\end{lemma}
\begin{proof}
By \ref{lem:asymptoticregionT1hat}, the new solution starts of in the  region $R_4$ in Figure \ref{subfig:regionsThat1hard}. The first apparent singularity of the old solution is a $z_-$ and the new solution will have a $p_-$ singularity there. Thus, the new solution will have to make its way from $R_4$ to $R_1$ before this $p_-$ singularity, without crossing the curves $C_{\hat{T}_1(z_+)}$ and $C_{\hat{T}_1(p_-)}$ nor having a $z_+$ or $p_-$ singularity.
This allows only for the movement of the solution curve as shown in \eqref{eq:regionsequencestartT1hat} and the lemma follows.
\end{proof}

\begin{lemma} \label{lem:T1hatregion2:4cycles1}
For parameters in region $\mathrm{II}$,
\begin{equation*}
    (z_-\, p_+\, p_-\, z_+\, z_-) \xrightarrow{\hat{T}_1} (p_-\, z_+\, z_-\, p_+ \,p_-),
\end{equation*}
using the notation introduced in \eqref{eq:5.10}. 
The new solution passes through regions as follows:
\begin{equation*}
    \begin{tikzcd}
    & & &  L_1 \arrow[dr, "C_{\hat{T}_1(p_-)}"] & & R_2 \arrow[dr, "C_{\hat{T}_1(z_+)}"] & & \\
~ \arrow[r, "p_-"]  & L_2\cup L_4 \arrow[r,"z_+"]  & R_1 \arrow[rr,"z_-\cap C_{\hat{T}_1(p_-)}"] \arrow[ur,"z_-"] \arrow[dr,swap, "C_{\hat{T}_1(p_-)}"] &  & L_2\cup L_4 \arrow[rr,"p_+\cap C_{\hat{T}_1(z_+)}"] \arrow[ur,"p_+"] \arrow[dr,swap, "C_{\hat{T}_1(z_+)}"] & & R_1 \arrow[r,"p_-"] &~ \\
    & &  & R_3 \arrow[ur, swap, "z_-"]& & L_3 \arrow[ur, swap, "p_+"] & & 
    \end{tikzcd}
\end{equation*}
where we remind the reader that $L_4$ denotes the empty set if $t^2+2a_2<0$.
\end{lemma}

\begin{lemma} \label{lem:T1hatregion2:zonelemma1}
For parameters in region $\mathrm{II}$ and positive integer $\ell$,
\begin{equation*}
    (z_-\, p_+\, p_-)\, (z_+\, p_-)^{\ell}\,\hat{z}_+\xrightarrow{\hat{T}_1} (p_-\, z_+\, z_-)\, (p_+\, z_-)^{\ell}\,\hat{p}_+,
\end{equation*}
using the notation introduced in \eqref{eq:5.10}. 
The new solution passes through regions as follows:
\begin{equation*} 
    \begin{tikzcd}
    & & &  L_1 \arrow[dr, "C_{\hat{T}_1(p_-)}"] & &  & \\
~ \arrow[r, "p_-"]  & L_2\cup L_4 \arrow[r,"z_+"]  & R_1 \arrow[rr,"z_-\cap C_{\hat{T}_1(p_-)}"] \arrow[ur,"z_-"] \arrow[dr,swap, "C_{\hat{T}_1(p_-)}"] &  & \big[L_2\cup L_4 \arrow[r,"\mathcal{M}"] & L_2\cup L_4\big]^{\ell} \arrow[r,"\hat{p}_+"] &~ \\
    & &  & R_3 \arrow[ur, swap, "z_-"]& &  & 
    \end{tikzcd}
\end{equation*}
where we have used $L_2 \cup L_4 \xrightarrow{\mathcal{M}} L_2 \cup L_4$ to denote the following movement between regions, which includes a new $p_+$ followed by a new $z_-$: 
\begin{equation*} 
   \begin{tikzcd}[row sep=1cm,column sep=1cm]
  & & R_2 \arrow[dr,"C_{\hat{T}_1(z_+)}"]& & L_1 \arrow[dr,"C_{\hat{T}_1(p_-)}"] & ~ \\
\mathcal{M} : &  L_2 \cup L_4 \arrow[rr,"p_+\cap C_{\hat{T}_1(z_+)}"]  \arrow[ur,"p_+"] \arrow[dr,swap,"C_{\hat{T}_1(z_+)}"] & &  R_1 \arrow[rr,"z_-\cap C_{\hat{T}_1(p_-)}"] \arrow[ur,"z_-"] \arrow[dr,swap, "C_{\hat{T}_1(p_-)}"]  & & L_2 \cup L_4 \\
  &  & L_3 \arrow[ur,swap,"p_+"]& & R_3 \arrow[ur,swap,"z_-"] & ~
    \end{tikzcd}
    \qquad \quad
\end{equation*}
\end{lemma}

\begin{lemma} \label{lem:T1hatregion2:zonelemma2}
For parameters in region $\mathrm{II}$ and positive integer $\ell$,
\begin{equation*}
    (z_-\, p_+)\, (p_-\, z_+)^{\ell}\,\hat{p}_-\xrightarrow{\hat{T}_1} (p_-\, z_+)\, (z_-\, p_+)^{\ell}\,\hat{z}_-,
\end{equation*}
using the notation introduced in \eqref{eq:5.10}. 
The new solution passes through regions as follows:
\begin{equation*} 
    \begin{tikzcd}
~ \arrow[r, "p_-"]  & L_2\cup L_4 \arrow[r,"z_+"]  & \big[R_1 \arrow[r,"\mathcal{M}"]  & R_1\big]^{\ell} \arrow[r,"\hat{z}_-"] &~
    \end{tikzcd}
\end{equation*}
where we have used $M_1 \xrightarrow{\mathcal{M}} M_1$ to denote the following movement between regions, which includes a new $z_-$ followed by a new $p_+$: 
\begin{equation*} 
   \begin{tikzcd}[row sep=1cm,column sep=1cm]
  & & L_1 \arrow[dr,"C_{\hat{T}_1(p_-)}"] & & R_2 \arrow[dr,"C_{\hat{T}_1(z_+)}"]& ~ \\
\mathcal{M} : &    R_1 \arrow[rr,"z_-\cap C_{\hat{T}_1(p_-)}"] \arrow[ur,"z_-"] \arrow[dr,swap, "C_{\hat{T}_1(p_-)}"]  & & L_2 \cup L_4 \arrow[rr,"p_+\cap C_{\hat{T}_1(z_+)}"]  \arrow[ur,"p_+"] \arrow[dr,swap,"C_{\hat{T}_1(z_+)}"] & & R_1 \\
  & & R_3 \arrow[ur,swap,"z_-"]&  & L_3 \arrow[ur,swap,"p_+"] & ~
    \end{tikzcd}
    \qquad \quad
\end{equation*}
\end{lemma}

With Lemmas \ref{lem:asymptoticregionT1hat},
\ref{lem:originT1hat},
\ref{lem:T1hatregion2:4cycles1},
\ref{lem:T1hatregion2:zonelemma1} and
\ref{lem:T1hatregion2:zonelemma2}
in hand, we can prove the four inductive steps for $\hat{T}_1$ in region $\mathrm{II}$.

\subsection{Translation $T_1$ in region $\mathrm{II}$}

It turns out that for parameters in region $\mathrm{II}$ with $a_1>1$, the division of the real $(f,g)$-plane into regions by curves relevant to $T_1$ is the same as for region $\mathrm{I}$. 
Further, the placement of apparent singularities relative to the curves is also the same so we can argue using Figure \ref{fig:masterpictureT1}.
We remark that in region $\mathrm{II}$, a $p_+$ leaves the real $(f,g)$-plane with $g<0$ and returns with $g>0$ because $a_2<0$, but this does not change that fact that a $p_+$ is a transition from $L_2$ to $R_3$.

\subsubsection{Boundary conditions at $t=0$ and $t\rightarrow -\infty$}

\begin{lemma} \label{lem:asymptoticregionT1region2}
For parameters $a$ in region $\mathrm{II}$ and sufficiently large negative $t$, the Okamoto rational solution $(f_{m,n}(t),g_{m,n}(t))$ lies in region $R_2$.
\end{lemma}

\begin{lemma}
If a solution has a special apparent singularity at the origin $t=0$ of type $p_{\pm}'$ or $z_{\pm}'$, then for sufficiently small negative $t$ the solution $(f(t),g(t))$ will lie in the following regions as indicated in Figure \ref{fig:masterpictureT1}:
\begin{equation*}
z_-'  : R_2, \qquad z_+' : L_4, \qquad p_-' : R_1, \qquad p_+' : L_2.
\end{equation*}
\end{lemma}

\subsection{The inductive arguments for $T_1$ in region $\mathrm{II}$}

\begin{lemma} \label{lem:asymptoticT1region2}
A new $z_-$ must occur before the first new $p_+$ (first old $z_-$).
This corresponds to the following movement of the solution curve through the regions indicated in Figure \ref{fig:masterpictureT1},
\begin{equation*}
    \begin{tikzcd}
        R_2 \arrow[r,"z_-"] & L_2 \arrow[r,"p_+"] & R_3.
    \end{tikzcd}
\end{equation*}
\end{lemma}

\begin{lemma} \label{lem:T1region2:4cycles1}
For parameters in region $\mathrm{II}$,
\begin{equation*}
    (z_-\, p_+\, p_-\, z_+\, z_-) \xrightarrow{T_1} (p_+\, p_-\, z_+\, z_- \,p_+),
\end{equation*}
using the notation introduced in \eqref{eq:5.10}. 
The new solution passes through regions as follows:
\begin{equation*}
    \begin{tikzcd}
    & & & & & L_1 \arrow[dr, "C_{T_1(z_+)}"] & & \\
~ \arrow[r, "p_+"]  & R_3 \arrow[r,"C_{T_1(p_+)}"]  & R_2 \arrow[r,"p_-"] & L_4 \arrow[r,"z_+"] & R_1 \arrow[rr,"z_-\cap C_{T_1(z_+)}"] \arrow[ur,"z_-"] \arrow[dr,swap, "C_{T_1(z_+)}"] &  & L_2 \arrow[r,"p_+"] &~ \\
    & & & & & R_2 \arrow[ur, swap, "z_-"]& &
    \end{tikzcd}
\end{equation*}
\end{lemma}

With Lemmas 
\ref{lem:T1:phasetransition},
\ref{lem:T1:2cycles1}, \ref{lem:T1:2cycles2}, \ref{lem:T1:phasetransition2} and
\ref{lem:asymptoticT1region2} in hand, we can prove the four inductive steps for $T_1$ in region $\mathrm{II}$ and the proof of Theorem \ref{th:region2} is complete.

\section{Proof for region $\mathrm{III}$} \label{sec:region3}

We will give the inductive proof for region $\mathrm{III}$, where $a_1>0$, $a_2<0$ and $a_0 > 0$, i.e. $a_1+a_2 < 1$, using $T_2^{-1}$ and $\hat{T}_1$.
The base cases for the induction will be $(m,n)=(-1,0)$ and $(m,n)=(-1,1)$, from which all $(m,n)$ in region $\mathrm{III}$ can be obtained by applying $T_2^{-1}$ and $\hat{T}_1$. The corresponding singularity signatures are as follows (see Figure \ref{fig:okamotoratsbasecaseregion3}):
\begin{figure}[ht]
 \begin{subfigure}{0.45\textwidth}
    \centering
    \includegraphics[width=\textwidth]{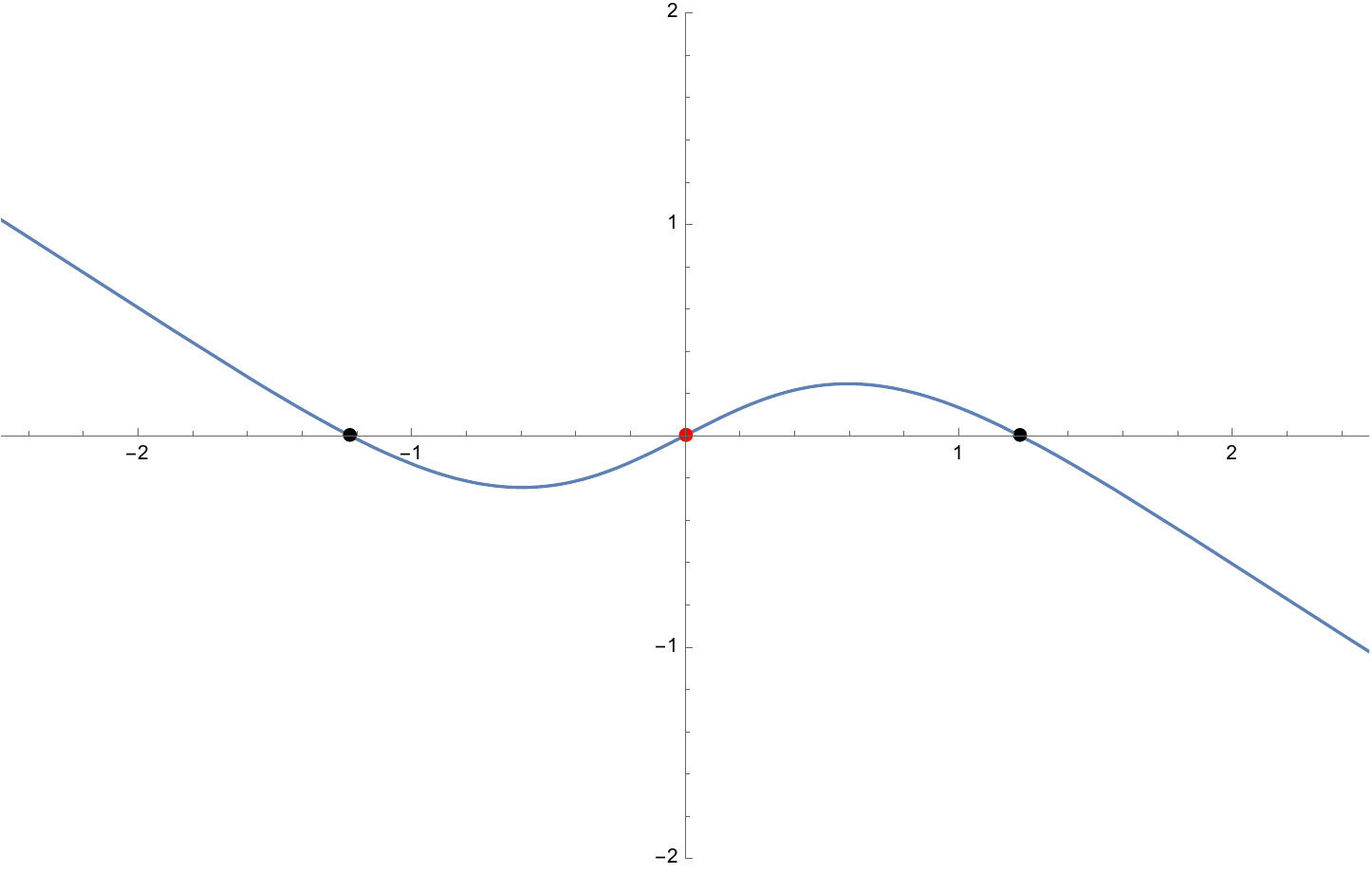}
\caption{$q_{-1,0}(t)$}
\label{subfig:okamotorat-1,0}
 \end{subfigure}
\qquad
 \begin{subfigure}{0.45\textwidth}
    \centering
    \includegraphics[width=\textwidth]{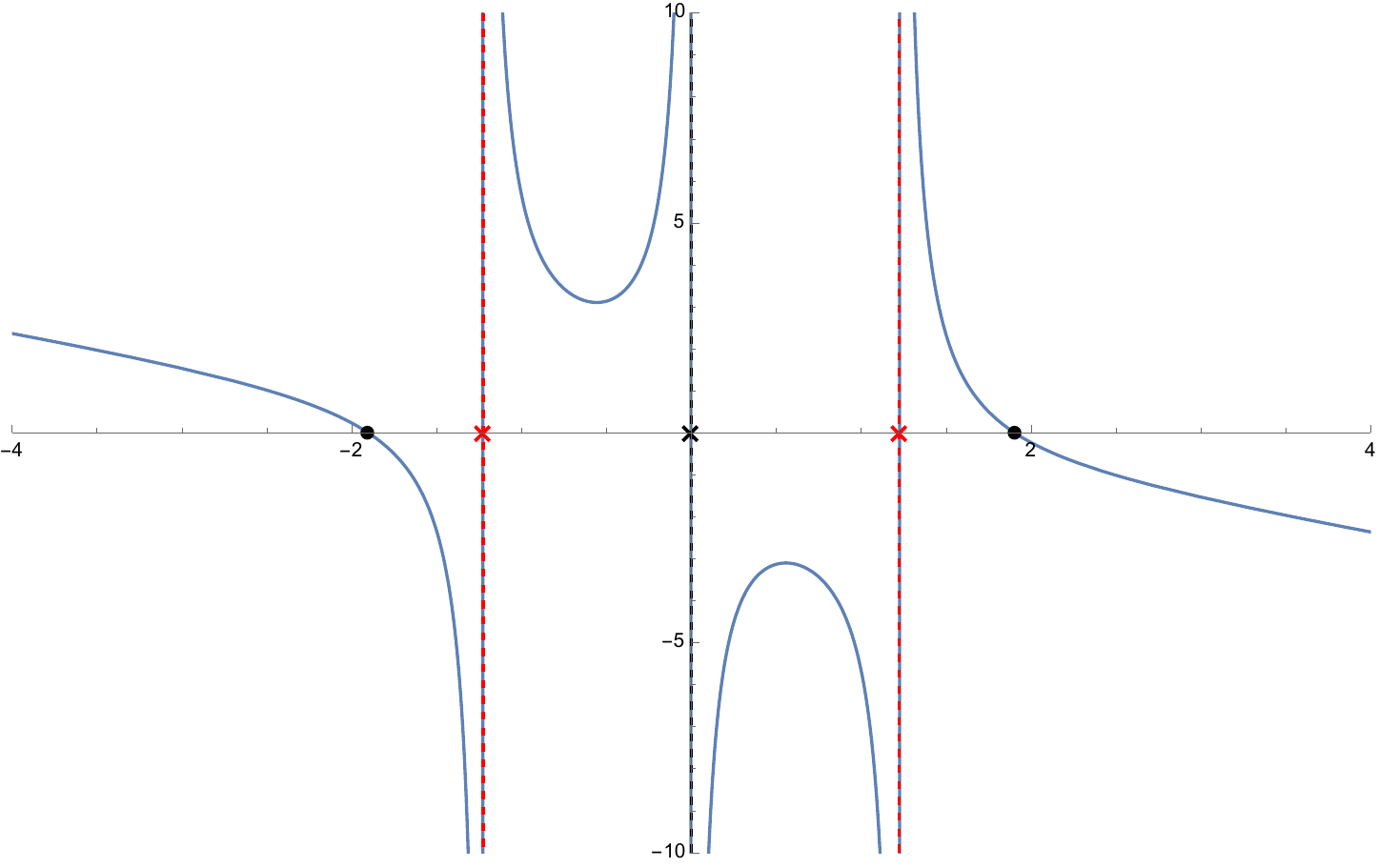}
\caption{$q_{-1,1}(t)$}
\label{subfig:okamotorat-1,1}
 \end{subfigure}
    \caption{Base cases for induction in region $\mathrm{III}$}
    \label{fig:okamotoratsbasecaseregion3}
\end{figure}



\subsection{Translation $T_2^{-1}$ in region $\mathrm{III}$}
Recall from Section \ref{sec:tools} that $T_2^{-1}$ acts on apparent singularities as follows:
\begin{equation} \label{singsmovementT2inv}
p_+ \xrightarrow{~T_2^{-1}~} p_-, \qquad p_- \xrightarrow{~T_2^{-1}~} C_{T_2^{-1}(p_-)}, \qquad z_+ \xrightarrow{~T_2^{-1}~} C_{T_2^{-1}(z_+)}, \qquad z_- \xrightarrow{~T_2^{-1}~} C_{T_2^{-1}(z_-)}.
\end{equation}

\subsubsection{Curve crossings in the real $(f,g)$-plane relevant to $T_2^{-1}$ in region $\mathrm{III}$}

\begin{lemma} \label{lem:regionsT2inv}
For parameters $a$ in region $\mathrm{III}$, the curves
\begin{equation*}
\begin{aligned}
    C_{T_2^{-1}(p_-)} &:  f g + 2a_2= 0, 
    \\
    C_{T_2^{-1}(z_+)} &:  g = 0, \\
    C_{T_2^{-1}(z_-)} &:  f g - 2 a_1 = 0, 
    \\
    C_{z_-} &: f=0,
\end{aligned}
\end{equation*}
 divide the real $(f,g)$-plane into the eight regions $L_1,\dots, L_4$, $R_1,\dots,R_4$ shown in Figure \ref{fig:masterpictureT2inv}.
\end{lemma}


\begin{lemma} \label{lem:crossingsT2inv}
When the parameters $a$ lie in region $\mathrm{III}$ with $a_1+a_2<0$, apparent singularities of real solutions at $t=t_*<0$ correspond to the following behaviours in the real $(f,g)$-plane as shown in Figure \ref{fig:masterpictureT2inv}:
\begin{itemize}
    \item $z_+$ is either
    \begin{itemize}
        \item a crossing from $L_2$ to $R_2$,
        \item a crossing from $L_2$ to $R_3$, with the curve $C_{T_2^{-1}(z_-)}$ on $X$ crossed simultaneously,
        \item a crossing from $L_3$ to $R_3$.
    \end{itemize}
    \item $z_-$ is either
        \begin{itemize}
            \item a crossing from $R_3$ to $L_1$,
            \item a crossing from $R_3$ to $L_2$, with the curve $C_{T_2^{-1}(z_+)}$ on $X$ crossed simultaneously,
            \item a crossing from $R_4$ to $L_2$.
        \end{itemize}
    \item $p_+$ is either
        \begin{itemize}
            \item a crossing from $L_3$ to $R_1$,
            \item a crossing from $L_4$ to $R_1$, with the curve $C_{T_2^{-1}(p_-)}$ on $X$ crossed simultaneously,
            \item a crossing from $L_4$ to $R_2$.
        \end{itemize}
    \item $p_-$ is a crossing from $R_1$ to $L_4$.
\end{itemize}
\end{lemma}

\begin{figure}[htb]
    \centering
	\begin{tikzpicture}
    \node (pic) at (0,0) {\includegraphics[width=.6\textwidth]{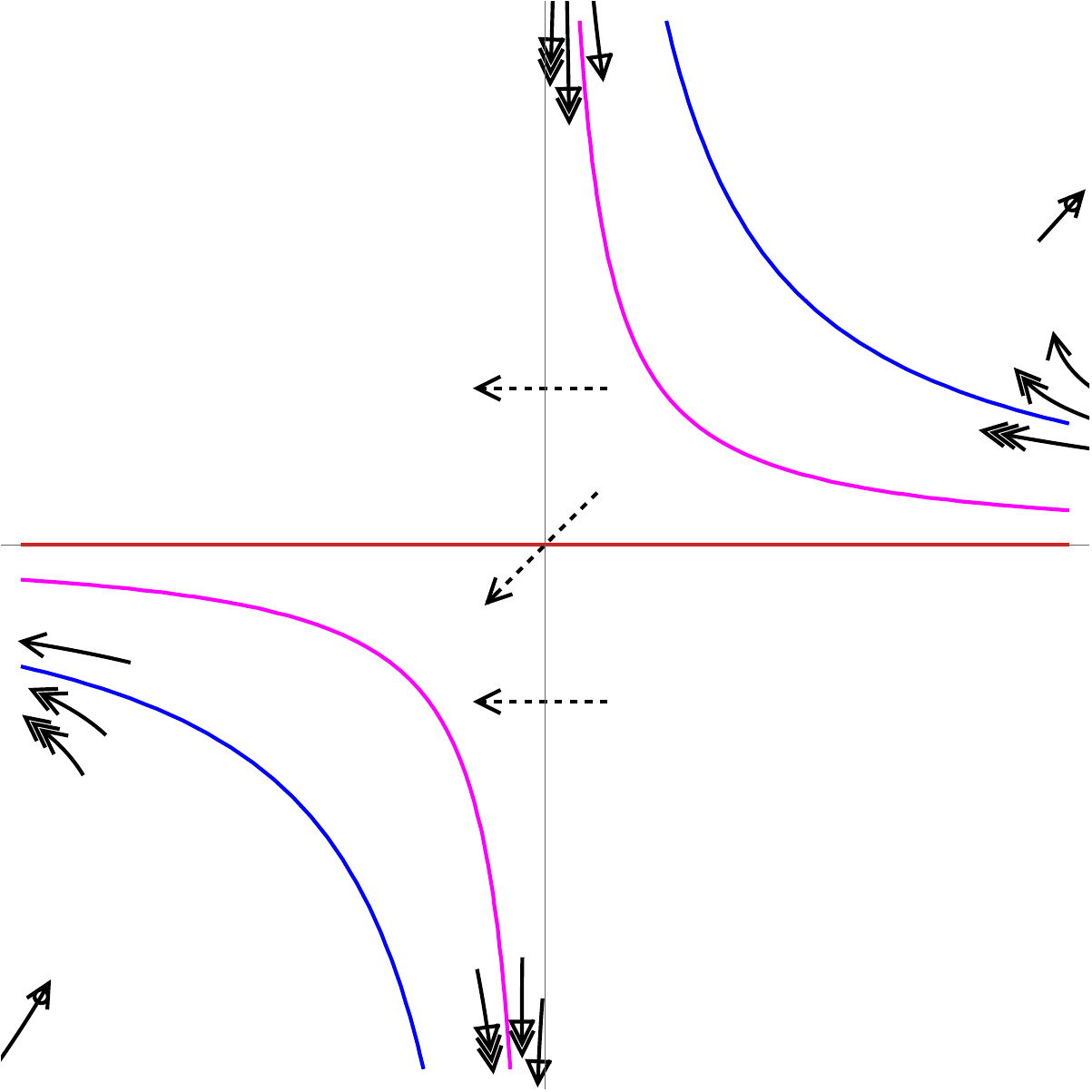}};
       \node (L1) at (-2,2) {$L_1$};
    \node (L2) at (-.7,-.7) {$L_2$};
    \node (L3) at (-1.6,-1.6) {$L_3$};
    \node (L4) at (-3,-3) {$L_4$};    
    \node (R1) at (3,3) {$R_1$};
    \node (R2) at (1.6,1.6) {$R_2$};
    \node (R3) at (.7,.7) {$R_3$};
    \node (R4) at (2,-2) {$R_4$};    
    \node (bluecurvelabelL) at (-1.8,-4) {{\color{blue}$C_{T_2^{-1}(p_-)}$}};    
    \node (bluecurvelabelR) at (1.8,4) {{\color{blue}$C_{T_2^{-1}(p_-)}$}};    
    \node (magentacurvelabelL) at (-2.6,-.8){{\color{magenta}$C_{T_2^{-1}(z_-)}$}};    
    \node (magentacurvelabelR) at (2.8,.8) {{\color{magenta}$C_{T_2^{-1}(z_-)}$}};    
    \node (redcurvelabelL) at (-3.8,+.4) {{\color{red}$C_{T_2^{-1}(z_+)}$}};    
    \node (redcurvelabelR) at (3.8,-.4) {{\color{red}$C_{T_2^{-1}(z_+)}$}};    
\begin{scope}[xshift=.2cm]       
   \node (legend) at (6.6,-0.05) {\includegraphics[width=.09\textwidth]{arrowslegendT2better.png}};
    \node (pminus1) at (7.2,2.55) [right] {$p_+$};
    \node (pminus2) at (7.2,1.8) [right] {$p_+ \cap C_{T_2^{-1}(p_-)}$};
    \node (pminus3) at (7.2,1.05) [right] {$p_+$};
    \node (zplus) at (7.2,0.3) [right] {$z_+$};
    \node (zplus2) at (7.2,-0.4) [right] {$z_+ \cap C_{T_2^{-1}(z_-)}$};
    \node (zplus3) at (7.2,-1.15) [right] {$z_+$};
    \node (pplus) at (7.2,-1.95) [right] {$p_-$};
    \node (zminus) at (7.2,-2.65) [right] {$z_-$};
\end{scope}
	\end{tikzpicture}
    \caption{Apparent singularities relative to curves for $T_2^{-1}$, for $(m,n)$ in region $\mathrm{III}$ and $t<0$.}
    \label{fig:masterpictureT2inv}
\end{figure}

\subsubsection{Boundary conditions at $t=0$ and $t\rightarrow -\infty$}

\begin{lemma} \label{lem:asymptoticregionT2inv}
For sufficiently large negative $t$, the Okamoto rational solution $(f_{m,n}(t),g_{m,n}(t))$ for $(m,n)$ in region $\mathrm{III}$ lies in region $R_4$.
\end{lemma}

\begin{lemma}
For parameters $a$ in region $\mathrm{III}$, if a solution has a special apparent singularity at the origin $t=0$ of type $p_{\pm}'$ or $z_{\pm}'$, then for sufficiently small negative $t$ the solution $(f(t),g(t))$ will lie in the following regions:
\begin{equation*}
z_-'  : R_3 \qquad z_+' : L_2 \qquad p_-' : R_1 \qquad p_+' : L_4.
\end{equation*}
\end{lemma}

\subsection{The inductive arguments for $T_2^{-1}$ in region $\mathrm{III}$}

\begin{lemma} \label{lem:asymptoticT2inv}
A new sequence $(z_-\,p_+)$ must occur before the first new $p_-$ (first old $p_+$).
This corresponds to the following movement of the solution curve through the regions indicated above:
\begin{equation*}
    \begin{tikzcd}
        R_4 \arrow[r,"z_-"] & L_2 \arrow[r,"C_{T_2^{-1}(z_-)}"] & L_3 \arrow[r,"p_+"] & R_1 
    \end{tikzcd}
\end{equation*}
\end{lemma}

\begin{lemma} \label{lem:T2inv:4cycles1}
For parameters in region $\mathrm{III}$,
\begin{equation*}
    (p_+\, p_-\, z_+\, z_-\, p_+) \xrightarrow{T_2^{-1}} (p_-\, z_+\, z_-\, p_+ \,p_-),
\end{equation*}
using the notation introduced in \eqref{eq:5.10}. 
The new solution passes through regions as follows:
\begin{equation*}
    \begin{tikzcd}
    & & & & R_4 \arrow[dr, "z_-"] & & & & \\
~ \arrow[r, "p_-"]  & L_4 \arrow[r,"C_{T_2^{-1}(p_-)}"]  & L_3 \arrow[r,"z_+"] & R_3 \arrow[rr,"z_-\cap C_{T_2^{-1}(z_+)}"] \arrow[dr,swap, "z_-"] \arrow[ur, "C_{T_2{-1}(z_+)}"] &  & L_2 \arrow[r,"C_{T_2^{-1}(z_-)}"] & L_3 \arrow[r,"p_+"] & R_1 \arrow[r,"p_-"] &~ \\
    & & & & L_1 \arrow[ur, swap, "C_{T_2^{-1}(z_+)}"]& & & & 
    \end{tikzcd}
\end{equation*}
\end{lemma}

\begin{lemma} \label{lem:T2inv:2cycles1}
For parameters in region $\mathrm{III}$,
\begin{equation*}
    (p_+\, p_-\, p_+) \xrightarrow{T_2^{-1}} (p_-\,p_+\,p_-),
\end{equation*}
using the notation introduced in \eqref{eq:5.10}. 
The new solution passes through regions as follows:
\begin{equation*} 
    \begin{tikzcd}
                    &       & R_2 \arrow[dr,"C_{T_2^{-1}(p_-)}"]  &         ~   \\
    ~ \arrow[r,"p_-"] & L_4 \arrow[ur,"p_+"]\arrow[rr,"p_+\cap C_{T_2^{-1}(p_-)}"]\arrow[dr,swap,"C_{T_2^{-1}(p_-)}"] &       & R_1  \arrow[r,"p_-"]   &~  \\
                    &       & L_3 \arrow[ur,swap, "p_+"]  &          ~
    \end{tikzcd}
\end{equation*}
\end{lemma}

\begin{lemma}  \label{lem:T2inv:zonelemma1}
For parameters in region $\mathrm{III}$, and positive integer $\ell$,
\begin{equation*}
    p_+\,p_- (z_+\,z_-)^{\ell} \hat{z}_+
    \xrightarrow{T_2^{-1}}
    p_-\,z_+\,(z_-\,z_+)^{\ell} \hat{z}_-,
\end{equation*}
using the notation introduced in \eqref{eq:5.10}. 
The new solution passes through regions as follows:

\begin{equation*} 
    \begin{tikzcd}
~\arrow[r,"p_-"] & L_4 \arrow[r,"C_{T_2^{-1}(p_-)}"] & L_3 \arrow[r,"z_+"] &\big[R_3 \arrow[r,"\mathcal{M}"] & R_3\big]^{\ell}  \arrow[r,"\hat{z}_-"] & ~
    \end{tikzcd}
\end{equation*}
where we have used $R_3 \xrightarrow{\mathcal{M}} R_3$ to denote the following movement between regions, which includes a new $z_-$ followed by a new $z_+$: 
\begin{equation*} 
   \begin{tikzcd}[row sep=1cm,column sep=1cm]
  & & L_1 \arrow[dr,"C_{T_1^{-1}(z_+)}"]& & R_2 \arrow[dr,"C_{T_2^{-1}(z_-)}"] & ~ \\
\mathcal{M} : &  R_3 \arrow[rr,"z_-\cap C_{T_2^{-1}(z_+)}"]  \arrow[ur,"z_-"] \arrow[dr,swap,"C_{T_2^{-1}(z_+)}"] & & L_2 \arrow[rr,"z_+\cap C_{T_2^{-1}(z_-)}"]  \arrow[ur,"z_+"] \arrow[dr,swap,"C_{T_2^{-1}(z_-)}"] & & R_3 \\
  &  & R_4 \arrow[ur,swap,"z_-"]& & L_3 \arrow[ur,swap,"z_+"] & ~
    \end{tikzcd}
    \qquad \quad
\end{equation*}
\end{lemma}

\begin{lemma}  \label{lem:T2inv:zonelemma2}
For parameters in region $\mathrm{III}$, and positive integer $\ell$,
\begin{equation*}
    p_+\,p_-\, z_+ (z_-\,z_+)^{\ell} \hat{z}_-
    \xrightarrow{T_2^{-1}}
    p_-\,(z_+\,z_-)^{\ell} \hat{z}_+,
\end{equation*}
using the notation introduced in \eqref{eq:5.10}. 
The new solution passes through regions as follows:

\begin{equation*} 
    \begin{tikzcd}
    & & & & &    L_1 \arrow[dr,"C_{T_1^{-1}(z_+)}"]&  &  \\
~\arrow[r,"p_-"] & L_4 \arrow[r,"C_{T_2^{-1}(p_-)}"] & L_3 \arrow[r,"z_+"] &\big[ R_3 \arrow[r,"\mathcal{M}"] & R_3 \big]^{\ell-1} \arrow[rr,"z_-\cap C_{T_2^{-1}(z_+)}"]  \arrow[ur,"z_-"] \arrow[dr,swap,"C_{T_2^{-1}(z_+)}"] & & L_2  \arrow[r,"\hat{z}_-"] & ~ \\
   & & & & & R_4 \arrow[ur,swap,"z_-"]& &  
    \end{tikzcd}
\end{equation*}
\end{lemma}

With Lemmas \ref{lem:asymptoticT2inv}, \ref{lem:T2inv:4cycles1}, \ref{lem:T2inv:2cycles1}, \ref{lem:T2inv:zonelemma1} and \ref{lem:T2inv:zonelemma2} in hand, we can prove the four inductive steps for $T_2^{-1}$ in region $\mathrm{III}$.

\subsection{Translation $\hat{T}_1$ in region $\mathrm{III}$}
For translation $\hat{T}_1$ in region $\mathrm{III}$, the topological picture of exceptional curves corresponding to apparent singularities, and their iterations, is the same as in region $\mathrm{II}$.
By direct calculation along the same lines as above, it can be verified that for parameters in region $\mathrm{III}$ with $a_1>1$ the division of the real $(f,g)$-plane into regions by curves relevant to $\hat{T}_1$ and the placement of apparent singularities relative to the curves is the same as for region $\mathrm{II}$, so we can argue using Figure \ref{fig:masterpictureT1hat}.
Further, the boundary conditions at $t=0$ and $t\rightarrow -\infty$ are the same as in Lemmas \ref{lem:asymptoticregionT1hat} and \ref{lem:originT1hat}.



\subsection{The inductive arguments for $\hat{T}_1$ in region $\mathrm{III}$}

\begin{lemma} \label{lem:asymptoticT1hatregion3}
A new sequence $(z_-\,p_+)$ must occur before the first new $p_-$ (first old $z_-$).
This corresponds to the following movement of the solution curve through the regions on Figure \ref{fig:masterpictureT1hat}:
\begin{equation*}
    \begin{tikzcd}
        R_4 \arrow[r,"z_-"] &  L_3 \arrow[r,"p_+"] & R_1 
    \end{tikzcd}
\end{equation*}
\end{lemma}

\begin{lemma} \label{lem:T1hatregion3:4cycles1}
For parameters in region $\mathrm{III}$,
\begin{equation*}
    (z_-\, p_+\, p_-\, z_+\, z_-) \xrightarrow{\hat{T}_1} (p_-\, z_+\, z_-\, p_+ \,p_-),
\end{equation*}
using the notation introduced in \eqref{eq:5.10}. 
The new solution passes through regions as follows:
\begin{equation} \label{regionsCD:T1hatregion3:4cycles1}
    \begin{tikzcd}
    & & & L_1  \arrow[dr,"C_{\hat{T}_1(p_-)}"] & & R_2 \arrow[dr,"C_{\hat{T}_1(z_+)}"] & & \\
~\arrow[r,"p_-"] & L_2 \cup L_4 \arrow[r,"z_+"] & R_1 \arrow[ur,"z_-"] \arrow[dr,swap, "C_{\hat{T}_1(p_-)}"] \arrow[rr, "z_- \cap C_{\hat{T}_1(p_-)}"] & & L_2 \cup L_4 \arrow[dr,swap,"C_{\hat{T}_1(z_+)}"] \arrow[ur,"p_+"] \arrow[rr,"p_+ \cap C_{\hat{T}_1(z_+)}"] & & R_1 \arrow[r,"p_-"] &~ \\
    & & & R_3 \arrow[ur,swap,"z_-"] & & L_3 \arrow[ur,swap,"p_+"] & & \\
    \end{tikzcd}
\end{equation}
\end{lemma}

\begin{lemma} \label{lem:T1hatregion3:2cycles1}
For parameters in region $\mathrm{III}$,
\begin{equation*}
    (p_+\, p_-\, p_+) \xrightarrow{\hat{T}_1} (z_+\,z_-\,z_+),
\end{equation*}
using the notation introduced in \eqref{eq:5.10}. 
The new solution passes through regions as follows:
\begin{equation*}
    \begin{tikzcd}
                    &       & L_1 \arrow[dr,"C_{\hat{T}_1(p_-)}"]  &         ~   \\
    ~ \arrow[r,"z_+"] & R_1 \arrow[ur,"z_-"]\arrow[rr,"z_-\cap C_{\hat{T}_1(p_-)}"]\arrow[dr,swap,"C_{\hat{T}_1(p_-)}"] &       & L_2 \cup L_4 \arrow[r,"z_+"]   &~  \\
                    &       & R_3 \arrow[ur,swap, "z_-"]  &          ~
    \end{tikzcd}
\end{equation*}
\end{lemma}

\begin{lemma} \label{lem:T1hatregion3:2cycles2}
For parameters in region $\mathrm{III}$,
\begin{equation*}
    (z_-\, z_+\, z_-) \xrightarrow{\hat{T}_1} (p_-\,p_+\,p_-),
\end{equation*}
using the notation introduced in \eqref{eq:5.10}. 
The new solution passes through regions as follows:
\begin{equation*}
    \begin{tikzcd}
                    &       & R_2 \arrow[dr,"C_{\hat{T}_1(z_+)}"]  &         ~   \\
    ~ \arrow[r,"p_-"] & L_2 \cup L_4 \arrow[ur,"p_+"]\arrow[rr,"p_+\cap C_{\hat{T}_1(z_+)}"]\arrow[dr,swap,"C_{\hat{T}_1(z_+)}"] &       & R_1  \arrow[r,"p_-"]   &~  \\
                    &       & L_3 \arrow[ur,swap, "p_+"]  &          ~
    \end{tikzcd}
\end{equation*}
\end{lemma}
With Lemmas \ref{lem:asymptoticT1hatregion3}, \ref{lem:T1hatregion3:4cycles1}, \ref{lem:T1hatregion3:2cycles1} and \ref{lem:T1hatregion3:2cycles2} in hand, we can prove the four inductive steps for $\hat{T}_1$ in region $\mathrm{III}$ and the proof of Theorem \ref{th:region3} is complete.

\section{Proofs for regions $\mathrm{IV}$, $\mathrm{V}$ and $\mathrm{VI}$}  \label{sec:regions456}

While it is in principle possible to prove the Theorems \ref{th:region4},\ref{th:region5} and \ref{th:region6} for the remaining regions using similar strategies to those employed in the preceding sections, we will make use of the symmetry of $\pain{IV}$.

For $(m,n)\in \Z^2$, by the symmetry 
\begin{equation} \label{symmetry2}
a_0 \mapsto \hat{a}_0 = a_0 + a_1, \quad a_1 \mapsto \hat{a}_1 = - a_1, \quad a_2 \mapsto \hat{a}_2 = a_1 + a_2,    
\end{equation}
the generalised Okamoto rational $q_{m,n}(t)$ is also a solution of $\pain{IV}$ with parameters 
\begin{equation} \label{eq:hatparams}
    \hat{a}_0 = \tfrac{2}{3} - m, \quad \hat{a}_1 = -\tfrac{1}{3} - n, \quad  \hat{a}_2 = \tfrac{2}{3} + m + n.
\end{equation}
Let $\hat{m} = m + n $, $\hat{n}=-n$ and denote $q_{m+n,-n}(t)$, when regarded as a solution for parameters \eqref{eq:hatparams}, by $\hat{q}_{\hat{m},\hat{n}}(t)$.

Note that the singularity signature of $\hat{q}_{\hat{m},\hat{n}}$ as a solution of $\pain{IV}$ for parameters $\hat{a}$ is related to that of $q_{m,n}$ as a solution for parameters $a$ by 
\begin{equation} \label{eq:signatureshatrelation}
    \mathfrak{S}(q_{m,n}) = \mathfrak{S}(\hat{q}_{\hat{m},\hat{n}})|_{z_+\leftrightarrow z_-},\quad \hat{m}=m+n,\quad \hat{n}=-n,
\end{equation}
where the right-hand side has $z_-$ and $z_+$ swapped because $\hat{a}_1 = -a_1$, so the roles are interchanged between plus zeroes and minus zeroes defined according to their Laurent series expansions in Lemma \ref{lem:singularities_expansions}.
\subsection{Proof for region $\mathrm{VI}$}\label{subsec:region6}
Parameters in regions $\mathrm{I}$ and $\mathrm{VI}$ are related by the symmetry \eqref{symmetry2}, so we are able to use parts of the proof of Theorem \ref{th:region1} to prove Theorem \ref{th:region6} as follows.
The idea is to find a seed solution $q_{m_0,n_0}$, for $(m_0,n_0)$ in region $\mathrm{VI}$, whose associated $\hat{q}_{\hat{m}_0,\hat{n}_0}$ has singularity signature which coincides with one in Theorem \ref{th:region1}.
Then Theorem \ref{th:region1} allows us to deduce the singularity signature of the result of applying $T_1^{s}T_2^{r}$ to $\hat{q}_{\hat{m}_0,\hat{n}_0}$ for nonnegative integers $r,s$, which we translate back via \eqref{eq:signatureshatrelation} to obtain the desired formula for $\mathfrak{S}(q_{m,n})$ with $(m,n)$ in region $\mathrm{VI}$.

Consider the solution $q_{2,-1}$, whose singularity signature can be computed to be $\mathfrak{S}(q_{2,-1})=p_- \,z_-\,z_+\,\hat{p}_+\,z_+\,z_-\,p_-$, so the corresponding solution $\hat{q}_{1,1}$ for $\hat{a}$ in region $\mathrm{I}$ has singularity signature
\begin{equation*}
    \mathfrak{S}(\hat{q}_{1,1})= \mathfrak{S}(q_{2,-1})|_{z_+\leftrightarrow z_-} = p_- \,z_+\,z_-\,\hat{p}_+\,z_-\,z_+\,p_-.
\end{equation*}
We note that this does not coincide with any of the formulas in Theorem \ref{th:region1}, and the same is true of the next natural candidate for seed solution $q_{3,-1}=\hat{q}_{2,1}$.
Therefore we proceed to consider $q_{3,-2}=\hat{q}_{1,2}$, for which we find does coincide with one of the signatures established in Theorem \ref{th:region1} for region $\mathrm{I}$: 
\begin{equation*}
    \mathfrak{S}(\hat{q}_{1,2})= \mathfrak{S}(q_{3,-2})|_{z_+\leftrightarrow z_-} = p_- \,z_+\,z_-\,p_+\,\hat{z}_-\,p_+\,z_-\,z_+\,p_- = \mathfrak{S}(q_{2,0}).
\end{equation*}
Now given that $\hat{q}_{1,2}$ is a solution for parameters $\hat{a}$ in region $\mathrm{I}$ such that the arguments involved in the proof of Theorem \ref{th:region1} hold, we can deduce that the singularity signature of the solution $\hat{q}_{1+r,2+s}$ for $r,s\geq 0$ obtained via the B\"acklund transformation $T_1^s T_2^r$ for $r,s\geq 0$ is given by the formulas in Theorem \ref{th:region1} with $m=r+2$, $n=s$, i.e. $\mathfrak{S}(\hat{q}_{1+r,2+s})= \mathfrak{S}(q_{2+r,s})$, so using the relation \eqref{eq:signatureshatrelation} we have that 
\begin{equation} \label{eq:relationregion1region6}
    \mathfrak{S}(q_{m,n})=\mathfrak{S}(q_{\hat{m}+1,\hat{n}-2})|_{z_+\leftrightarrow z_-},\quad \hat{m}=m+n,\quad \hat{n}=-n,
\end{equation}
for $\hat{m} \geq 1$ and $\hat{n}\geq 2$, i.e. all cases $(m,n)$ in region $\mathrm{VI}$ except for those with $n =-1$.
With the relation \eqref{eq:relationregion1region6}, the formulas in Theorem \ref{th:region6} can be obtained from those in Theorem \ref{th:region1}, so it remains only to deal with the case of $(m,n)$ in region $\mathrm{VI}$ with $n=-1$.
We will deduce these by proving formulas for the corresponding $\hat{q}_{\hat{m},\hat{n}}$ with $\hat{m},\hat{n}$ in region $\mathrm{I}$. 
The first of these, corresponding to $m=2\mu \geq 2$, is 
\begin{equation*}
    \mathfrak{S}(\hat{q}_{2\mu-1,1}) = (p_-\,z_+\,z_-\,p_+)^{\mu-1}\,p_-\,z_+\,z_-\,\hat{p}_+\,z_-\,z_+\,p_-\,(p_+\,z_-\,z_+\,p_-)^{\mu-1},
\end{equation*}
and the second, corresponding to $m=2\mu+1$, is 
\begin{equation*}
    \mathfrak{S}(\hat{q}_{2\mu,1}) = (p_-\,z_+\,z_-\,p_+)^{\mu}\,\hat{p}_-\,(p_+\,z_-\,z_+\,p_-)^{\mu}.
\end{equation*}
The two inductive steps here can be dealt with using Lemmas \ref{lem:asymptoticregionT2} and \ref{lem:T2region1:4cycles1} for $T_2$ in region $\mathrm{I}$, after which the relation \eqref{eq:signatureshatrelation} gives the required formulas for region $\mathrm{VI}$ with $n=-1$ and the proof of Theorem \ref{th:region6} is complete.

\subsection{Proof for region $\mathrm{V}$}
\label{subsec:region5}
Parameters in regions $\mathrm{II}$ and $\mathrm{V}$ are related by the symmetry \eqref{symmetry2}.
Consider the solution $q_{0,-2}$, the singularity signature of which can be computed to be 
\begin{equation*}
    \mathfrak{S}(q_{0,-2}) = z_+\,p_+\,p_-\,\hat{z}_-\,p_-\,p_+\,z_+.
\end{equation*}
The signature of the corresponding hatted solution then coincides with that of one of the solutions to which the formulas in region $\mathrm{II}$ apply:
\begin{equation*}
    \mathfrak{S}(\hat{q}_{-2,2}) =  \mathfrak{S}(q_{-1,2}).
\end{equation*}
Then applying the arguments for $T_1^{r} \hat{T}_1^s  = T_1^{r+s} T_2^{-s}$ as in the proof of Theorem \ref{th:region2} we have $\mathfrak{S}(\hat{q}_{-2-s,2+r+s})=\mathfrak{S}(q_{-1-s,2+r+s})|_{z_+\leftrightarrow z_-}$, which for $r=m \geq 0$, $s=-n-m-2 \geq 0 $ leads to 
\begin{equation} \label{eq:relationregion2region5}
    \mathfrak{S}(q_{m,n}) =  \mathfrak{S}(\hat{q}_{m+n,-n})=\mathfrak{S}(q_{m+n+1,-n})|_{z_+\leftrightarrow z_-},
\end{equation}
for any $(m,n)$ in region $\mathrm{V}$ such that $m+n\leq -2$.

It remains to prove that formula \eqref{eq:relationregion2region5} holds for the cases
 with $m+n=0,-1$, which we do as follows.
To cover the $(m,n)$ in region $\mathrm{V}$ along the subdiagonal $m+n=-1$ in Figure \ref{fig:regions}, consider the seed solution $q_{0,-1}$, the signature of which can be computed to be
\begin{equation*}
    \mathfrak{S}(q_{0,-1})|_{z_+\leftrightarrow z_-} = \mathfrak{S}(\hat{q}_{-1,1}) = z_-\, \hat{p}_+ \,z_- = \mathfrak{S}(q_{0,1}).
\end{equation*}
Note that $q_{0,1}$ is a solution for $(m,n)$ in region $\mathrm{I}$, so we can apply the arguments for $T_1^r$ with $r = n-1\geq 0$ in region $\mathrm{I}$ to obtain 
\begin{equation*}
\mathfrak{S}(q_{n-1,-n})|_{z_+\leftrightarrow z_-}  = \mathfrak{S}(\hat{q}_{-1,n})  = \mathfrak{S}(q_{0,n}),
\end{equation*}
for all $n\geq 0$. 
Taking $m=0$ in the formulas for $\mathfrak{S}(q_{m,n})$ with $(m,n)$ in region $\mathrm{I}$ we see that the relation \eqref{eq:relationregion2region5} also holds on the subdiagonal $m+n=-1$ in region $\mathrm{V}$ as claimed.

Finally, the diagonal $m+n=0$ is dealt with similarly using the seed solution $q_{2,-2}$, whose associated $\hat{q}_{0,2}$ has signature coinciding with with one in region $\mathrm{I}$:
\begin{equation*}
    \mathfrak{S}(q_{2,-2})|_{z_+\leftrightarrow z_-} = p_-\, \hat{z}_+ \,p_- = \mathfrak{S}(q_{1,0}).
\end{equation*}
Again applying the arguments for $T_1^r$ in region $\mathrm{I}$ the relation \eqref{eq:relationregion2region5} follows for $(m,n)$ in region $\mathrm{V}$ with $m+n=0$. 
Then the formulas in Theorem \ref{th:region2} can be translated via the relation \eqref{eq:relationregion2region5} to obtain those in Theorem \ref{th:region5} and the proof is complete.

\subsection{Proof for region $\mathrm{IV}$}
\label{subsec:region4}
Parameters in regions $\mathrm{III}$ and $\mathrm{IV}$ are related by the symmetry \eqref{symmetry2}.
Consider the solution $q_{-1,-1}$, from which all generalised Okamoto rationals can be obtained by applying $T_1^{-1}$ and $T_2^{-1}$.
The singularity signature of this solution can be computed to be 
\begin{equation*}
    \mathfrak{S}(q_{-1,-1}) = z_+\,p_+\,\hat{p}_-\,p_+\,z_+.
\end{equation*}
The singularity signature of its corresponding $\hat{q}_{\hat{m},\hat{n}}$ coincides with that of one of the solutions to which the formulas in region $\mathrm{III}$ apply:
\begin{equation*}
 \mathfrak{S}(q_{-1,-1})|_{z_+\leftrightarrow z_-} =   \mathfrak{S}(\hat{q}_{-2,1}) =   \mathfrak{S}(q_{-1,1}).
\end{equation*}
Then applying the arguments for $\hat{T}_1^s T_2^{-r} = T_1^{s} T_2^{-r-s}$ with $r=-m-1$, $s=-n-1$ as in the proof of Theorem \ref{th:region3} we have, for any $(m,n)$ in region $\mathrm{IV}$,
\begin{equation*}
    \mathfrak{S}(q_{m,n})=\mathfrak{S}(q_{m+n+1,-n})|_{z_+\leftrightarrow z_-},
\end{equation*}
from which the formulas in Theorem \ref{th:region4} can be deduced and the proof is complete.

\section{Conclusions} \label{sec:conclusions}

The main goal of this paper was to determine the number of real and imaginary roots of the generalised Okamoto polynomials.
In order to do this, we developed an inductive procedure to derive the qualitative distribution of singularities on the real line of real solutions of Painlev\'e equations related by translations, based on the known distribution of a seed solution. 
For the rational solutions of $\pain{IV}$ expressed in terms of generalised Okamoto polynomials, this procedure led not only to exact formulas for the numbers of real roots of the polynomials, but also to proofs of various interlacing properties.

We expect that our approach can be adapted to other hierarchies of real solutions of Painlev\'e equations with a seed solution for which one can explicitly compute the singularity signature. Natural candidates for our approach are thus the hierarchies of rational, algebraic and special function solutions.


Considering for example the rational solutions of $\pain{II}$, expressed in terms of Yablonskii-Vorob'ev polynomials \cite{yablonskii, vorobev}, the numbers of real roots of these polynomials $Y_n$ are known \cite{pieterYV}, as well as interlacing of roots of $Y_{n-1}$ and $Y_{n+1}$ on the real line \cite{clarksonsurvey}, but interlacing of roots of consecutive polynomials $Y_n$ and $Y_{n+1}$ has not been proved.
We remark that for the other rational special solution hierarchy of $\pain{IV}$ which is expressed in terms of generalised Hermite polynomials
\cite{noumiyamada}, the number of real roots of the generalised Hermite polynomials is known \cite{davidepieter}.

There are also special solution hierarchies of $\pain{III}$ (in the generic case of surface type $D_6^{(1)}$ in the Sakai scheme \cite{SAKAI2001}) and $\pain{V}$  expressed in terms of Umemura polynomials \cite{umemurapolynomials,noou}, see also \cite{kajiwaramasuda, noumiyamadapolynomials}. 
Furthermore, $\pain{VI}$ admits two-parameter families of rational solutions expressed in terms of Wronskians of Jacobi polynomials \cite{okamotopvi,forresterwittepvi}.


For the degenerate case of $\pain{III}$ of surface type $D_7^{(1)}$ there are solutions which are rational functions of $t^{1/3}$ expressed in terms of Ohyama polynomials \cite{studies5,clarksonP3}, and do not contain parameters.

Perhaps the most challenging and interesting case to adapt this method to, would be the algebraic solutions of $\pain{VI}$ in the $45$ exceptional orbits under the action of the affine Weyl group of type $F_4^{(1)}$ as classified in \cite{lisovyyalgebraic}. 

Finally, in light of the appearance of generalised Okamoto polynomials in relation to the rational solutions constructed by Yang and Yang, it is natural to ask whether the fourth Painlev\'e equation itself governs classes of solutions through a (perhaps asymptotic) similarity reduction of the Sasa-Satsuma equation. 
We may address some of the above questions in later work.

\appendix

\section{Rational solutions of the Sasa-Satsuma equation} \label{app:rationalsolutions}

In order to define the rational solutions to the Sasa-Satsuma equation derived in  \cite{yangyang}, we require the elementary Schur polynomials $S_j(\boldsymbol{\mathrm{x}})$, $j\geq 0$, for an infinite vector $\boldsymbol{\mathrm{x}}=(x_1,x_2,x_3,\ldots)$, defined by the generating function
\begin{equation*}
\sum_{j=0}^\infty S_j(\boldsymbol{\mathrm{x}})\epsilon^j=\exp{\left(\sum_{j=1}^\infty{x_j\epsilon^j}\right)}.
\end{equation*}
We further fix a sign $\pm$ and let $p=p(\kappa)$ be the unique real analytic solution around $\kappa=0$ to
\begin{equation*}
    p+\frac{8p}{4p^2+1}=\pm\tfrac{1}{2}\sqrt{3}\left[e^\kappa+2e^{-\kappa/2}\cos{(\tfrac{1}{2}\sqrt{3}\kappa)}\right],
\end{equation*}
so that, in particular, $p(0)=\pm\tfrac{1}{2}\sqrt{3}$. We use $p(\kappa)$ to generate four sequences of numbers, $(p_r)_{r\geq 0}$, $(\beta_r)_{r\geq 0}$, $(s_r)_{r\geq 0}$ and $(\theta_r)_{r\geq 1}$, by
\begin{align*}
    &p(\kappa)=\sum_{r=0}^\infty p_r \kappa^r,
&\log\left[\frac{2p_0}{p_1\kappa}\frac{p(\kappa)-p_0}{p(\kappa)+p_0}\right]=\sum_{r=1}^\infty s_r \kappa^r,\\  
    &p(\kappa)^3=\sum_{r=0}^\infty \beta_r \kappa^r,
    &\log\left[\frac{2p_0}{p_1\kappa}\frac{p(\kappa)-\frac{1}{2}i}{p_0-\frac{1}{2}i}\right]=\sum_{r=1}^\infty \theta_r \kappa^r.
\end{align*}
We note that $(p_r)_{r\geq 0}$, $(\beta_r)_{r\geq 0}$ and $(s_r)_{r\geq 0}$ are real, whereas $(\theta_r)_{r\geq 1}$ also has non-real entries.

\begin{theorem}[Yang and Yang \cite{yangyang}]\label{thm:yangyang}
For integers $M,N\geq 0$,
\begin{equation*} 
u_{M,N}^{\pm}(x,t) = h_{M,N}^{\pm}(x,t) u_{\operatorname{bg}}(x,t),\qquad u_{\operatorname{bg}}(x,t)=e^{i[\alpha(x+6t)-\alpha^3 t]},
\end{equation*}
with $\alpha=\tfrac{1}{2}$, defines a solution to the Sasa-Satsuma equation, where $h_{M,N}$ equals
\begin{equation*}
    h_{M,N}^{\pm}=\frac{f_{M,N}^{\pm}}{g_{M,N}^{\pm}},
\end{equation*}
with $f_{M,N}^{\pm}$ and $g_{M,N}^{\pm}$ polynomials in $x,t$, given by the determinants of block matrices
\begin{equation*}
f_{M,N}^{\pm}=\begin{vmatrix}
\sigma_0^{1,1} & \sigma_0^{1,2}\\
\sigma_0^{2,1} & \sigma_0^{2,2}
\end{vmatrix},\qquad
g_{M,N}^{\pm}=\begin{vmatrix}
\sigma_1^{1,1} & \sigma_1^{1,2}\\
\sigma_1^{2,1} & \sigma_1^{2,2}
\end{vmatrix},
\end{equation*}
where, writing $(N_1,N_2)=(M,N)$, the individual block $\sigma_{k}^{I,J}$ is an $N_I\times N_J$ matrix,
\begin{equation*}
    \sigma_{k}^{I,J}=(\phi_{3i-I,3j-J}^{(k,I,J)})_{1\leq i\leq N_I,1\leq i\leq N_J},
\end{equation*}
for $k\in\{0,1\}$ and $I,J\in \{1,2\}$. Here, the entries of these matrices are defined by
\begin{equation*}
   \phi_{i,j}^{(k,I,J)}=\sum_{\nu=0}^{\min{(i,j)}} \left(\frac{p_1}{2p_0}\right)^{2\nu}S_{i-\nu}(\boldsymbol{\mathrm{x}}_I^+(k)+\nu\; \boldsymbol{\mathrm{s}})
   S_{j-\nu}(\boldsymbol{\mathrm{y}}_J^-(k)+\nu \;\boldsymbol{\mathrm{s}}),
\end{equation*}
where 
\begin{align*}
    \boldsymbol{\mathrm{s}}&=(s_1,s_2,s_3,\ldots),\\
\boldsymbol{\mathrm{x}}_I(k)&=(\boldsymbol{\mathrm{x}}_{1,I}(k),\boldsymbol{\mathrm{x}}_{2,I}(k),\boldsymbol{\mathrm{x}}_{3,I}(k),\ldots), \\
\boldsymbol{\mathrm{y}}_J(k)&=(\boldsymbol{\mathrm{y}}_{1,J}(k),\boldsymbol{\mathrm{y}}_{2,J}(k),\boldsymbol{\mathrm{y}}_{3,J}(k)\ldots),
\end{align*}
are infinite length vectors, with
\begin{align*}
\boldsymbol{\mathrm{x}}_{r,I}(k)&=p_r(x+6t)+\beta_r t+k\theta_r+a_{r,I}, \\
\boldsymbol{\mathrm{y}}_{r,J}(k)&=p_r(x+6t)+\beta_r t-k\overline{\theta}_r+a_{r,J}, 
\end{align*}
for $k\in\{0,1\}$, $I,J\in\{1,2\}$ and $r\geq 1$.
Here 
\begin{equation*}
    (a_{2,1},\ldots, a_{3M-1,1}),\quad (a_{1,2},\ldots, a_{3N-2,2}),
\end{equation*}
are free real constants that can be chosen at pleasure.
\end{theorem}
It is straightforward to check that the two solutions corresponding to the different choices of sign
are related by the following symmetry of the Sasa–Satsuma equation,
\begin{equation*}
    u\mapsto \hat{u},\quad \hat{u}(x,t)=\overline{u(-x,-t)}.
\end{equation*}
For explicit values of the first few entries in the sequences of numbers in Theorem \ref{thm:yangyang}, we refer to \cite{yangyang}.

\bibliographystyle{amsalpha}

\begin{thebibliography}{SKKT00}

\bibitem[AT16]{adalitanveer}
A. Adali and S. Tanveer,
\emph{Rigorous analytical approximation of tritronqu\'ee solution to Painlev\'e-I and the first singularity},
J. Differential Equations \textbf{261} (2016), no.~7, 3843--3863.



\bibitem[Ber12]{bertolaAS}
M. Bertola,
\emph{On the location of poles for the Ablowitz-Segur family of solutions to the second Painlev\'e equation},
Nonlinearity \textbf{25} (2012), no.~4, 1179--1185.



\bibitem[BM20]{bothnermiller}
T. Bothner and P. D. Miller,
\emph{Rational solutions of the Painlev\'e-III equation: large parameter asymptotics},
Constr. Approx. \textbf{51} (2020), no.~1, 123--224.

\bibitem[BMS18]{bothnermillersheng}
T. Bothner, P. D. Miller and Y. Sheng,
\emph{Rational solutions of the Painlev\'e-III equation},
Stud. Appl. Math. \textbf{141} (2018), no.~4, 626--679.

\bibitem[Bou13a]{boutrouxP1}
P. Boutroux,
\emph{Recherches sur les transcendantes de M. Painlev\'e et l'\'{e}tude asymptotique des \'equations diff\'erentielles du second ordre},
Ann. Sci. \'Ecole Norm. Sup. (3) \textbf{30} (1913), 255--375.

\bibitem[Bou13b]{boutrouxP12}
P. Boutroux,
\emph{Recherches sur les transcendantes de M. Painlev\'e et l'\'{e}tude asymptotique des \'equations diff\'erentielles du second ordre (suite)},
Ann. Sci. \'Ecole Norm. Sup. (3) \textbf{31} (1914), 99--159.

 \bibitem[Buc22]{buckinghamP4}
R. J. Buckingham, 
  \emph{Large-degree asymptotics of rational Painlevé-IV functions associated to generalized Hermite polynomials}, 
Int. Math. Res. Not. IMRN \text (2022), no.~18, 5534--5577. 

   \bibitem[BM14]{buckinghammillerP21}
R. J. Buckingham and P. D. Miller, 
  \emph{Large degree asymptotics of rational Painlev\'e-II functions: noncritical behaviour}, 
  Nonlinearity \textbf{27} (2014), no.~10, 2489--2578. 

   \bibitem[BM15]{buckinghammillerP22}
R. J. Buckingham and P. D. Miller, 
  \emph{Large degree asymptotics of rational Painlev\'e-II functions: critical behaviour}, 
  Nonlinearity \textbf{28} (2015), no.~6, 1539--1596. 

 \bibitem[BM22a]{buckmiller}
R. J. Buckingham and P. D. Miller, 
  \emph{Large-degree asymptotics of rational Painlevé-IV solutions by the isomonodromy method}, 
  Constr. Approx. \textbf{56} (2022), no.~2, 233--443. 
  
   \bibitem[BM22b]{buckinghammillerP3}
R. J. Buckingham and P. D. Miller, 
  \emph{On the algebraic solutions of the Painlev\'e-III $(D_7)$ equation}, 
  Phys. D \textbf{441} (2022), Paper No. 133493, 22pp. 

\bibitem[Chi16]{chiba}
H. Chiba, 
\emph{The first, second and fourth Painlev\'e equations on weighted projective spaces},
J. Differ. Equ. 260 (2016), no. 2, 1263--1313.

    \bibitem[Cla03a]{clarksonpolynomials}
P. A. Clarkson, 
  \emph{The fourth Painlev\'e equation and associated special polynomials}, 
  J. Math. Phys. \textbf{44} (2003), no.~11, 5350--5374. 

\bibitem[Cla03b]{clarksonP3}
P. A. Clarkson, 
\emph{The third Painlev\'e equation and associated special polynomials},
J. Phys. A \textbf{36} (2003), 9507--9532.

  
  \bibitem[Cla06b]{clarksonsurvey}
  P. A. Clarkson,
  \emph{Special polynomials associated with rational solutions of the Painlev\'e equations and applications to soliton equations},
  Comput. Methods Funct. Theory \textbf{6} (2006), 329--401.




\bibitem[CHT14]{CHT}
O. Costin, M. Huang and S. Tanveer,
\emph{Proof of the Dubrovin conjecture and analysis of the tritronquée solutions of $P_{I}$},
Duke Math. J. \textbf{163} (2014), no.~4, 665-704.

\bibitem[Dea23]{deanoP1}
A. Dea\~{n}o,
\emph{On the Riemann-Hilbert approach to asymptotics of tronqu\'ee solutions of Painlev\'e I},
J. Phys. A \textbf{56} (2023), no.~31, 314001.

\bibitem[EG17]{eremenkogabrielov}
A. Eremenko and A. Gabrielov,
\emph{Circular Pentagons and Real Solutions of Painlevé VI Equations},
 Comm. Math. Phys. \textbf{355} (2017), 51–95

\bibitem[EGH17]{EGH17}
A. Eremenko, A. Gabrielov and A. Hinkkanen,
\emph{Exceptional solutions of the Painlev\'e $\mathrm{VI}$ equation},
J. Math. Phys. \textbf{58} (2017), no. 1, 012701.











\bibitem[FW04]{forresterwittepvi}
P. J. Forrester and N. S. Witte 
\emph{Application of the {$\tau$}-function theory of {P}ainlev\'{e}
              equations to random matrices: {$\mathrm{P}_{\mathrm{VI}}$}, the
              {JUE}, {C}y{UE}, c{JUE} and scaled limits},
Nagoya Math. J. \textbf{174} (2004), 29--114.


 \bibitem[FOU00]{fukutani}
S. Fukutani, K. Okamoto and H. Umemura, 
 \emph{Special polynomials and the Hirota bilinear relations of the second and fourth Painlev\'e equations}, 
 Nagoya Math. J. \textbf{159} (2000), 179--200. 


 \bibitem[GG15]{paper65}
M. Garc\'{\i}a-Ferrero and D. G\'{o}mez-Ullate,
 \emph{Oscillation theorems for the {W}ronskian of an arbitrary
              sequence of eigenfunctions of {S}chr\"{o}dinger's equation}, 
 Lett. Math. Phys. \textbf{105} (2015), 551--573. 




 \bibitem[Gro87]{gromakfourth}
V. I. Gromak,
\emph{On the theory of the fourth Painlev\'e equation},
Diff. Eqns. \textbf{23} (1987), 506--513.


\bibitem[GLS02]{gromakbook}
V.I. Gromak, I. Laine and S. Shimomura,
\emph{Painlev\'e Differential Equations in the Complex Plane}, Walter de Gruyter \& Co., Berlin, 2002.


\bibitem[HMZ22]{HMZ22}
V. Hussin, I. Marquette and K. Zelaya,
\emph{Third-order ladder operators, generalized Okamoto and exceptional orthogonal polynomials},
J. Phys. A \textbf{55} (2022), no.~4, 045205.


\bibitem[IK98]{itskapaev}
A.R. Its and A.A. Kapaev, 
\emph{Connection formulae for the fourth Painlev\'e transcendent: Clarkson-McLeod solution}, 
J. Phys. A \textbf{31} (1998), 4073--4113.

\bibitem[IO16]{iwasakiokada} 
K. Iwasaki and S. Okada, 
\emph{On an orbifold Hamiltonian structure for the first Painlev\'e equation}, 
J. Math. Soc. Japan 68 (2016), no. 3, 961--974.

 \bibitem[JKM06]{JKM}
N. Joshi, K. Kajiwara and M. Mazzocco,
\emph{Generating function associated with the Hankel determinant formula for the solutions of the Painlev\'e IV equation},
Funkcial. Ekvac. \textbf{49} (2006), 451--468.

\bibitem[JK01]{joshikitaev}
N. Joshi and A. V. Kitaev,
\emph{On Boutroux's tritronqu\'ee solutions of the first Painlev\'e equation},
Stud. Appl. Math. \textbf{107} (2001), no.~3, 253--291.

 \bibitem[JR16]{joshiradnovicP4}
N. Joshi and M. Radnovi\'c,
\emph{Asymptotic behaviour of the fourth Painlev\'e transcendents in the space of initial values},
Constr. Approx. 44 (2016), no.~2, 195--231.

\bibitem[KM99]{kajiwaramasuda}
K. Kajiwara and T. Masuda,
\emph{On the Umemura polynomials for the Painlev\'e $\mathrm{III}$ equation},
Phys. Lett. A \textbf{260} (1999), 462--467.

\bibitem[KO98]{kajiwaraohta}
K. Kajiwara and Y. Ohta 
\emph{Determinant structure of the rational solutions for the Painlev\'e IV equation},
J. Phys. A \textbf{31} (1998), no.~10, 2431–-2446.


 \bibitem[Kap96]{kapaev1996}
 A. A. Kapaev, 
 \emph{Global asymptotics of the fourth Painlevé transcendent},
 Steklov Math. Inst. and IUPUI Preprint \#96-5, 1996.

 \bibitem[Kap98]{kapaev1998}
 A. A. Kapaev,
 \emph{Connection formulae for degenerated asymptotic solutions of the fourth Painleve equation},
 arXiv:solv-int/9805011, 1998.


 \bibitem[KT02]{kirillovtaneda}
A. N. Kirillov and M. Taneda,
\emph{Generalized Umemura polynomials},
Rocky Mt. J. Math. \textbf{32} (2002), no.~2, 691--702.
 
  \bibitem[KH87]{kodamahasegawa}
Y. Kodama and A. Hasegawa, 
  \emph{Nonlinear pulse propagation in a monomode dielectric guide}, 
  IEEE J. Quantum Electron. \textbf{23} (1987), 510--524. 


\bibitem[LT14]{lisovyyalgebraic}
  O. Lisovyy and Y. Tykhyy,
 \emph{Algebraic solutions of the sixth Painlev\'{e} equation},
  J. Geom. Phys. \textbf{85} (2014), 124--163.


  
  \bibitem[Luk67]{lukashevich}
  N. A. Lukashevich, 
  \emph{Theory of the fourth Painlev\'e equation}, 
  Diff. Eqns. \textbf{3} (1967), 395-399.
  
   \bibitem[Mas10a]{davidetritronquee1}
  D. Masoero,
 \emph{Poles of int\'egrale tritronqu\'ee and anharmonic oscillators. A WKB approach},
 J. Phys. A \textbf{43} (2010), no.~9, 095201.
 
   \bibitem[Mas10b]{davidetritronquee2}
  D. Masoero,
 \emph{Poles of int\'egrale tritronqu\'ee and anharmonic oscillators. Asymptotic localization from WKB analysis},
 Nonlinearity \textbf{23} (2010), no.~10, 2501--2507.
  

  
  \bibitem[MR18]{davidepieter}
D. Masoero and P. Roffelsen, 
  \emph{Poles of Painlev\'e IV rationals and their distribution}, 
  SIGMA Symmetry Integrability Geom. Methods Appl. \textbf{14} (2018), Paper No. 002, 49 pp. 

  \bibitem[MR21]{davidepieterhermite}
D. Masoero and P. Roffelsen, 
  \emph{Roots of generalised Hermite polynomials when both parameters are large}, 
  Nonlinearity \textbf{34} (2021), no.~3, 1663--1732. 
  
  \bibitem[MMT99]{takano2} 
T. Matano, A. Matumiya, and K. Takano, 
\emph{On some Hamiltonian structures of Painlev\'e systems. II}, 
J. Math. Soc. Japan 51 (1999), no. 4, 843--866.

\bibitem[Mat97]{takano3}  
A. Matumiya, 
\emph{On some Hamiltonian structures of Painlev\'e systems. III}, 
Kumamoto J. Math.  10 (1997), 45--73.

  \bibitem[MTC97]{optics1}
  D. Mihalache, N. Truta and L-C. Crasovan,
  \emph{Painlev\'e analysis and bright solitary waves of the higher-order nonlinear Schr\"odinger equation containing third-order dispersion and self-steepening term},
  Phys. Rev. E \textbf{56} (1997), 1064.

  \bibitem[MCB97]{MCB97}
  A. E. Milne, P. A. Clarkson and A. P. Bassom,
  \emph{Application of the isomonodromy deformation method to the fourth Painlev\'e equation},
  Inverse Problems \textbf{13} (1997), no.~2, 421--439.

  \bibitem[Mur85]{muratarationals}
  Y. Murata, 
  \emph{Rational solutions of the second and the fourth Painlev\'e equations}, Funkcial. Ekvac. \textbf{28} (1985), 1–32.

  \bibitem[NOOU98]{noou}
  M. Noumi, S. Okada, K. Okamoto and H. Umemura,
  \emph{Special polynomials associated with the Painlev\'e equations II},
  Proceedings of Taniguchi Symposium 1997: Integrable systems and algebraic geometry, 
  World Scientific, 1998.

\bibitem[NO97]{noumiokamoto}
M. Noumi and K. Okamoto,
\emph{Irreducibility of the second and the fourth Painlev\'e equations},
Funkcial. Ekvac. \textbf{40} (1997), no.~1, 139--163.

  \bibitem[NY98]{noumiyamadapolynomials}
M. Noumi and Y. Yamada, 
  \emph{Umemura polynomials for the Painlev\'e $\mathrm{V}$ equation}, 
  Phys. Lett. A \textbf{247} (1998), 65--69. 
  
  \bibitem[NY99]{noumiyamada}
M. Noumi and Y. Yamada, 
  \emph{Symmetries in the fourth Painlev\'e equation and Okamoto polynomials}, 
  Nagoya. Math. J. \textbf{153} (1999), 53--86. 

  










  \bibitem[OKSO06]{studies5}
Y. Ohyama, H. Kawamuko, H. Sakai and K. Okamoto, 
  \emph{Studies on the Painlev\'e equations V. Third Painlev\'e equations of special type $\pain{III}(D_7)$ and $\pain{III}(D_8)$}, 
  J. Math. Sci. Univ. Tokyo \textbf{13} (2006), 145--204. 
  
\bibitem[Oka79]{OKAMOTO1979} K. Okamoto, 
\emph{Sur les feuilletages associ\'es aux \'equations du second ordre \`a points critiques fixes de P. Painlev\'e}, 
Japan. J. Math. (N.S.) \textbf{5} (1979), no. 1, 1--79. 
  
  \bibitem[Oka80]{okamotohamiltonians}
K. Okamoto, 
\emph{Polynomial {H}amiltonians associated with {P}ainlev\'{e}
  equations. {I}},
  Proc. Japan Acad. Ser. A Math. Sci. \textbf{56} (1980),
  no.~6, 264--268.




  
  \bibitem[Oka86]{okamotostudies}
K. Okamoto, 
  \emph{Studies on the Painlev\'e equations. III. Second and fourth Painlev\'e equations $\pain{II}$ and $\pain{IV}$}, 
  Math. Ann. \textbf{275} (1986), 221--255. 

  \bibitem[Oka87]{okamotopvi}
K. Okamoto, 
\emph{Studies on the {P}ainlev\'{e} equations. {I}. {S}ixth
              {P}ainlev\'{e} equation {$P_{{\rm VI}}$}},
  Ann. Mat. Pura Appl. \textbf{146} (1987),
  337--381.


\bibitem[Rof11]{pieterYV}
P. Roffelsen,
\emph{On the number of real roots of the Yablonskii-Vorob'ev polynomials}
  SIGMA Symmetry Integrability Geom. Methods Appl. \textbf{8} (2011), Paper No. 099, 9 pp. 

\bibitem[ST04]{STnodalcurves}
M.-H. Saito and H. Terajima,
\emph{Nodal curves and Riccati solutions of Painlev\'e equations},
J. Math. Kyoto Univ. \textbf{44-3} (2004), 529--568.

\bibitem[Sak01]{SAKAI2001} 
H. Sakai, 
\emph{Rational surfaces associated with affine root systems and geometry of the Painlev\'e equations}, 
Comm. Math. Phys. \textbf{220} (2001), no. 1, 165--229.


  \bibitem[SS91]{sasasatsuma}
N. Sasa and J. Satsuma, 
  \emph{New type of soliton solutions for a higher-order nonlinear Schr\"odinger equation}, 
  J. Phys. Soc. Jpn. \textbf{60} (1991), 409--417. 
  
 \bibitem[ST19]{twiton1}
 J. Schiff and M. Twiton,
 \emph{A dynamical systems approach to the fourth Painlev\'e equation},
 J. Phys. A \textbf{52} (2019), no.~14, 145201, 16 pp.
 

  \bibitem[ST24]{twiton2}
 J. Schiff and M. Twiton,
 \emph{Classification of real solutions of the fourth Painlev\'e equation},
 Mathematics \textbf{12} (2024), no.~3, 463.

\bibitem[ST97]{takano1} 
T. Shioda and  K. Takano, 
\emph{On some Hamiltonian structures of Painlev\'e systems. I}, 
Funkcial. Ekvac. 40 (1997), no. 2, 271--291.

\bibitem[SRKJ07]{optics2}
D. Solli, C. Ropers, P. Koonath and B. Jalali,
\emph{Optical rogue waves},
Nature \textbf{450} (2007), 1054--1057.

\bibitem[Sun21]{optics3}
F. Sun,
\emph{Optical solutions of Sasa-Satsuma equation in optical fibres},
Optik \textbf{228} (2021), 166127.


\bibitem[Ume90]{umemura}
H. Umemura,
\emph{Birational automorphism groups and differential equations},
Nagoya Math. J. \textbf{119} (1990), 1--80.

\bibitem[UW98]{umemurawatanabe}
H. Umemura and H. Watanabe,
\emph{Solutions of the second and fourth Painlev\'e equations},
Nagoya Math. J. \textbf{151} (1998), 58--59.

\bibitem[Ume20]{umemurapolynomials}
H. Umemura,
\emph{Special polynomials associated with the Painlev\'e equations I},
Ann. Fac. Sci. Toulouse Math. \textbf{29} (2020), no.~5, 1063--1089.


\bibitem[Ver97]{vereshch}
V. L. Vereshchagin,
\emph{Global asymptotic formulae for the fourth Painleve transcendent},
Sb. Math. \textbf{188} (1997), no.~12, 1739.

\bibitem[Vor65]{vorobev}
A. P. Vorob'ev,
\emph{On rational solutions of the second Painlev\'e equation},
Diff. Eqns. \textbf{1} (1965), 58--59.

\bibitem[Yab59]{yablonskii}
A. I. Yablonskii,
\emph{On rational solutions of the second Painlev\'e equation},
Vesti Akad. Navuk. BSSR Ser. Fiz. Tkh. Nauk. \textbf{3} (1959), 30--35.

\bibitem[YY23]{yangyang}
B. Yang and J. Yang, 
  \emph{Partial-rogue waves that come from nowhere but leave with a trace in the Sasa-Satsuma equation}, 
  Phys. Lett. A \textbf{458} (2023), 128573. 

\bibitem[YBAETZKAB20]{optics4}
Y. Y{\i}ld{\i}r{\i}m, A. Biswas, M. Asma, M. Ekici, H. Triki, E. M. E. Zayed, A. Kamis Alzahrani, M. R. Belic,
\emph{Optical solitons with Sasa-Satsuma equation}.
Optik \textbf{219} (2020), 165183.

\end{thebibliography}

\end{document}